\documentclass[12pt]{article}
\usepackage[top=1in, bottom=1in, left=1in, right=1in]{geometry}

\usepackage{graphics}
\usepackage{float}
\usepackage{theorem}
\usepackage{latexsym}
\usepackage{amsfonts}
\usepackage{amssymb}
\usepackage{amsmath}
\usepackage{graphicx}
\usepackage{natbib}
\usepackage{lscape}
\usepackage{booktabs}
\usepackage{threeparttable}
\usepackage{appendix}
\usepackage{subfig}
\usepackage[colorlinks=true,linkcolor=black,citecolor=black,urlcolor=black,pdfstartview=FitH,bookmarks=false]{hyperref}
\usepackage{setspace,enumerate}
\usepackage{titlesec}
\usepackage{color}
\usepackage{natbib}
\usepackage{bm}
\usepackage{multicol}
\usepackage{soul}
\usepackage{tabularx} 

\usepackage[utf8]{inputenc}

\allowdisplaybreaks

\newtheorem{thm}{Theorem}
\newtheorem{assumption}{Assumption}

\newtheorem{lemma}{Lemma}

\newtheorem{remark}{Remark}

\newtheorem{remarkapp}{Remark}

\numberwithin{thm}{section}
\numberwithin{lemma}{section}
\numberwithin{assumption}{section}
\numberwithin{lemmaapp}{section}
\numberwithin{assumptionapp}{section}
\numberwithin{proposition}{section}
\numberwithin{corollary}{section}
\numberwithin{example}{section}
\numberwithin{remark}{section}
\numberwithin{remarkapp}{section}
\numberwithin{thmapp}{section}

\numberwithin{algo}{section}

\numberwithin{equation}{section}
\numberwithin{thm}{section}
\numberwithin{lemma}{section}
\numberwithin{assumption}{section}

\usepackage{authblk}

\begin{document}
\thispagestyle{empty}
\vspace*{5mm}
\thispagestyle{empty}
\vspace*{5mm}
\begin{center}

{\Large \bf  Testing for Monotone Equilibrium Strategies\\ in Games of Incomplete Information}\footnote{This paper supersedes an earlier version that is entitled ``Testing for Monotone Equilibrium Strategies in First Price Auctions.'' We thank Phil Haile, Bradley Larsen, Matt O'Keefe and Andres Santos for helpful comments on the early versions of the paper.   Yu-Chin Hsu gratefully acknowledges the research support from the
National Science and Technology Council of Taiwan (NSTC113-2628-H-001-001 and NSTC114-2410-H-001-080-MY3) and the Academia Sinica
Investigator Award of the Academia Sinica, Taiwan (AS-IA-110-H01).  Chu-An Liu gratefully acknowledges the research support by the Academia Sinica Career Development Award
(AS-CDA-110-H02).  Yu-Chin Hsu and Chu-An Liu gratefully acknowledge the research support by the Center for Research in Econometric Theory
and Applications, Taiwan (Grant No. 113L8601). Tong Li and Hidenori Takahashi gratefully acknowledge financial support from the Joint Usage/Research Center, Institute of Economic Research, Kyoto University. Hidenori Takahashi gratefully acknowledges financial support from KAKENHI (Grant No. 24K00252).}

\vspace*{5mm}

\textbf{\large Yu-Chin Hsu$^{\dag}$}{\large{} }

\vspace*{1mm}
Institute of Economics, Academia Sinica\\
Department of Finance, National Central University\\
Department of Economics, National Chengchi University and\\
CRETA, National Taiwan University

\vspace*{3mm} \textbf{\large Tong Li$^{*}$}

\vspace*{1mm} Department of Economics, Vanderbilt University

\vspace*{3mm} \textbf{\large Chu-An Liu$^{\ddag}$}

\vspace*{1mm} Institute of Economics, Academia Sinica

\vspace*{3mm}

\textbf{\large Hidenori Takahashi$^{\star}$}

\vspace*{1mm}
Kyoto Institute of Economic Research, Kyoto University

\vspace*{5mm}
This version: \today
\end{center}

\vfill{} 
\noindent $^{\dag}${\fontsize{8.5pt}{9pt}\selectfont ychsu@econ.sinica.edu.tw}, $^{*}${\fontsize{8.5pt}{9pt}\selectfont tong.li@vanderbilt.edu}, $^{\ddag}${\fontsize{8.5pt}{9pt}\selectfont caliu@econ.sinica.edu.tw},
$^{\star}${\fontsize{8.5pt}{9pt}\selectfont  takahashi.hidenori@kier.kyoto-u.ac.jp}

\newpage
\vspace*{30mm}
\begin{abstract}


This paper develops a specification test for the monotonicity of Bayesian Nash equilibrium strategies in unobserved types in games of incomplete information. We first establish a rationalization result for general Bayesian games under the symmetric independent private type paradigm. We show that monotonicity of the equilibrium strategy in private types is equivalent to monotonicity of its quasi-inverse in observed actions, and that the quasi-inverse is identified from observed actions. This equivalence allows us to transform the problem of testing monotonicity in unobserved types into a testable restriction on identified reduced-form objects and to reformulate this restriction as a countable collection of moment inequalities. A Cram\'{e}r--von Mises test with bootstrap critical values based on generalized moment selection is proposed. Our test accommodates observed and unobserved game heterogeneity and performs well in simulations and an application to procurement auctions.


\end{abstract}

\thispagestyle{empty} \newpage \setcounter{page}{1}

\section{Introduction}
In games of incomplete information, or Bayesian games, players' strategies are functions mapping their private types such as signals, valuations, or costs, to actions. A Bayesian Nash equilibrium (BNE) is a profile of such strategy functions where each player's strategy is a best response to others, given common knowledge of opponent type distributions. Monotonicity of equilibrium strategies, which means that a player's action is monotone with their type, is one of the most empirically informative predictions in games of incomplete information: higher valuations, lower costs, or stronger signals lead to more aggressive actions. This “sorting” property underlies how economists interpret actions as revealing private information in a wide range of environments, from auctions to contests, public-good provision, and quantity competition with private information. It is also a cornerstone of the theory: monotonicity is typically the condition that delivers well-behaved equilibria (existence and uniqueness), intuitive comparative statics, and tractable mechanism design. 
On the theoretical ground, \cite{Athey2001} establishes seminal results on the existence of a monotone pure strategy BNE under the Spence-Mirrlees single-crossing condition, which are further extended to accommodate multi-dimensional type and action spaces by \cite{McAdams2003},  and to allow for more general spaces by \cite{Reny2011}. From an econometric perspective, monotone equilibrium strategies are equally central because they justify the inversion logic that maps observed actions back into latent types, enabling identification and inference in structural models.
See, e.g., \cite{GPV2000}, \cite{LPV2002} and a large related literature in first price auctions, \cite{Kline2025} in allocation-transfer games, \cite{LiZhangZhao2025} and \cite{LiZhangZhao2026} for Bayesian games with continuous payoffs, and \cite{AryalZincenko2024} for Cournot competition with private information, to name only a few.

While there has been significant progress in identification and estimation of Bayesian games, there is remarkably little work on testing the monotonicity of BNE strategies, despite its role as a key equilibrium property and as a cornerstone in identification and estimation. This gap is consequential because monotonicity is routinely derived as a model implication and an equilibrium restriction, as well as a prerequisite for inversion-based identification and estimation. If monotonicity fails in the data, this means that the structural model and its implied counterfactuals may be invalid. In other words, rejecting monotonicity in the data means that the data cannot have been generated by equilibrium play in the assumed Bayesian game. Consider first price auctions within the symmetric independent private value paradigm as an example. It has been well established, e.g., \cite{MaskinRiley2003}, that the symmetric strictly increasing bidding equilibrium is the unique BNE. As such, monotonicity serves as a testable restriction, and testing for it constitutes a model specification test: rejecting monotonicity in a bidding data set means that at least one of the model assumptions could fail: either independence, or symmetry, or private values, or equilibrium play, or a combination of these. For instance, suppose symmetry and independent private values are maintained as true, rejecting monotonicity in the data means that bidders in the data do not play the BNE. Then possible explanations could include bounded rationality or bidding mistakes, collusion or bid coordination, or bidders best-responding to biased beliefs rather than the true value distribution.\footnote{For identification of first price auctions with entry and biased beliefs, see \cite{LiZhu2026}.} Relatedly, monotonicity has independent empirical content in contexts where departures from competitive behavior are of interest. For example, to detect collusive behavior of bidders, \citet{BajariYe2003} derive a set of sufficient and necessary conditions for the symmetric and competitive bidding equilibrium including monotonicity of the bidding equilibrium, thus rejection of monotonicity can constitute a joint test of all assumptions. However, they bypass testing for monotonicity, and only test for other conditions such as exchangeability, largely because tools for testing monotone equilibrium strategies have been absent. 

Testing monotone equilibrium strategies in types is challenging because types are unobserved to the econometrician. To date, the only development is in the auction setting considered in \cite{LV2021}, who exploit the equivalence between monotonicity of BNE bidding strategies in private values and monotonicity of the quasi-inverse of the bidding strategy, which is identified from observed bids.\footnote{In a somewhat related yet different vein, noting that the nonparametric estimator of the quasi-inverse bidding strategy in \cite{GPV2000} does not impose the monotonicity condition, \cite{HLMPP2012} and \cite{MMSX2021} develop monotonicity constrained nonparametric estimators of the quasi-inverse strategies. Also, \cite{PinkseSchurter2019} consider estimation of first-price auctions with shape restrictions.}  They reformulate the test as one for the concavity of the integrated quasi-inverse strategy to avoid nonparametric estimation of bid densities in the quasi-inverse formulation.  They 
then construct a statistic based on the difference between this function and its least concave majorant (LCM).

This paper aims to fill the gap in the literature by developing a new and easy-to-implement test for monotone equilibrium strategies in general Bayesian games. To this end, we first characterize the sufficient and necessary conditions for the
existence of a private type distribution that can rationalize the observed action distribution in general Bayesian games considered in \cite{Athey2001} within the symmetric independent private type paradigm. As such, we extend the rationalization logic in the empirical auction literature to games of incomplete information, which include not only auctions that have discontinuous payoffs, but also other games with continuous payoffs such as contests, public-good provision, and Cournot competition with private information, among others.\footnote{The rationalization result in the empirical auction literature was first established in \cite{GPV2000} for the symmetric independent private value setting, and has been extended to the symmetric affiliated private value setting in \cite{LPV2002}, and the asymmetric affiliated private value setting in \cite{CPV2003}, as well as for risk-averse bidders in \cite{GPV2009}.}

Our rationalization result delivers an equivalence that forms the basis of our test: within the maintained framework, monotonicity of the equilibrium strategy in private types holds if and only if the quasi-inverse equilibrium strategy is monotone in actions. The quasi-inverse is a ratio of the derivatives of two reduced-form objects—the expected allocation rule and the expected payment rule—that are functions of observed actions and can therefore be identified from the distribution of actions given the known institutional mapping from actions to allocations and transfers. As a result, the testing problem becomes to test whether this identified quasi-inverse map is monotone in the action.

Although this insight is in a similar spirit to 
\cite{LV2021}, our testing approach differs fundamentally. First, we move beyond first-price auctions and develop a single and unified framework that covers a broad class of Bayesian games with continuous actions, including contests, public-good provision, and Cournot competition with private information, among others. Second, rather than relying on an LCM-based concavity test, we reformulate the monotonicity null into a countable collection of moment inequalities that involve only unconditional expectations and can be estimated by simple sample analogs, following the approach of \cite{HLS2019}, \cite{HsuShen2021} and \cite{HHLL2025}, who reformulate the null hypothesis of monotonicity in regressions, conditional treatment effects under regression discontinuity designs, and mean potential outcomes in continuous treatment effect models, respectively, to moment inequalities. Third, this moment-inequality formulation naturally extends to heterogeneous games, allowing us to control for game-specific covariates and to develop procedures that allow for both observed and unobserved game heterogeneity—features that are essential in empirical work and difficult to accommodate with LCM-type approaches.

Building on this moment-inequality representation, we construct a Cram\'{e}r–von Mises (CvM) type statistic based on the estimated inequalities.  
We combine a nonparametric bootstrap and the generalized moment selection method of \cite{AndrewsShi2013} to construct the critical values.  We show that our test controls size asymptotically and is consistent against any fixed alternative. We also provide a practical guide by specializing the general procedure to four widely used models: first-price auctions, Tullock contests, public-good provision, and Cournot competition with private information. 
An important advantage of our moment inequalities-based test for monotonicity  is it naturally accommodates heterogeneous games. Controlling for game-specific covariates is essential in almost all empirical applications. Moreover, our approach can accommodate unobserved game heterogeneity, noting that it has been an important and yet challenging issue in the structural auction literature to develop methods in identification and inference with unobserved auction heterogeneity.\footnote{While it is important to control for unobserved game heterogeneity, little attention has been paid to controlling for unobserved heterogeneity in Bayesian games other than auctions. For identification and estimation of auction models with unobserved heterogeneity, see \cite{LiZheng2009}, \cite{Krasnokutskaya2011},  \cite{Roberts2013}, \cite{CHS2020}, and \cite{LSX2024}, \cite{Asker2010}, \cite{LarsenZhang2026b} (auctions and bargaining), among others, and \cite{HaileKitamura2019}  for a survey.} Specifically, we develop fully nonparametric tests that can control for observed heterogeneity and also provide tests based on implications that remain valid under unobserved heterogeneity. A nice feature of our nonparametric test is that each moment condition to be jointly tested can be estimated at the parametric rate even with controlling for covariates. On the other hand, our fully nonparametric procedures may suffer from the curse of dimensionality problem due to the increasing number of moments that need to be jointly tested as the number of continuous covariates grows. Therefore, in Appendix C, we additionally propose semiparametric tests based on action homogenization, extending the bid homogenization method proposed in \cite{HHS2003} beyond auctions to general Bayesian games. As such, our paper provides a monotonicity testing framework for incomplete-information games that can accommodate both observed and unobserved heterogeneity in a way suitable for applied work.

The Monte Carlo evidence indicates that the test performs well in finite samples, including sample sizes typical in auction applications. We apply our test to the asphalt paving procurement datasets studied in \cite{Aaltio_etal2025}, where cartel detection is a central concern. The empirical findings are consistent with the evidence reported in \cite{Aaltio_etal2025}: we do not reject monotonicity in the dataset where no collusion was found, while in datasets where collusion was detected, our results provide evidence against monotonicity. To our knowledge, this is the first empirical application to formally test monotonicity of equilibrium strategies in a BNE framework using real data.

Our work contributes to the econometrics of games of incomplete information. Our focus is on Bayesian games with continuous actions, which have witnessed vast development in auctions, and the recent emerging literature in other games. See \cite{LarsenZhang2026a}, \cite{LarsenZhang2026b}, \cite{ALS2023}, \cite{Kline2025},  \cite{LiZhangZhao2025}, and \cite{LiZhangZhao2026} for identification of general Bayesian games with different payoff/information structures, \cite{HeHuang2021} for estimation of a contest model, \cite{AryalGabrielli2020} for estimation of a competitive nonlinear pricing model, and \cite{AryalZincenko2024} for estimation of a Cournot competition model with private information.\footnote{Notably, there has been significant development in the econometrics of games of incomplete information with discrete actions, e.g., \cite{Seim2006}, \cite{Aradillas-Lopze2010},  \cite{Tang2010}, \cite{dePaulaTang2012}, \cite{WanXu2014}, \cite{AguirregabiriaMira2019}, among others.} 

Our work also contributes to the econometric literature on testing shape restrictions, and monotonicity in particular.\footnote{See, e.g., \cite{CSS2018} for a review on the econometrics of shape restrictions.} The recent projection-based approach of \cite{FS2021}, for example, provides a general framework for testing shape restrictions defined by convex cones. As \cite{LV2021} emphasize, most existing tests focus on monotonicity in observed variables, whereas both their work and ours address monotonicity in latent variables.
While \cite{LV2021} rely on the LCM in the spirit of \cite{FS2021},  our test employs a transformation that leads to moment inequalities akin to \cite{HLS2019}, \cite{HsuShen2021} and  \cite{HHLL2025}; the advantages of our test highlight a flexible route to testing shape restrictions in structural models and demonstrate the practical value of such methods in applications. As such, we also contribute to the literature on using moment inequalities to test monotonicity, 
which has been focusing on regression monotonicity (\cite{Chetverikov2019}, \cite{HLS2019}), treatment effect monotonicity (\cite{HsuShen2021}, \cite{HHLL2025}), stochastic monotonicity (\cite{LLW2009}, \cite{HLS2019}), among others, and our paper is the first one to use moment inequalities to test monotonicity restrictions in game-theoretic models.\footnote{In the local treatment effect literature, \cite{Kitagawa2015} and \cite{MW2017} derive testable implications of IV validity, including IV monotonicity, in the form of moment inequalities. \cite{Kitagawa2015} employs the approach of \cite{AndrewsShi2013}, whereas \cite{MW2017} use the intersection bounds approach of \cite{CLR2013}.  
\cite{AHKMW2022} extend this analysis to the fuzzy regression discontinuity design. \cite{DelgadoEscanciano2012} and \cite{Seo2018} study stochastic monotonicity by reformulating the problem as one of testing concavity and relying on the LCM approach. 
In addition, moment inequality methods have been widely used in the partial identification literature; see, for example, the surveys by \cite{Tamer2010}, \cite{CanayShaikh2017} and \cite{Molinari2020}.}

This paper is organized as follows. Section 2 characterizes differentiable symmetric equilibrium strategies within the symmetric independent private-type paradigm and establishes the rationalization result linking monotonicity in types to monotonicity of the quasi-inverse map in actions. It then develops our benchmark test for homogeneous games and derives its asymptotic properties. Section 3 illustrates implementation of our benchmark test in auctions, contests, public-good provision, and Cournot competition with private information. Section 4 considers the extensions of our tests in Section 2 to allow for heterogeneous games, controlling for observed and unobserved game heterogeneity. Section 5 reports Monte Carlo evidence. Section 6 presents the asphalt paving procurement auction application. Section 7 concludes. 
Extensions such as the semiparametric test using action homogenization and all the technical proofs are collected in the Appendix.  

\section{Games of Incomplete Information}
\label{sec: incomplete information}
In this section, we first characterize the sufficient and necessary conditions for the
existence of a private type distribution that can rationalize the observed action distribution in general Bayesian games considered in \cite{Athey2001} within the symmetric independent private type paradigm, in which all agents play a symmetric and monotone strategy.  The framework here covers a large class of Bayesian games, including the trading games with auction games as special cases in \cite{LarsenZhang2026a} and \cite{LarsenZhang2026b}, the allocation-transfer games considered in \cite{Kline2025}, which include auctions, contests, and public-good provision, and general Bayesian games with continuous payoffs in \cite{LiZhangZhao2025}, which include contests, public-good provision, and Cournot competition with incomplete information. 

\subsection{Strictly Monotone Bayesian Nash Equilibrium and Rationalization Result}
\label{sec: necessary and sufficient conditions}

Suppose that there are $N \ge 2$ players (agents) in a game of incomplete information. We begin with an arbitrary observed joint distribution $\boldsymbol{G}(\cdot)$ of the action vector
$\bm B=(B_1,\ldots,B_N)$, with $B_i$, $i=1,2, ..., N$, representing player $i$'s action in 
$\mathcal{B}=[\underline b,\overline b]$, where
$0\leq \underline b<\overline b<\infty$. 
Let $\bm B_{-i}$ be the vector of $\bm B$ without its $i$-th element.  The allocation and net-payment rules are known
functions of the action profile. In particular, for agent $i$, let
$D_i\equiv D(B_i,\bm B_{-i})$ denote the allocation and let
$C_i\equiv C(B_i,\bm B_{-i})$ denote the net payment made or cost incurred. The outside
option is normalized to zero, agents are risk neutral, and the payoff of a type-$v_i$
agent is linear in type, $v_iD_i-C_i$.

Consider the symmetric, independent private type case, where all $N$ agents' actions are independent and have the same marginal distribution $G(\cdot)$ on $[\underline b,\overline b]$ for their own actions. Given the known
rules $D$ and $C$, define the reduced-form expected allocation and transfer from
choosing action $b$ when opponents' actions are independently distributed according
to $G(\cdot)$ as
\begin{align}
P(b;G)&=E_{ \bm B_{-i}}[D(b, \bm B_{-i})], \qquad
T(b;G)=E_{\bm B_{-i}}[C(b, \bm B_{-i})],
\label{eq: P and T}
\end{align}
Thus, $P(\cdot;G)$ and $T(\cdot;G)$ are determined by the game rules and the action distribution $G(\cdot)$. We now impose the following regularity condition on $G(\cdot)$.
\begin{assumption}\label{assu: PT}
$G(\cdot)$ is continuously differentiable on $[\underline b,\overline b]$ and has a
density $g(\cdot)>0$ on $(\underline b,\overline b)$. The functions $P(\cdot;G)$ and
$T(\cdot;G)$ are twice continuously differentiable on
$(\underline b,\overline b)$, and $P'(b;G)>0$ throughout this interval. Moreover,
\[
\xi(b;G):=\frac{T'(b;G)}{P'(b;G)}, \qquad b\in(\underline b,\overline b),
\]
extends continuously to $[\underline b,\overline b]$ with $\underline\xi:=\xi(\underline b;G)$ and
$\overline\xi:=\xi(\overline b;G)$.
\end{assumption}

The increasing-differences condition ($P'(b;G)>0$) implies that any interior best response is strictly increasing in type. To see this, suppose otherwise: let \(v_2>v_1\) and assume that both types optimally choose the same interior action \(b^*\). Their first-order conditions would imply
\[
0
=
(v_2-v_1)P'(b^*;G)>0,
\]
a contradiction. Thus, any best-response selection is strictly increasing over types whose chosen actions are interior \citep{EdlinShannon1998}. Similar smoothness conditions are used by \cite{LarsenZhang2026a} and \cite{LarsenZhang2026b}  for point identification in trading games and by \cite{Kline2025} for point identification in allocation-transfer games.

We now define the economic model to be rationalized. The observed joint distribution $\boldsymbol{G}(\cdot)$
of $\bm B$ is \emph{rationalizable by a strictly monotone symmetric equilibrium} if
there exists a compact type interval $\mathcal V=[\underline v,\overline v]$, a CDF
$F(\cdot)$ on $\mathcal V$, i.i.d.\ private types $V_i\sim F(\cdot)$, and a continuous, strictly
increasing symmetric strategy $s:\mathcal V\to\mathcal B$ such that: (i) the actions
$B_i=s(V_i)$ induce the observed joint distribution of $\bm B$; and (ii), for every
$v\in(\underline v,\overline v)$, $s(v)$ is an interior best response when opponents'
actions are i.i.d.\ with CDF $G(\cdot)$, namely,
\[
s(v)\in\arg\max_{b\in\mathcal B} U(b,v;G),
\qquad U(b,v;G):=vP(b;G)-T(b;G).
\]
If such a rationalization exists, the first-order condition at the equilibrium action $b=s(v)$ is
\[
0=U_b(s(v),v;G)=vP'(s(v);G)-T'(s(v);G).
\]
Given $P'(b;G)\not=0$, this condition implies
\begin{align}
v=\frac{T'(s(v);G)}{P'(s(v);G)}=\xi(s(v);G).
\label{eq:implicit}
\end{align}
Consequently, under a strictly increasing rationalizing equilibrium,
$\xi(b;G)=s^{-1}(b)$. Equation (\ref{eq:implicit}) corresponds to Corollary 2 in
\cite{LarsenZhang2026a} for trading games and equation (26) in \cite{Kline2025} for
allocation-transfer games. We therefore refer to $\xi(\cdot;G)$ as the quasi-inverse
equilibrium strategy.

\begin{thm}[Rationalizability and Quasi-Inverse Equilibrium Strategy]
\label{thm: LZ2018}
Fix an arbitrary observed joint distribution of $\bm B$ and a candidate CDF $G$
satisfying Assumption \ref{assu: PT}. The observed joint action distribution is
rationalizable by a strictly monotone symmetric equilibrium in which the private type $V$ follows CDF $F(\cdot)$  if and only if the
following conditions hold:
\begin{itemize}
\item[] (C1) Actions $\{B_i:~i=1,\ldots,N\}$ are i.i.d.\ with common CDF $G(\cdot)$.

\item[] (C2) $\xi(\cdot;G)$ is strictly increasing on
$(\underline b,\overline b)$.
\end{itemize}
In addition, the type distribution $F(\cdot)$ is unique for any strictly monotone interior
symmetric equilibrium that rationalizes the observed joint action distribution.
\end{thm}

Theorem \ref{thm: LZ2018} extends Theorem 1 of \cite{GPV2000} from first-price
auctions under symmetric independent private values to the general Bayesian games
considered by \cite{Athey2001}. Condition (C1) is a restriction on the observed joint
distribution of actions, whereas condition (C2) is the monotonicity restriction on the
quasi-inverse map constructed from the candidate marginal $G$. Together, they are
necessary and sufficient for monotone-equilibrium rationalizability. The theorem also
establishes nonparametric point identification of the rationalizing type distribution:
if $B\sim G$, then $V=\xi(B;G)$ and hence
\[
F(v)=\Pr\{\xi(B;G)\leq v\}=G(\xi^{-1}(v;G)),
\qquad v\in[\underline\xi,\overline\xi].
\]
Our result extends the point-identification result in \cite{LiZhangZhao2026} for
Bayesian games with continuous payoffs. Under a different set of assumptions,
\cite{Kline2025} establishes a related point-identification result for
allocation-transfer games.

\subsection{Proposed Test}\label{sec: proposed test}
We restrict our attention to the symmetric independent private value paradigm, in which (C1) in Theorem \ref{thm: LZ2018} is satisfied.\footnote{With the vast literature on nonparametric tests for independence or/and symmetry, (C1) can be tested.} Suppose that with a data set where the information on all the agents' actions is observed across the games, the econometrician wants to test if actions can be treated as a strictly increasing BNE, or in words, their joint distribution can be rationalized by a strictly monotone symmetric equilibrium in which the private type $V$ follows the CDF $F(\cdot)$, where $V$ and $F(\cdot)$ are not observed by the econometrician. Under the null, actions are strictly monotone equilibrium plays, from Theorem \ref{thm: LZ2018}, it is equivalent to that the quasi-inverse equilibrium $\xi(\cdot;G)$ is strictly monotone in actions. Therefore, we formulate the null hypothesis as
\begin{align}\label{eq: general H0}
			H_0:~ \xi(b)=\frac{T'(b)}{P'(b)}~\text{is weakly increasing in $b$.}  
\end{align}
In (\ref{eq: general H0}), we suppress the dependence of $\xi$, $T$ and $P$ functions on $G(\cdot)$ for notational simplicity. 
Ideally, the null should be that $\xi(b)$ is strictly increasing, which usually results in a degenerate case for the asymptotic distribution of the test statistics. Therefore, we include the equality in the null hypothesis, as is usually done in the literature, e.g., \cite{LV2021}.\footnote{Another reason to consider weakly increasing in the null hypothesis is that any test of the null of strictly increasing is likely to have power equal to size against alternatives that are weakly (but not strictly) increasing. This is related to the problems in \cite{Romano2004}. We thank Andres Santos for pointing this out.}
We propose a test for the null hypothesis in (\ref{eq: general H0}) and our test is similar to those in \cite{HLS2019}, \cite{HsuShen2021} and \cite{HHLL2025}.  We first transform the null hypothesis in (\ref{eq: general H0}) to countably many inequalities without loss of information and we summarize this in the following lemma.  Define $a=\overline{b}-\underline{b}$, which is the length of $[\underline{b},\overline{b}]$.

\begin{lemma}\label{lemma: general H0 transformation} 
Let $h(b)$ be a known function such that $0 <h(b)P'(b) \leq M<\infty$.  
Then $H_0$ in (\ref{eq: general H0}) is equivalent to
\begin{align}
&H_0': \nu(b_1,b_2, q)=M(b_2,q)W(b_1,q)-M(b_1,q)W(b_2,q)\leq 0~\text{for all $(b_1,b_2, q)\in\mathcal{L}$} \label{eq: general H0'},
\end{align}
where 
\begin{align}
&M(b,q)= \int_{b}^{b+\frac{a}{q}} h(\tilde{b})T'(\tilde{b}) d\tilde{b}, 
~~~
W(b,q)=\int_{b}^{b+\frac{a}{q}} h(\tilde{b})P'(\tilde{b}) d\tilde{b},~\text{and}\label{eq: general M W}\\
&\mathcal{L}=\Big\{(b_1,b_2, q): \Big(\frac{b_1-\underline{b}}{a}, \frac{b_2-\underline{b}}{a}\Big)\cdot q\in(0,1,\ldots,q-1)^2, b_1>b_2~\text{for q=2,3,\ldots}  \Big\}.\label{eq: L}
\end{align}
\end{lemma}
Lemma \ref{lemma: general H0 transformation} shows that the transformation of the null hypothesis $H_0$ to $H_0'$ in (\ref{eq: general H0'}) which includes a countable many moment inequalities is without loss of information.  The $h(b)$ function can be any function that satisfies $0 <h(b)P'(b) \leq M<\infty$, therefore, in practice, $h(b)$ can be case specific and one can pick $h(b)$ in a way such that $\nu(b_1,b_2, q)$ can be estimated in a computationally convenient way and weakly converges to a Gaussian process. 

Specifically, let $\widehat{M}(b,q)$ and $\widehat{W}(b,q)$ denote the estimators for ${M}(b,q)$ and ${W}(b,q)$, respectively.  Denote the estimator for $\nu(b_1,b_2, q)$ as 
\begin{align}
&\widehat{\nu}(b_1,b_2, q)=\widehat{M}(b_2,q)\widehat{W}(b_1,q)-\widehat{M}(b_1,q)\widehat{W}(b_2,q). \label{eq: general nu estimator} 
\end{align}

\begin{assumption}\label{assu: general DGP}
Assume that we observe actions 
$\{B_{i,\ell}: i=1,\ldots, N,~\ell=1,\ldots, L \}$ 
that are i.i.d.\ random variables across $L$ games and $N$ agents in each game. 
\end{assumption}
Assumption \ref{assu: general DGP} imposes conditions on the data that econometricians observe. When we discuss the asymptotics, we will assume that $N$ is fixed and $L$ goes to infinity.   


\begin{assumption}\label{assu: general gaussian limit}
Assume that as $L$ goes to infinity, 
\begin{align*}
\sqrt{NL}(\widehat{\nu}(\cdot,\cdot,\cdot)-{\nu}(\cdot,\cdot,\cdot))\Rightarrow
\Phi_\kappa(\cdot,\cdot,\cdot)
\end{align*}
in which $\Rightarrow$ denotes weak convergence and  $\Phi_\kappa(\cdot,\cdot,\cdot)$ is a Gaussian process with covariance kernel being
$\kappa\big((b'_1,b'_2,q'),(b''_1,b''_2,q'')\big)$ for $(b'_1,b'_2,q'),(b''_1,b''_2,q'')\in \mathcal{L}$. In addition, $\kappa$ is not a zero function.  
\end{assumption}
Assumption \ref{assu: general gaussian limit} requires that there exists a consistent estimator for $\nu(\cdot)$ which weakly converges to a Gaussian process with proper rescaling. In addition, that $\kappa$ is not a zero function implies that $\widehat{\nu}(\cdot)$ is not degenerate. 
Define $\sigma^2_\nu(b_1,b_2,q)=\kappa\big((b_1,b_2,q),(b_1,b_2,q)\big)$ which is the asymptotic variance of
$\sqrt{NL}(\widehat{\nu}(b_1,b_2,q)-{\nu}(b_1,b_2,q))$ and its estimator as $\widehat{\sigma}^2_
\nu(b_1,b_2,q)$

\begin{assumption} \label{assu: general consistent variance}
Assume that $\widehat{\sigma}^2_\nu(b_1,b_2,q)$ is uniformly consistent for $\sigma^2_\nu(b_1,b_2,q)$ in that $\sup_{(b_1,b_2,q)\in \mathcal{L}}|\widehat{\sigma}^2_\nu(b_1,b_2,q)-{\sigma}^2_\nu(b_1,b_2,q)|\stackrel{p}{\rightarrow}0$. In addition, ${\sigma}^2_\nu(\underline{b}, (\underline{b}+\overline{b})/2,2)>0$. 
\end{assumption}
Assumption \ref{assu: general consistent variance} requires that there exists a uniformly consistent estimator for the asymptotic variance.  We specifically assume that one of the moments is not degenerate. In practice, if the influence functions of the $\widehat{\nu}(\cdot)$ are known, then one can use the estimated influence functions to construct an estimator for the variance function.   Or one might use the bootstrap estimator to estimate the variance function. 

For some small $\epsilon>0$, define
\begin{align}
\widehat{\sigma}^2_{\nu,\epsilon}(b_1,b_2,q)=\max\{\widehat{\sigma}^2_
\nu(b_1,b_2,q),\epsilon\cdot
\widehat{\sigma}^2_\nu(\underline{b}, (\underline{b}+\overline{b})/2,2)\}.\label{eq: general estimated sigma epsilon}
\end{align}
For a weighted function $Q(b_1,b_2,q)>0$ and $\sum_{(b_1,b_2,q)\in\mathcal{L}}Q(b_1,b_2,q) <\infty $, define our test statistic as
\begin{align}
\widehat{T}_{NL}=\sum_{(b_1,b_2,q)\in\mathcal{L}} \max\Big\{ \sqrt{NL} \frac{\widehat{\nu}(b_1,b_2,q)}
{\widehat{\sigma}_{\nu,\epsilon}(b_1,b_2,q)},0\Big\}^2 Q(b_1,b_2,q). \label{eq: general test statistics}
\end{align}

We propose a nonparametric bootstrap combined with the generalized moment selection (GMS) method introduced by \cite{AndrewsSoares2010} and \cite{AndrewsShi2013} to construct the critical values for our test.\footnote{The GMS method is similar to the
recentering method in \cite{Hansen2005} and \cite{DonaldHsu2016}, as
well as the contact set approach in \cite{LSW2010}.} For the nonparametric bootstrap, we bootstrap games. That is, let $\{(B^*_{1\ell},\ldots, B^*_{N\ell})
:\ell  \leq L \}$ be a
bootstrap sample in which $B^*_{i,\ell}=B_{i\ell^*}$ and $\{
\ell^*: \ell \leq L\}$ is an i.i.d.\ bootstrap sample
drawn from the empirical distribution of $\{\ell : \ell\leq L\}$.
Define the bootstrap estimator for $\nu(b_1,b_2, q)$ as 
\begin{align}
&\widehat{\nu}^*(b_1,b_2, q)
=\widehat{M}^*(b_2,q)\widehat{W}^*(b_1,q)-\widehat{M}^*(b_1,q)\widehat{W}^*(b_2,q), \label{eq: general bootstrap nu estimator}
\end{align}
where $\widehat{M}^*(b,q)$ and $\widehat{W}^*(b,q)$
are bootstrap estimators for ${M}(b,q)$ and ${W}(b,q)$, respectively, based on the bootstrap sample. 
The bootstrapped process is defined as $\Phi^*(\cdot,\cdot,\cdot)=\sqrt{NL}(\widehat{\nu}^*(\cdot,\cdot,\cdot)-\widehat{\nu}(\cdot,\cdot,\cdot)).$

\begin{assumption}\label{assu: general bootstrap limit}
Assume that as $L$ goes to infinity, $\Phi^*(\cdot,\cdot,\cdot)\Rightarrow \Phi_\kappa(\cdot,\cdot,\cdot)$ conditional on the sample path with probability approaching 1 where $\Phi_\kappa(\cdot,\cdot,\cdot)$ is as in Assumption \ref{assu: general gaussian limit} and we denote it as $\Phi^*(\cdot,\cdot,\cdot)\stackrel{p}{\Rightarrow} \Phi_\kappa(\cdot,\cdot,\cdot)$.  
\end{assumption}
Assumption \ref{assu: general bootstrap limit} requires that we have a bootstrap method that approximates the limiting process $\Phi_\kappa(\cdot,\cdot,\cdot)$ well. 
As most papers in the moment inequality literature, we use
the GMS method to construct the critical value.
By doing this, one can construct a more powerful test without resorting
to the least favorable configuration. Let $\{\beta_{j},~j=1,2,\ldots\}$ and $\{\kappa_{j},~j=1,2,\ldots\}$ be sequences of positive numbers.
Define the GMS
function $\psi_{NL}(b_1,b_2,q)$ as
\begin{align}
\psi_{NL}(b_1,b_2,q)=-\beta_{NL} \cdot 1\Big(\frac{\sqrt{NL}\widehat{\nu}(b_1,b_2,q)}{\widehat{\sigma}_{\nu,\epsilon}(b_1,b_2,q)}<-\kappa_{NL}\Big). \label{eq: general GMS-function}
\end{align}
We impose conditions on $\beta_{j}$ and $\kappa_{j}$.
\begin{assumption}\label{assu: general GMS}
Assume that (i) $\lim_{j\rightarrow\infty}\kappa_{j}=\infty$ and
$\lim_{j\rightarrow\infty}\kappa_{j}/\sqrt{j}=0$ and (ii)
$\beta_{j}$ is non-decreasing, $\lim_{j\rightarrow\infty}\beta_{j}=\infty$ and
$\lim_{j\rightarrow\infty}\beta_{j}/\kappa_{j}=0$.
\end{assumption}
For a significance level $\alpha<1/2$, define the bootstrapped critical value
$\hat{c}_\eta$ as
\begin{align*}
\hat{c}_\eta=\sup\Big\{c\big|P^*\Big(\sum_{(b_1,b_2,q)\in\mathcal{L}} \max\Big\{
\frac{{\Phi}^*(b_1,b_2,q)}
{\widehat{\sigma}_{\nu,\epsilon}(b_1,b_2,q)}+\psi_{NL}(b_1,b_2,q),0\Big\}^2 Q(b_1,b_2,q)\leq c\Big)\leq 1-\alpha+\eta\Big\}+\eta,
\end{align*}
where $\eta>0$ is an arbitrarily small positive number. Note that $\hat{c}_\eta$ is the $(1-\alpha+\eta)$-th
quantile of the simulated null distribution plus $\eta$ which is
called an infinitesimal uniformity factor in \cite{AndrewsShi2013}. 
Following \cite{HLS2019}, we set $\beta_{j}=0.85\cdot\ln(j)/\ln\ln(j)$, $\kappa_{j}=0.15\cdot\ln(j)$ and
$\eta=10^{-6}$. 

The decision rule is: ``Reject $H_0$ ($H_0'$) when $\widehat{T}_{NL}>\hat{c}_\eta$."

The following theorem shows the asymptotic size control and the power against fixed alternatives of our test. 
\begin{thm} Suppose Assumptions \ref{assu: general DGP}, \ref{assu: general gaussian limit}, \ref{assu: general consistent variance}, \ref{assu: general bootstrap limit}, \ref{assu: general GMS} hold. Then, for $\alpha<1/2$, \\
(a) under $H_0$,
${\lim}_{L\rightarrow\infty}
P(\widehat{T}_{NL}>\hat{c}_\eta)\leq \alpha$;\\
(b) under $H_1$,
${\lim}_{L\rightarrow\infty}
P(\widehat{T}_{NL}>\hat{c}_\eta)= 1.$ \label{thm: general test size and power}
\end{thm}

Theorem \ref{thm: general test size and power} shows that our test has asymptotic size control under the null hypothesis and is consistent against fixed alternatives.  

\begin{remark}\label{remark: hetero number}
In practice, the data may include games with varying numbers of agents. In the Appendix, we discuss how to conduct tests that allow for such heterogeneity. The key issue is that the quasi-inverse equilibrium strategy, the $\xi$ function, implicitly depends on the number of agents in a game. If this strategy varies with the number of agents, then the corresponding moment conditions must be treated separately for each group. Accordingly, when games involve different numbers of agents, we implement a joint test across these groups. For further details, see Appendix \ref{sec: hetero numbers}.
\end{remark}

\begin{remark}\label{remark: boundary}
In practice, if the boundaries $\overline{b}$ and $\underline{b}$ are unknown, we can estimated them by $\min_{i,\ell}B_{i,\ell}$ and $\max_{i,\ell}B_{i,\ell}$ and these estimators converge at a faster rate than parametric rate, so we can ignore the estimation effects. 

\end{remark}

\section{Examples} \label{sec: examples}
In this section, we consider four examples of Bayesian games. The first is a first-price sealed-bid auction; the second is a Tullock contest; the third is a public-good provision game; and the fourth is a Cournot competition model with incomplete information.

Let $V$ denote a random variable representing a player's (or agent's) private type in a game of incomplete information. The interpretation of $V$ varies across games and will be specified in each example below. Throughout the examples, we maintain the following assumption.

\begin{assumption}\label{assu: auction DGP}
Assume that (i) the unobserved valuations
$\{V_{i}: i=1,\ldots, N \}$ are i.i.d.\ with probability density function (PDF), $f(\cdot)$, and cumulative distribution function (CDF), $F(\cdot)$;  (ii)
$f(\cdot)$ is strictly positive and bounded away from zero on a convex and compact support, 
$[\underline{v}, \overline{v}] \subseteq R$ , and is twice continuously differentiable on $(\underline{v}, \overline{v})$.
\end{assumption} 

\subsection{First-price Sealed-bid Auction Model}\label{subsec: auction}
As the first example, we consider a first-price sealed-bid  auction model for homogeneous goods with a fixed number of bidders. 
The bids submitted by the bidders are the observed actions in each auction.  For each auction, the econometrician observes $N$ bids $\{B_{i}:~i=1,\ldots, N\}$, but not bidders' valuations $\{V_{i}:~i=1,\ldots, N. \}$, where $V_{i}$ is bidder $i$'s value.
Thus, the allocation function  of the auctioned product for player $i$ is determined by the indicator function $D(B_{i}, \bm B_{-i})$ defined as
\begin{align*}
D(B_i, \bm B_{-i}) =
\left\{
\begin{array}{ll}
1, & \text{if } B_i > B_j \ \ \forall j \neq i, \\
0, & \text{otherwise.}
\end{array}\right.
\end{align*}
Also, the payment function of the auctioned object for player $i$ is determined by the cost function  defined as $C(B_{i}, \bm B_{-i})=B_i\cdot D(B_{i}, \bm B_{-i})$.
Integrating over the distribution of $\bm B_{-i}$ gives the monotone equilibrium expected allocation for bidder $i$, i.e., $E_{\bm B_{-i}}[D(B_i, \bm B_{-i})]$.  Then for agent $i$, his/her expected equilibrium payoff $\pi(B_i, V_i)$ is
\begin{equation*}
\pi(B_i, V_i) = V_i\cdot E_{\bm B_{-i}}[D(B_i, \bm B_{-i})]  - B_i\cdot E_{\bm B_{-i}}[D(B_i, \bm B_{-i})]
=V_i P(B_i)-T(B_i).
\end{equation*}
  
In this example, Assumption \ref{assu: auction DGP} is what \cite{GPV2000} impose on the distribution of the valuations, which gives rise to the symmetric independent private value paradigm. 
Assume that the observed bids are generated from the valuations satisfying Assumption \ref{assu: auction DGP} and by a BNE bidding strategy, $s(V_{i})$.  It then follows that
\begin{align}
B_{i}= s(V_{i})\equiv V_{i}- \frac{1}{F(V_{i})^{N-1}} \int_{\underline{v}}^{V_{i}}F(u)^{N-1} du,
\label{eq: BNE strategy}
\end{align}
and the observed bids
$\{B_{i}: i=1,\ldots, N. \}$ are i.i.d.\ with PDF $g(\cdot)$ and CDF $G(\cdot)$ with a convex and compact support, $[\underline{b},\overline{b}]\subseteq R$.\footnote{Note that $g(\cdot)$ and $G(\cdot)$ depend on $N$, the number of bidders in the auction, because the equilibrium bidding strategy in (\ref{eq: BNE strategy}) is a function of $N$.}
The BNE requires that the bidding strategy $s(\cdot)$ is strictly increasing.\footnote{Moreover, \cite{GPV2000} show that $s(\cdot)$ is at least three times continuously differentiable on $(\underline{v}, \overline{v})$.}  

It can be shown that $P(b)=G(b)^{N-1}$ and $T(b)=bG(b)^{N-1}$. As a result, $\xi(\cdot)$ in (\ref{eq:implicit}) becomes the quasi-inverse bidding strategy in \cite{GPV2000}:
\begin{align}
\xi(b) = b+\frac{1}{N-1}\frac{G(b)}{g(b)}, \label{eq: gpv}
\end{align}
where $G(\cdot)$ and $g(\cdot)$ are the CDF and PDF of bids, respectively.
\cite{GPV2000} show that the observed bid distribution $G(\cdot)$ is an equilibrium bid distribution if and only if the quasi-inverse strategy, $\xi(b)$, in (\ref{eq: gpv}) 
is strictly increasing, which is implied by Theorem \ref{thm: LZ2018}.   

Therefore, we define the null hypothesis as
\begin{align}
&H_{0,au}: \xi(b)= b+ \frac{1}{N-1}\frac{G(b)}{g(b)}~\text{is weakly increasing in $b$}. \label{eq: auction H0}
\end{align}

\cite{LV2021} test (\ref{eq: auction H0}) by testing weak concavity of the integral of $\xi(\cdot)$ using the LCM approach. We apply Lemma \ref{lemma: general H0 transformation} to transform the null hypothesis in (\ref{eq: auction H0}) by setting $h(b)=G(b)^{2-N}(N-1)^{-1}$, so we have 
\begin{align*}
&M(b,q)=\int_{b}^{b+\frac{a}{q}} h(\tilde{b}) T'(\tilde{b}) d\tilde{b}=
\int_{b}^{b+\frac{a}{q}} \big(\tilde{b} g(\tilde{b})+(N-1)^{-1}G(\tilde{b})\big) d\tilde{b},\\
&W(b,q)=\int_{b}^{b+\frac{a}{q}} h(\tilde{b}) P'(\tilde{b}) d\tilde{b}=\int_{b}^{b+\frac{a}{q}} g(\tilde{b}) d\tilde{b}.
\end{align*}
Note that $M(b,q)$ and $W(b,q)$ are both identified and can be expressed as unconditional means.  We summarize the results in the following lemma. 

\begin{lemma}\label{lemma: auction H0 transformation} 
The $H_{0,au}$ in (\ref{eq: auction H0}) is equivalent to
\begin{align}
&H_{0,au}': \nu(b_1,b_2, q)=M(b_2,q)W(b_1,q)-M(b_1,q)W(b_2,q)\leq 0~\text{for all $(b_1,b_2, q)\in\mathcal{L}$} \label{eq: auction H0'},
\end{align}
where 
\begin{align}
&M(b,q)= E\Big[B_{i} 1\Big( b \leq  B_{i}\leq b+\frac{a}{q} \Big)\Big]\notag\\
&~~~~~~~~~~~~~~~~+
\frac{1}{N-1} E\Big[ 1\Big(B_{i}\leq b+\frac{a}{q} \Big) \Big(b+\frac{a}{q} - B_{i}\Big)
 - 1(B_{i}\leq b ) (b - B_{i}) \Big], \label{eq: auction M}\\
&W(b,q)=E\Big[1\Big( b \leq  B_{i}\leq b+\frac{a}{q} \Big) \Big]\label{eq: auction W}.
\end{align}
\end{lemma}
Based on (\ref{eq: auction M}) and (\ref{eq: auction W}),  we can apply the test in Section \ref{sec: proposed test}. In Appendix A, we provide a uniform consistent estimator for the asymptotic variance for each moment and show that for this example, the sufficient conditions hold for Theorem \ref{thm: general test size and power}.

\begin{remark} \label{remark: procurement} {\bf (Procurement Auctions)} 
The example above has considered high-bid first-price auctions to be in line with the literature where the identification results starting with \cite{GPV2000} have been established for high-bid auctions. For low-bid first-price auctions such as procurement auctions considered in our empirical application, procurement with cost $c$ and bid $b$ is strategically equivalent to a high-bid first-price auction with value $\check{v}=K-c$ and bid $\check{b}=K-b$, for any constant $K\ge\bar{b}$. Let $\check{G}(\cdot)$, $\check{g}(\cdot)$ represent the CDF and PDF of $\check{b}$, respectively. Under reflection, $\check{G}(\check{b})=1-G(K-\check{b})$ and $\check{g}(\check{b})=g(K-\check{b})$. Expected allocation $\check{P}(\check{b})=\check{G}(\check{b})^{N-1}=(1-G(K-\check{b}))^{N-1}$, the probability that $\check{b}$ is the highest reflected bid --- exactly the probability that $b=K-\check{b}$ is the lowest original bid, so it coincides with the procurement winning probability $P(b)=(1-G(b))^{N-1}$ at $b=K-\check{b}$; reflection does not change who wins, it re-expresses the event so that winning is monotone increasing in the transformed action, restoring $\check{P}'(\cdot)>0$. Expected payment $\check{T}(\check{b})=\check{b}\cdot\check{G}(\check{b})^{N-1}$, where the ``payment'' is the reflected bid itself; with $\check{v}=K-c$, the winner's reflected profit is $\check{v}-\check{b}=(K-c)-(K-b)=b-c$, exactly the procurement profit, so $\check{v}\check{P}-\check{T}$ equals the true payoff $(b-c)\cdot\Pr(\text{win})$, therefore best responses are preserved, not merely distributions. The resulting quasi-inverse is the expression in (3.2) under $\check{G}$: $\check{\xi}(\check{b})= \check{b}+\check{G}(\check{b})/((N-1)\check{g}(\check{b}))= K-\xi(K-\check{b})$, where 
\begin{align*}
\xi(b)=b-\frac{1}{N-1}\frac{1-G(b)}{g(b)}.
\end{align*}
Hence $\check{\xi}$ is increasing in $\check{b}$ if and only if $\xi$ is increasing in $b$. 
In this case, we have $P(b)=(1-G(b))^{N-1}$ and $T(b)=b(1-G(b))^{N-1}$ and we would like to test $T'(b)/P'(b)$ is weakly decreasing in $b$ or equivalently $-T'(b)/P'(b)$ is weakly increasing in $b$.  In sum, 
$P'(b)=-g(b)(1-G(b))^{N-2}(N-1)$ and $-T'(b)=-(1-G(b))^{N-1} + b g(b)(1-G(b))^{N-2}(N-1)$.  Let 
$h(b)=-(1-G(b))^{2-N}(N-1)^{-1}$
and
define
\begin{align*}
&M(b,q)=\int_{b}^{b+\frac{a}{q}} h(\tilde{b})( -T'(\tilde{b})) d\tilde{b}=
\int_{b}^{b+\frac{a}{q}} \big(\tilde{b} g(\tilde{b})-(N-1)^{-1}(1-G(\tilde{b}))\big) d\tilde{b},\\
&W(b,q)=\int_{b}^{b+\frac{a}{q}} h(\tilde{b}) P'(\tilde{b}) d\tilde{b}=\int_{b}^{b+\frac{a}{q}} g(\tilde{b}) d\tilde{b},
\end{align*}
which are identified as
\begin{align}
&M(b,q)= E\Big[B_{i} 1\Big( b \leq  B_{i}\leq b+\frac{a}{q} \Big)\Big]-\frac{1}{N-1}\frac{a}{q}\notag\\
&~~~~~~~~~~~~~~~~+
\frac{1}{N-1} E\Big[ 1\Big(B_{i}\leq b+\frac{a}{q} \Big) \Big(b+\frac{a}{q} - B_{i}\Big)
 - 1(B_{i}\leq b ) (b - B_{i}) \Big], \label{eq: Pro auction M}\\
&W(b,q)=E\Big[1\Big( b \leq  B_{i}\leq b+\frac{a}{q} \Big) \Big]\label{eq: Pro auction W},
\end{align}
respectively. 
\end{remark}

\subsection{Tullock Contest Model}\label{subsec: contest model}
In the context of a Tullock contest such as in \cite{Tullock1980}, \cite{Ryvkin2010} and \cite{Ewerhart2014}, $\{ V_{i}: i=1,\ldots,N \}$ represent
the contestants' valuations which are unobservable to the econometrician.
The efforts exerted by contestants are the observed actions.
To be specific,  $N$ contestants choose effort levels $B_i \geq 0$  to compete for a prize. The winning probability, conditional on rivals' effort levels $\bm B_{-i}$, is given by the symmetric contest–success function
\begin{equation}
D(B_i, \bm B_{-i}) = \frac{B_i}{\sum_{s=1}^{N} B_s}.
\label{eq:csf}
\end{equation}
Contestant $i$ has a known effort cost $C(B_i,\bm B_{-i})=c(B_i)$, where $c$ is twice continuously differentiable and $c'(b)>0$. We allow marginal cost to decrease. The cost function and feasible effort set are held fixed across candidate action distributions.
Integrating over the distribution of $\bm B_{-i}$ gives the monotone equilibrium expected allocation $P(B_i)=E_{\bm B_{-i}}[D(B_i, \bm B_{-i})]$ and expected effort cost $T(B_i)=E_{\bm B_{-i}}[C(B_i, \bm B_{-i})]=c(B_i)$. 
The expected payoff is
\begin{equation*}
\Pi(B_i, V_i) = V_i \, E_{\bm B_{-i}}[D(B_i, \bm B_{-i})] - E_{\bm B_{-i}}[C(B_i, \bm B_{-i})]= V_iP(B_i)-T(B_i).
\end{equation*}

The FOC for an interior optimum is
\[
\frac{\partial \Pi(B_{i}; V_{i})}{\partial B_{i}}
= V_i\,P'(B_i)-c'(B_i)=0
\quad\Rightarrow\quad
V_i=\frac{c'(B_i)}{P'(B_i)}.
\]

In this case, we have $P(b)=E_{\bm B_{-i}}[D(b,\bm B_{-i})]$ and $T(b)=c(b)$. 
It is straightforward to see that 
$P'(b)=b^{-1}E_{\bm B_{-i}}[D(b,\bm B_{-i})(1-D(b,\bm B_{-i}))]$ which is strictly positive. 
We then have
\begin{align*}
\xi(b)=\frac{T'(b)}{P'(b)}=
\frac{c'(b)}{b^{-1}E_{\bm B_{-i}}[D(b,\bm B_{-i})(1-D(b,\bm B_{-i}))]}
\end{align*}
which reduces to the linear-cost case when $c(b)=b$. 
Therefore, we define the null hypothesis as
\begin{align}
&H_{0,con}: \xi(b)= \frac{c'(b)}{b^{-1}E_{\bm B_{-i}}[D(b,\bm B_{-i})(1-D(b,\bm B_{-i}))]}~\text{is weakly increasing in $b$}. \label{eq: contest H0}
\end{align}

To see when monotonicity is informative, let $S=\sum_{j\ne i}B_j$ and suppose $S>0$. Under Assumption \ref{assu: PT},
\[
P'(b;G)=E_G\!\left[\frac{S}{(b+S)^2}\right]>0,
\qquad P''(b;G)=-2E_G\!\left[\frac{S}{(b+S)^3}\right]<0,
\]
and therefore
\[
\xi'(b;G)=\frac{c''(b)P'(b;G)-c'(b)P''(b;G)}{[P'(b;G)]^2}.
\]
With linear or increasing convex cost, the quasi-inverse is strictly increasing for every admissible action distribution. Thus the weak-monotonicity null adds no restriction within the maintained class of action distributions. Allowing declining marginal cost makes its sign depend on the distribution of rivals' efforts.\footnote{Appendix \ref{app: example constructions} gives a fixed-payoff contest construction with a nonempty equilibrium null and a violating action distribution. This construction establishes nonemptiness for that particular cost function and feasible effort set; we do not provide general existence conditions for monotone equilibria under declining marginal cost, since the uniqueness result of \cite{Ewerhart2020} covers the linear-cost benchmark. Nonemptiness of the null should therefore be verified for the specific cost function maintained in a given application, as we do in Appendix \ref{app: example constructions}.}

We apply Lemma \ref{lemma: general H0 transformation} to transform the null hypothesis in (\ref{eq: contest H0}) by setting $h(b)=b g(b)$ with $g(b)$ being the density function of the actions under BNE,
we have 
\begin{align*}
&M(b,q)=\int_{b}^{b+\frac{a}{q}} h(\tilde{b}) T'(\tilde{b}) d\tilde{b}=
\int_{b}^{b+\frac{a}{q}} c'(\tilde{b}) \tilde{b}  g(\tilde{b}) d\tilde{b},\\
&W(b,q)=\int_{b}^{b+\frac{a}{q}} h(\tilde{b}) P'(\tilde{b}) d\tilde{b}=\int_{b}^{b+\frac{a}{q}}E_{\bm B_{-i}}[D(\tilde{b},\bm B_{-i})(1-D(\tilde{b},\bm B_{-i}))] g(\tilde{b}) d\tilde{b}.
\end{align*}
Note that $M(b,q)$ and $W(b,q)$ are both identified and can be expressed as unconditional means.  We summarize the results in the following lemma. 

\begin{lemma}\label{lemma: contest H0 transformation} 
The $H_{0,con}$ in (\ref{eq: contest H0}) is equivalent to
\begin{align}
&H_{0,con}': \nu(b_1,b_2, q)=M(b_2,q)W(b_1,q)-M(b_1,q)W(b_2,q)\leq 0~\text{for all $(b_1,b_2, q)\in\mathcal{L}$} \label{eq: contest H0'},
\end{align}
where 
\begin{align}
&M(b,q)= E\Big[ c'(B_i)  B_{i}  1\Big( b \leq  B_{i}\leq b+\frac{a}{q} \Big)\Big],\label{eq: contest M}\\
&W(b,q)=E\Big[1\Big( b \leq  B_{i}\leq b+\frac{a}{q} \Big) \frac{B_{i}}{\sum_{j=1}^N B_{j} }\Big(1-\frac{B_{i}}{\sum_{j=1}^N B_{j}}\Big)\Big]\label{eq: contest W}.
\end{align}
\end{lemma}
Based on  (\ref{eq: contest M}) and (\ref{eq: contest W}), we are able to apply the test in Section \ref{sec: proposed test}.

\subsection{public-good Provision Model}
\label{subsec: public}
Suppose there are $N$ contributors, and contributor $i$ voluntarily contributes $B_i$ to provide a public good. 
Let $\Omega$ be the common benefit function of the total contributions.  We assume that $\Omega$ is known and twice continuously differentiable, with $\Omega'(q)>0$ on the feasible range of total contributions. We allow its curvature to change sign. We assume $D(B_i, \bm B_{-i})=\Omega ( \sum_{s=1}^N B_s).$ Also, if contributor $i$ decides to contribute $B_i$, the net payment will be $B_i$ for contributor $i$, so we have $C(B_i, \bm B_{-i})=B_i$. Contributor $i$'s private value of the public good $V_i$ affects the marginal utility of consuming the public good. Then given the actions $(B_1,\ldots, B_N)$, the utility for contributor $i$ is $V_i D(B_i, \bm B_{-i})- C(B_i, \bm B_{-i}).$
Define $P(b)=E_{\bm B_{-i}}[D(b, \bm B_{-i})]$ and $T(b) =E_{\bm B_{-i}}[C(b, \bm B_{-i})]=b$.
Following \cite{BagRoy2011} who extend \cite{Varian1994}'s model with complete information to allow for incomplete information, the expected payoff for contributor $i$ is
\begin{equation*}
\pi(B_i, V_i) = V_i \, E_{\bm B_{-i}}[D(B_i, \bm B_{-i})] - B_i= V_i P(B_i)-T(B_i),
\end{equation*}
and we have 
\begin{align*}
\xi(b) = \frac{T'(b)}{P'(b)} =
\frac{1}{E_{\bm B_{-i}}\big[ \Omega'\big( b+ \sum_{j\neq i} B_j \big) \big]}. 
\end{align*}
Therefore, we define the null hypothesis as
\begin{align}
&H_{0,pg}: \xi(b)= \frac{1}{E_{\bm B_{-i}}\big[ \Omega'\big( b+ \sum_{j\neq i} B_j \big) \big]}~\text{is weakly increasing in $b$}. \label{eq: public H0}
\end{align}

To see when monotonicity is informative, let $S_{-i}=\sum_{j\ne i}B_j$. If the benefit function is concave throughout the feasible range, then
\[
P''(b;G)=E_G[\Omega''(b+S_{-i})]\leq0,
\qquad
\xi'(b;G)=-\frac{P''(b;G)}{[P'(b;G)]^2}\geq0
\]
for every admissible $G$. In that special case, the weak-monotonicity null adds no restriction on the action distribution beyond the maintained assumptions. Strict concavity gives a strictly increasing inverse, whereas an affine benefit function gives a constant inverse and does not satisfy the strict condition in Theorem \ref{thm: LZ2018}. Allowing $\Omega''$ to change sign can make monotonicity informative because the distribution of rivals' contributions determines the weights in expected curvature. This possibility must be checked for the fixed benefit function used in an application; Appendix \ref{app: example constructions} verifies both cases under a fixed public-good specification.

We apply Lemma \ref{lemma: general H0 transformation} to transform the null hypothesis in (\ref{eq: public H0}) by setting $h(b)= g(b)$ with $g(b)$ being the density function of the actions under BNE,
we have 
\begin{align*}
&M(b,q)=\int_{b}^{b+\frac{a}{q}} h(\tilde{b}) T'(\tilde{b}) d\tilde{b}=
\int_{b}^{b+\frac{a}{q}} g(\tilde{b}) d\tilde{b},\\
&W(b,q)=\int_{b}^{b+\frac{a}{q}} h(\tilde{b}) P'(\tilde{b}) d\tilde{b}=\int_{b}^{b+\frac{a}{q}}E_{\bm B_{-i}}\big[ \Omega'\big( \tilde{b}+ \sum_{j\neq i} B_j \big) \big] g(\tilde{b}) d\tilde{b}.
\end{align*}
Note that $M(b,q)$ and $W(b,q)$ are both identified and can be expressed as unconditional means.  We summarize the results in the following lemma. 

\begin{lemma}\label{lemma: public H0 transformation} 
The $H_{0,pg}$ in (\ref{eq: public H0}) is equivalent to
\begin{align}
&H_{0,pg}': \nu(b_1,b_2, q)=M(b_2,q)W(b_1,q)-M(b_1,q)W(b_2,q)\leq 0~\text{for all $(b_1,b_2, q)\in\mathcal{L}$} \label{eq: public H0'},
\end{align}
where 
\begin{align}
&M(b,q)= E\Big[1\Big( b \leq  B_{i}\leq b+\frac{a}{q} \Big)\Big],\label{eq: public M}\\
&W(b,q)=E\Big[1\Big( b \leq  B_{i}\leq b+\frac{a}{q} \Big) \Omega'(\sum_{j=1}^N B_{j} )\Big]\label{eq: public W}.
\end{align}
\end{lemma}


\subsection{Cournot Competition Model with Private Information}
\label{subsec: cournot}
As the last example, we consider testing in Cournot competition models. 
In the spirit of \cite{Vives2002}, let each firm $i$ choose output $B_i$ and have a private cost–efficiency parameter $V_i>0$. We allow total production cost to be $c(B_i)/V_i$, where $c$ is known and twice continuously differentiable, with $c'(b)>0$ on the feasible output interval and $c(0)=0$ when zero output is feasible. We do not impose convexity on $c$. 
Let $\mathcal{I}(\cdot)$ be the inverse demand function which is known.\footnote{\cite{AryalZincenko2024} extend \cite{Vives2002} to allow for stochastic demand and a common technology shock.}$^,$\footnote{In some cases,
the inverse demand function is unknown, but can be identified and estimated consistently from a separate sample.  In this case, if the number of observations in the separate sample is large compared to the sample size of the testing sample so that we can ignore the estimation effect of the demand function asymptotically, then we can treat the estimated inverse demand function as a known function when we conduct the test. } Also we assume that $\mathcal{I}(\cdot)$ is twice continuously differentiable, strictly decreasing and convex. Revenue is $D(B_i,\bm B_{-i})\equiv \mathcal{I}(B_i+\sum_{j\ne i}B_j)B_i$; the additional demand-curvature condition below ensures concavity of expected revenue. 
Let $C(B_i,\bm B_{-i})=c(B_i)$, so $c(B_i)/V_i$ is the total cost for agent $i$.  Therefore, the total profit for agent $i$ is $D(B_i,\bm B_{-i})-C(B_i,\bm B_{-i})/V_i$.   Integrating over the distribution of $\bm B_{-i}$, the expected payoff can be rewritten as
\begin{equation*}
\ddot{\pi}(B_i, V_i) = E_{\bm B_{-i}}[D(B_i, \bm B_{-i})]- \frac{c(B_i)}{V_i}
\equiv P(B_i)-\frac{T(B_i)}{V_i},
\end{equation*}
where $P(b)=E_{\bm B_{-i}}[D(b,\bm B_{-i})]$ and $T(b)=c(b)$. 
Assume that the demand curvature is known to the researcher and small enough, e.g., $-2 E[\mathcal{I}'] > E[\mathcal{I}''] \, b$ for all relevant $b$, where $\mathcal{I}'$ and $\mathcal{I}''$ denote the first and second derivatives, respectively. We also maintain $P'(b;G)>0$ on the interior of the feasible interval, as required by Assumption \ref{assu: PT}.
Note that since maximizing $\ddot{\pi}(b, V_i)$ over $b$ is equivalent to maximizing ${\pi}(b, V_i)\equiv V_i\cdot \ddot{\pi}(b, V_i) = V_iP(B_i)-T(B_i)$ over $b$, we can still have the same form as before. 
Therefore, we have 
\begin{align*}
\xi(b) = \frac{T'(b)}{P'(b)} =
\frac{c'(b)}{E_{\bm B_{-i}}\big[ \mathcal{I}\big( b+ \sum_{j\neq i}^N B_j \big)+\mathcal{I}'\big( b+ \sum_{j\neq i}^N B_j \big) b \big]}. 
\end{align*}
Thus, we define the null hypothesis as
\begin{align}
&H_{0,cc}: \xi(b)= \frac{c'(b)}{E_{\bm B_{-i}}\big[ \mathcal{I}\big( b+ \sum_{j\neq i}^N B_j \big)+\mathcal{I}'\big( b+ \sum_{j\neq i}^N B_j \big) b \big]}~\text{is weakly increasing in $b$}. \label{eq: Cournot H0}
\end{align}

Under the maintained demand-curvature condition, note that
\[
P''(b;G)=E_G[2\mathcal I'(b+S_{-i})+b\mathcal I''(b+S_{-i})]<0,
\qquad S_{-i}=\sum_{j\ne i}B_j.
\]
Consequently, linear cost makes the weak-monotonicity null automatic. More generally,
\[
\xi'(b;G)=\frac{c''(b)P'(b;G)-c'(b)P''(b;G)}{[P'(b;G)]^2}.
\]
Convex costs also give an increasing inverse under concave expected revenue. Allowing declining marginal production costs can make monotonicity depend on the output distribution even under a fixed linear inverse demand function.\footnote{Appendix \ref{app: example constructions} gives a fixed-payoff Cournot construction with a nonempty equilibrium null and a violating action distribution. As in the contest case, this establishes nonemptiness for the particular cost function and feasible output set used there rather than for the whole class of specifications with declining marginal cost, and the maintained demand-curvature condition does not by itself guarantee existence once convexity of $c(\cdot)$ is dropped.}

We apply Lemma \ref{lemma: general H0 transformation} to transform the null hypothesis in (\ref{eq: Cournot H0}) by setting $h(b)= g(b)$ with $g(b)$ being the density function of the actions under BNE,
so we have 
\begin{align*}
&M(b,q)=\int_{b}^{b+\frac{a}{q}} h(\tilde{b}) T'(\tilde{b}) d\tilde{b}=
\int_{b}^{b+\frac{a}{q}} c'(\tilde{b}) g(\tilde{b}) d\tilde{b},\\
&W(b,q)=\int_{b}^{b+\frac{a}{q}} h(\tilde{b}) P'(\tilde{b}) d\tilde{b}=\int_{b}^{b+\frac{a}{q}}E_{\bm B_{-i}}\big[ \mathcal{I}\big( \tilde{b}+ \sum_{j\neq i} B_j \big)+\mathcal{I}'\big( \tilde{b}+ \sum_{j\neq i} B_j \big) \tilde{b} \big] g(\tilde{b}) d\tilde{b}.
\end{align*}
Note that $M(b,q)$ and $W(b,q)$ are both identified and can be expressed as unconditional means.  We summarize the results in the following lemma. 

\begin{lemma}\label{lemma: Cournot H0 transformation} 
The $H_{0,cc}$ in (\ref{eq: Cournot H0}) is equivalent to
\begin{align}
&H_{0,cc}': \nu(b_1,b_2, q)=M(b_2,q)W(b_1,q)-M(b_1,q)W(b_2,q)\leq 0~\text{for all $(b_1,b_2, q)\in\mathcal{L}$} \label{eq: Cournot H0'},
\end{align}
where 
\begin{align}
&M(b,q)= E\Big[c'(B_i)1\Big( b \leq  B_{i}\leq b+\frac{a}{q} \Big)\Big],\label{eq: Cournot M}\\
&W(b,q)=E\Big[1\Big( b \leq  B_{i}\leq b+\frac{a}{q} \Big)\Big( \mathcal{I}\big(\sum_{j=1}^N B_{j} \big)+\mathcal{I}'\big(\sum_{j=1}^N B_{j} \big)B_{i}\Big)\Big]\label{eq: Cournot W}.
\end{align}
\end{lemma}


\section{Heterogeneous Games of Incomplete Information}\label{sec: Games of Incomplete Information}
In this section, we extend the test for homogeneous games of incomplete information in Section \ref{sec: incomplete information} to allow for heterogeneous games of incomplete information including observed game heterogeneity and/or unobserved game heterogeneity. For both cases, we consider nonparametric testing in the main text and  semiparametric testing via homogenization in Appendix C.


\subsection{Observed Game Heterogeneity: Nonparametric Test} \label{sec: observed hetero nonparametric}
We extend our test in Section \ref{sec: incomplete information}  to control for the observed heterogeneity in the model.
Specifically, let $X$ with support $\mathcal{X}$ denote a vector of covariates that is observed for each game. For example,  $X$ represents the characteristics of a game, such as the observed quality of a good. 
Without loss of generality, we assume that $X$ is a scalar and $\mathcal{X}=[0,1]$.\footnote{It is straightforward to allow for $X$ being vectors of observed characteristics.}
We impose the following conditions on the distribution of types. 

\begin{assumption} \label{assu: value hetero}
Assume that (i) Conditioning on $X=x$ for all $x\in \mathcal{X}$, agent $i$, $i = 1, \ldots, N,$ draws his/her type $V_i$ independently from the other agents from a conditional CDF $F(\cdot|x)$ with the conditional PDF $f(\cdot|x)$. 
(ii)
$f(\cdot|x)$ is strictly positive and bounded away from zero on its support, a compact
interval $\mathcal{V}_x=[\underline{v}_x, \overline{v}_x] \subseteq R$ , and is twice continuously differentiable on $(\underline{v}_x, \overline{v}_x)$.
\end{assumption}



We define $P(b,x)$ and $T(b,x)$ to be the conditionally expected outcome allocations and conditionally expected payment, respectively, for agent $i$ when agent $i$ plays action $b$ that is integrated over the distribution of rivals' strategies $\bm B_{-i}$ given $X=x$.  Specifically,
\begin{align}
&P(b,x)=E_{\bm B_{-i}|X=x}[D_i(B_i,\bm B_{-i})|B_i=b,X=x],\nonumber\\
&T(b,x)=E_{\bm B_{-i}|X=x}[C_i(B_i,\bm B_{-i})|B_i = b,X=x].
\label{eq: P and T covaraite}
\end{align}

Suppose all the conditions in Section \ref{sec: necessary and sufficient conditions} holds after we condition $X=x$ for all $x\in\mathcal{X},$ then we can formulate the null hypothesis as
\begin{align}\label{eq: general observed hetero H0}
H_{0,x}:~ \xi(b,x)=\frac{T'(b,x)}{P'(b,x)}~\text{is weakly increasing in $b$ for all $x\in\mathcal{X}$.}  
\end{align}
We then have the following equivalence result.
\begin{lemma}\label{lemma: general observed hetero H0 transformation} 
Let $h(b,x)$ be a known function such that $0 <h(b,x)P'(b,x) \leq M<\infty$.  
Then $H_0$ in (\ref{eq: general observed hetero H0}) is equivalent to
\begin{align}
H_{0,x}': \nu(b_1,b_2, x, q)&=M(b_2,x,q)W(b_1,x,q)-M(b_1,x,q)W(b_2,x,q)\leq 0\notag \\
&~\text{for all $(b_1,b_2,x, q)\in\mathcal{L}_x$} \label{eq: general observed hetero H0'},
\end{align}
where 
\begin{align}
&M(b,x,q)= \int_{x}^{x+\frac{1}{q}}\int_{b}^{b+\frac{a}{q}} h(\tilde{b},\tilde{x})T'(\tilde{b},\tilde{x}) d\tilde{b}d\tilde{x}, \notag\\
&
W(b,x,q)=\int_{x}^{x+\frac{1}{q}}\int_{b}^{b+\frac{a}{q}} h(\tilde{b},\tilde{x})P'(\tilde{b},\tilde{x}) d\tilde{b}d\tilde{x},~\text{and}\label{eq: general observed hetero M W}\\
&\mathcal{L}_x=\Big\{(b_1,b_2, x, q):
\Big(\frac{b_1-\underline{b}}{a}, \frac{b_2-\underline{b}}{a}, x\Big)
q\in(0,1,\ldots,q-1)^3, b_1>b_2~\text{for q=2,3,\ldots}  \Big\}.\label{eq: L cov}
\end{align}
\end{lemma}
With Lemma \ref{lemma: general observed hetero H0 transformation}, the general test will be similar to that in Section \ref{sec: proposed test}. Therefore, we omit the details. 


We conclude this subsection by summarizing the $M(b,x,q)$ and $W(b,x,q)$ in the forms of unconditional means in each example in Section \ref{sec: examples}. 
 
\begin{enumerate}
\item {\bf First-price Sealed-bid Auction Model.} For example, in the context of timber auctions, a natural auction-level covariate is tract volume, which scales the economic stakes of the sale and can mechanically affect bids even in the absence of changes in strategic incentives. 
\begin{align*}
&M(b,x,q)= E\Big[B_{i} 1\Big( b \leq  B_{i}\leq b+\frac{a}{q} \Big) 1(x\leq X \leq x+\frac{1}{q})\Big]\notag\\
&~~~~~~+
\frac{1}{N-1} E\Big[ \Big(1\Big(B_{i}\leq b+\frac{a}{q} \Big) \Big(b+\frac{a}{q} - B_{i}\Big)
 - 1(B_{i}\leq b ) (b - B_{i})\Big) 1(x\leq X \leq x+\frac{1}{q})\Big], \\
&W(b,x,q)=E\Big[1\Big( b \leq  B_{i}\leq b+\frac{a}{q} \Big) 1(x\leq X \leq x+\frac{1}{q})\Big].
\end{align*}

For procurement auctions as in Remark \ref{remark: procurement}, a natural auction-level covariate is engineer's estimate, as in the empirical application in this paper. We have 
\begin{align*}
&M(b,x,q)= E\Big[B_{i} 1\Big( b \leq  B_{i}\leq b+\frac{a}{q} \Big) 1(x\leq X \leq x+\frac{1}{q})\Big]- \frac{1}{N-1}\frac{a}{q}
E\Big[1(x\leq X \leq x+\frac{1}{q})\Big]\notag\\
&~~~~~~+
\frac{1}{N-1} E\Big[ \Big(1\Big(B_{i}\leq b+\frac{a}{q} \Big) \Big(b+\frac{a}{q} - B_{i}\Big)
 - 1(B_{i}\leq b ) (b - B_{i})\Big) 1(x\leq X \leq x+\frac{1}{q})\Big], \\
&W(b,x,q)=E\Big[1\Big( b \leq  B_{i}\leq b+\frac{a}{q} \Big) 1(x\leq X \leq x+\frac{1}{q})\Big].
\end{align*}

\item {\bf Tullock Contest Model.} Maintain the known effort-cost function $c$ from Section \ref{subsec: contest model} across covariate cells. In a Tullock-type contests, prize shifters are one of the key covariates because they scale the marginal benefit of effort. The optimal effort responds mechanically to the prize, i.e.,  higher stakes raise equilibrium effort and change the distribution of observed actions. 
\begin{align*}
&M(b,x,q)= E\Big[c'(B_i) B_{i}   1\Big( b \leq  B_{i}\leq b+\frac{a}{q} \Big)1(x\leq X \leq x+\frac{1}{q})\Big],\\
&W(b,x,q)=E\Big[1\Big( b \leq  B_{i}\leq b+\frac{a}{q} \Big) \frac{B_{i}}{\sum_{j=1}^N B_{j} }\Big(1-\frac{B_{i}}{\sum_{j=1}^N B_{j}}\Big)1(x\leq X \leq x+\frac{1}{q})\Big].
\end{align*}

\item {\bf Public-Good Provision Model.} 
Let $\Omega(\cdot,x)$ be a known common benefit function of the total contributions conditional on $X=x$. A salient covariate in models of public-good provision is the endowment. In experimental settings, the endowment is a core game-level primitive because it shapes both the maximum amount that can be allocated to the public good and the monetary payoffs associated with any given contribution profile.
\begin{align*}
&M(b,x,q)= E\Big[1\Big( b \leq  B_{i}\leq b+\frac{a}{q} \Big)1(x\leq X \leq x+\frac{1}{q})\Big],\\
&W(b,x,q)=E\Big[1\Big( b \leq  B_{i}\leq b+\frac{a}{q} \Big) \Omega'\Big(\sum_{j=1}^N B_{j}, X\Big)1(x\leq X \leq x+\frac{1}{q})\Big].
\end{align*}

\item {\bf Cournot Competition Model.} Maintain the known production-cost function $c$ from Section \ref{subsec: cournot} across covariate cells. Let $\mathcal{I}(\cdot,x)$ be the inverse demand function conditional on $X=x$ which is known. Market level covariates include demand shifters such as population, income, and seasonality. It is important to control demand shifters in demand estimation because these variables shift the level (and often the slope) of market demand, and therefore move equilibrium prices and quantities even when firms’ costs and conduct are unchanged. Ignoring these covariates results in lack of comparability across markets or time, and can lead the econometrician to misattribute demand-driven variation in outputs and prices to differences in marginal costs, productivity, or the degree of market power.  
\begin{align*}
M(b,x,q)=& E\Big[c'(B_i)1\Big( b \leq  B_{i}\leq b+\frac{a}{q} \Big)1(x\leq X \leq x+\frac{1}{q})\Big],\\
W(b,x,q)=&E\Big[1\Big( b \leq  B_{i}\leq b+\frac{a}{q} \Big)\Big( \mathcal{I}\big(\sum_{j=1}^N B_{j},X \big)\Big)1(x\leq X \leq x+\frac{1}{q})\Big]\\
&+E\Big[1\Big( b \leq  B_{i}\leq b+\frac{a}{q} \Big)\Big( \mathcal{I}'\big(\sum_{j=1}^N B_{j},X \big)B_{i}\Big)1(x\leq X \leq x+\frac{1}{q})\Big].
\end{align*}

\end{enumerate}

\begin{remark}\label{remark: covariates}
We consider a nonparametric test for games of incomplete information with covariates. While the test is general, it may suffer from the curse of dimensionality, as is typically the case with nonparametric estimation and testing, when the number of continuous covariates is large. Importantly, with or without covariates, each moment condition to be jointly tested can be estimated at the parametric ($\sqrt{L}$) rate. Thus, the curse of dimensionality arises here not from a deterioration in the convergence rate of individual moments, but from the increasing number of moments that must be jointly tested as the number of continuous covariates grows. This increase in the dimensionality of the moment conditions may reduce the power of the test. In the Appendix, following \cite{HHS2003}, \cite{MMSX2021}, and \cite{CHS2020}, we impose a semiparametric model that homogenizes actions across covariates, thereby mitigating the curse of dimensionality.

\end{remark}

\subsection{Unobserved Game Heterogeneity} \label{sec: unobserved hetero nonparametric}
We extend our tests in Section \ref{sec: observed hetero nonparametric}  to allow for unobserved  heterogeneity in the model to the nonparametric case.  

Suppose that for each game, in addition to an observed vector of covariates, $X$, there is an unobserved vector of $Z$ denoting the unobserved characteristics which is common knowledge among players, but not directly observable to econometricians. Without loss of generality, we assume both $X$ and $Z$ to be scalars with $\mathcal{X}=[0,1]$ and $\mathcal{Z}=[0,1]$.\footnote{It is straightforward to allow for $X$ and $Z$ being vectors of observed and unobserved characteristics, respectively.}
The model is the same as in Section \ref{sec: observed hetero nonparametric} except that we have an additional unobserved covariate $Z$. 
For a given strategy $s$ and given $X=x$ and $Z=z$, we write $s(v,x,z)$ for the action played by type $v$.

Following Section \ref{sec: observed hetero nonparametric}, we have
\begin{align}
\xi(b,x,z) \;\equiv\; \frac{T'(b,x,z)}{P'(b,x,z)}~\text{is strictly increasing in $b$.} \label{eq:implicit unobserved hetero}
\end{align}

However, we cannot formulate the null hypothesis as in \ref{eq: general observed hetero H0}, because $Z$ is unobservable to the econometrician.  
Instead, (\ref{eq:implicit unobserved hetero}) implies the following testable implications.  

\begin{lemma}\label{lemma: general unobserved hetero H0 transformation} 
Let $h(b,x,z)$ be a known function such that $0 <h(b,x,z)P'(b,x,z) \leq M<\infty$.  
Then (\ref{eq:implicit unobserved hetero}) implies that
\begin{align}
H_{0,x,z}': \nu(b_1,b_2, x, q)&=M(b_2,x,q)W(b_1,x,q)-M(b_1,x,q)W(b_2,x,q)\leq 0\notag \\
&~\text{for all $(b_1,b_2,x, q)\in\mathcal{L}_x$} \label{eq: general unobserved hetero H0'},
\end{align}
where 
\begin{align}
&M(b,x,q)= \int_{0}^1\int_{x}^{x+\frac{1}{q}}\int_{b}^{b+\frac{a}{q}} h(\tilde{b},\tilde{x},z)T'(\tilde{b},\tilde{x},z) d\tilde{b}d\tilde{x} dz, \notag\\
&
W(b,x,q)=\int_{0}^1\int_{x}^{x+\frac{1}{q}}\int_{b}^{b+\frac{a}{q}} h(\tilde{b},\tilde{x},z)P'(\tilde{b},\tilde{x},z) d\tilde{b}d\tilde{x}dz,~\text{and}\label{eq: general unobserved hetero M W}\\
&\mathcal{L}_x=\Big\{(b_1,b_2, x, q):
\Big(\frac{b_1-\underline{b}}{a}, \frac{b_2-\underline{b}}{a}, x\Big)
q\in(0,1,\ldots,q-1)^3, b_1>b_2~\text{for q=2,3,\ldots}  \Big\}.
\end{align}
\end{lemma}
The basic idea of Lemma \ref{lemma: general unobserved hetero H0 transformation} is to integrate out $Z$, so $Z$ will not show up in the expressions anymore.  In this case, we may use the same tests as in Section \ref{sec: observed hetero nonparametric}.  
The $M(b,x,q)$ and $W(b,x,q)$ for four examples are identical to those in Section \ref{sec: observed hetero nonparametric} and we omit the details. 
Note that the testable implications in these cases are the same as those in Section \ref{sec: observed hetero nonparametric}, so we can apply the same test in Section \ref{sec: observed hetero nonparametric} to test the models with unobserved heterogeneity and we omit the details.  The extension to the settings without observed covariates 
$X$ follows by a similar argument and we omit the details.

\section{Simulation Studies}\label{sec: simulation}

We conduct Monte-Carlo simulations to evaluate the finite sample performance of the proposed tests.  
To implement our test in practice, one has to choose several tuning parameters in advance. We first suggest the choices of these tuning parameters and then present the related Monte Carlo simulation results.

\begin{enumerate}

\item Instrumental functions: We opt for using a set of indicator functions of countable hypercubes. Define
\begin{align*}
\mathcal{L}=\{(b_1,b_2, x, q): (b_1,b_2,x )q\in\{0,1,\ldots,q-1\}^{2+d_x}, b_1>b_2~\text{for}~q=2,3,\ldots,q_1 \},
\end{align*}
where $q_1$ is a natural number and is chosen such that the expected sample size of the smallest cube is around 20, and $d_x$ is the dimension of (continuous) covariates. We also consider three alternative choices of $q_1$, each resulting in the expected sample size of the smallest cube $n_c=15$, $25$, and $30$.\footnote{Specifically, for a specific $n_c$, we set $q_1=\lfloor (NL/n_c)^{1/(1+d_x)}\rceil$, where $\lfloor c \rceil$ gives the integer nearest to $c$. In addition, for heterogeneous numbers of bidders, we have $q_{1,t}$ a function of $t$ and  $q_{1,t}=\lfloor (N_tL_t/n_c)^{1/(1+d_x)}\rceil$, where $t$ represents type $t$ auctions with $N_t$ number of bidders.} Our simulations show that the results are robust to various expected sample sizes.

\item {\bf $Q(b_1,b_2, x, q)$:}  The distribution $Q(b_1,b_2, x, q)$ assigns weight $\propto q^{-2}$ to each $q$, and for each $q$, $Q(\ell)$ assigns an equal weight to each instrumental function with the last element of $\ell$ equal to $q^{-1}$. Recall that for each $q$, there are $(q(q+1)/2)\cdot q^{d_x}$ instrumental functions with the last element of $(b_1,b_2, x, q))$ equal to $q$.
	
\item {\bf $\kappa_{j}$, $\beta_{j}$, $\epsilon$, $\eta$:} We set $\kappa_{j}=0.15\cdot\ln(j)$, $\beta_{j}=0.85\cdot\ln(j)/\ln\ln(j)$, $\epsilon=10^{-6}$, and $\eta=10^{-6}$ as suggested by \cite{HLS2019}. These choices are used in all the simulations that we report below and seem to perform well.  

\item All our Monte Carlo results are based on 1000 simulations. In each simulation, the critical value is approximated by 1000 bootstrap replications. The nominal size of the test is set at 10\%.

\end{enumerate}

We first consider the case without covariates. To study the size and power properties, the CDF of the observed bids is specified as
\begin{align}
G(b)=\left( \frac{b}{k-(k-1)b}\right)^{1/5},\label{eq:CDF}
\end{align}
and the associated quantile function of the true bid distribution is
\begin{align}
Q(\tau)=\frac{k \tau^{5}}{1+(k-1)\tau^5}.\label{eq:quantile}
\end{align}
We set $k=0.5$ for the size analysis and set $k=5$, $10$, $15$, and $20$ for the power analysis. Note that $k=5$ corresponds to the example that violates the monotonicity condition provided in \cite{GPV2000}. For all values of $k$, we set $N=2$ and $L=100$, 250, or 500. Figure \ref{fig:bidder_value} presents the quasi-inverse equilibrium strategy, $\xi(b)=b+{G(b)/(g(b)(N-1))}$
for different values of $k$. From this figure, we can observe that the monotonicity condition holds for $k=0.5$, but is violated for $k\geq5$. The figure also shows that it will be easier to detect the violation as $k$ increases.

\begin{figure}[t!]
\centering
\includegraphics[width=12cm,height=9cm]{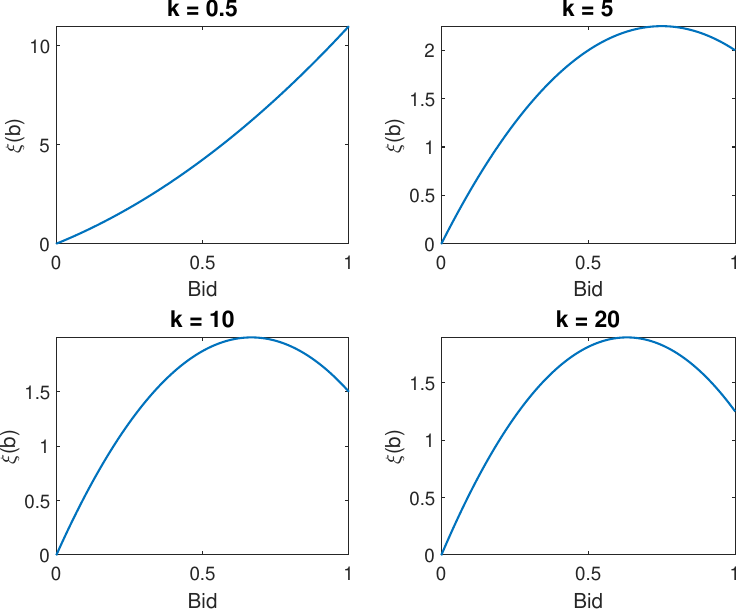}
\caption{Quasi-inverse equilibrium strategy, $\xi(b)$}\label{fig:bidder_value}
\end{figure}

Table \ref{table:DGP1} shows the rejection probabilities of our tests. The results show that the proposed test controls size well for $k=0.5$. Note that when $k=0.5$, the null hypothesis holds with strict inequality. Therefore, the size will converge to zero since every moment would hold with strict 
inequality. It follows that the test statistics will converge to zero and the critical value is bounded away from zero. For the power analysis, the results show that the power increases with sample size and with $k$ for $k\geq 5$, but the rejection probabilities are a bit less than the nominal size $0.1$ for $k=5$. Table \ref{table:DGP1} also shows that the choices of $q_1$ do not affect the test performance much.

\begin{table}[H]
\begin{small}
\begin{center}
\caption{Rejection probabilities for the case without covariates}\label{table:DGP1}
\begin{tabular}{cccccc}
\toprule
k&L&$n_c$=15&$n_c$=20&$n_c$=25&$n_c$=30\\
\midrule
0.5&100&0.001&0.004&0.001&0.002\\
0.5&250&0.003&0.003&0.001&0.003\\
0.5&500&0.001&0.002&0.001&0.001\\
\midrule
5&100&0.080&0.084&0.084&0.067\\
5&250&0.080&0.070&0.073&0.077\\
5&500&0.084&0.084&0.098&0.106\\
\midrule
10&100&0.311&0.310&0.297&0.297\\
10&250&0.504&0.519&0.470&0.464\\
10&500&0.737&0.754&0.741&0.741\\
\midrule
20&100&0.710&0.690&0.700&0.682\\
20&250&0.953&0.939&0.958&0.951\\
20&500&1.000&1.000&1.000&1.000\\
\bottomrule
\end{tabular}
\end{center}
\end{small}
\end{table}


We next consider the case with covariates. Let $X$ to be a uniform distribution on [0, 1]. To consider the test with covariates, the quantile function of the true bid distribution is specified as
\begin{align*}
Q(x,\tau)=\frac{k(x)\tau^5}{1+(k(x)-1)\tau^5}.
\end{align*}
We consider the following cases for the size and power analyses:
\begin{align*}
\text{Case 1: }&  k(x)=0.5+2x~~\text{for the size case},\\
\text{Case 2: }&  k(x)=5+5x~~\text{for the power case},\\
\text{Case 3: }&  k(x)=10+5x~~\text{for the power case},\\
\text{Case 4: }&  k(x)=15+5x~~\text{for the power case},\\
\text{Case 5: }&  k(x)=20+5x~~\text{for the power case}.
\end{align*}
For all cases, we set $N=2$ and $L=100$, 250, 500, or 1000. Table \ref{table:DGP2} shows the rejection probabilities of our tests for the case with covariates, and the results are consistent with our theoretical results. The proposed test controls size well in Case 1, and the rejection probabilities increase with the sample size and are greater than the nominal size $0.1$ in Cases 2-5. Like the results in Table \ref{table:DGP1}, we do not find significant differences for different choices $q_1$.

\begin{table}[H]
\begin{small}
\begin{center}
\caption{Rejection probabilities for the case with covariates}\label{table:DGP2}
\begin{tabular}{cccccc}
\toprule
Case&L&$n_c$=15&$n_c$=20&$n_c$=25&$n_c$=30\\
\midrule
1&100&0.013&0.019&0.014&0.021\\
1&250&0.002&0.004&0.004&0.007\\
1&500&0.003&0.000&0.002&0.001\\
1&1000&0.000&0.000&0.000&0.000\\
\midrule
2&100&0.119&0.128&0.137&0.136\\
2&250&0.123&0.143&0.144&0.127\\
2&500&0.175&0.172&0.155&0.165\\
2&1000&0.256&0.256&0.240&0.248\\
\midrule
3&100&0.234&0.237&0.225&0.231\\
3&250&0.412&0.413&0.391&0.351\\
3&500&0.648&0.643&0.622&0.620\\
3&1000&0.910&0.896&0.910&0.885\\
\midrule
4&100&0.364&0.313&0.288&0.325\\
4&250&0.663&0.638&0.596&0.606\\
4&500&0.927&0.920&0.910&0.902\\
4&1000&0.998&0.997&0.997&0.999\\
\midrule
5&100&0.463&0.412&0.409&0.419\\
5&250&0.797&0.794&0.760&0.760\\
5&500&0.978&0.985&0.972&0.985\\
5&1000&1.000&1.000&1.000&1.000\\
\bottomrule
\end{tabular}
\end{center}
\end{small}
\end{table}

\section{Empirical Application}\label{sec:data}

As discussed in \citet{BajariYe2003}, monotonicity of equilibrium bidding strategies is one of the sufficient and necessary conditions for the symmetric and competitive bidding equilibrium. While they try to test for efficient collusion by testing these conditions, they do not test for monotonicity due to the lack of monotonicity tests.\footnote{\citet{AryalGabrielli2013} propose a stochastic dominance test to test for collusion. To implement the test, they first use the tests of \citet{BajariYe2003} to identify potential colluding bidders.} Clearly, rejection of monotonicity can be likely an indication of potential collusive or irrational behavior. To demonstrate the effectiveness of our monotonicity test, we use the asphalt paving auction datasets studied in \cite{Aaltio_etal2025}, who analyze bidding behavior in state-level road-paving procurements in Sweden, Finland and California.

\subsection{Data}
The Nordic data are drawn from contracts procured by the Swedish Road Administration and the Finnish Transport Infrastructure Agency, covering 1993-2009 in Sweden and 1994-2019 in Finland. Both datasets contain all submitted bids in each tender, the identity of the winning bidder, paving area (m$^2$), and the region where the pavement project took place. 

Both Nordic markets experienced bid-rigging cartels during the late 1990s and early 2000s. To focus on bidding behavior during legally established cartel operation while avoiding contamination around detection and enforcement, we restrict attention to the court-documented conviction windows and exclude the dawn-raid years. Specifically, we use tenders from 1997–2000 in Sweden and 1995–2001 in Finland.\footnote{Both markets featured complementary bidding consistent with designated-winner coordination.} The California asphalt paving market is widely considered as competitive and there is no evidence of collusion during this period, and used in \cite{Aaltio_etal2025} as a control market. We use California asphalt paving contracts procured by the California Department of Transportation from 1999 to 2008.  

For the monotonicity test, we work with normalized bids to make outcomes comparable across contracts of different scales. In Sweden and Finland, we normalize bids by paving area (price per $m^2$). In the test that control for observed auction heterogeneity, we control for paving area. The California dataset contains similar information as the Nordic datasets, including all submitted bids, winner identity, and project region, except that the former contains the information on engineer's cost estimate but not paving area, while the latter contains the information on paving area but not engineer's cost estimate. In California, we normalize bids by engineer's cost estimate, which is also controlled for in the test that controls for observed auction heterogeneity. 

\subsection{Test results}

We first conduct the monotonicity test for the California dataset. Table \ref{table:Cal} reports the test statistics and associated p-values for auctions with 2 to 10 bidders and heterogeneous numbers of bidders. The tuning parameters are chosen as in Section \ref{sec: simulation}, and the test with heterogeneous numbers of bidders is discussed in Appendix \ref{sec: hetero numbers}. The left panel of Table \ref{table:Cal} reports the test results based on the nonparametric test without controlling for covariates, and the right panel of Table \ref{table:Cal} reports the test results based on the nonparametric test with covariates discussed in Section \ref{sec: observed hetero nonparametric}.  The results show that there is no case with strong evidence against the monotonicity null hypothesis, which is consistent with the literature where the California market is usually modeled as competitive and no cartel investigation was conducted during 1999 to 2008. 


We next apply the monotonicity tests to the Sweden and Finland datasets. Tables \ref{table:Swe} and \ref{table:Fin} report the test statistics and associated p-values for auctions with various numbers of bidders and heterogeneous numbers of bidders for Sweden and Finland, respectively. For Sweden, the p-values for the auction with 5 bidders and the cases associated with 5 bidders range from 0.03 to 0.08, which provide evidence against the monotonicity. For Finland, there is no case with strong evidence against the monotonicity null hypothesis. Interestingly, \cite{Aaltio_etal2025} find that the cartel had a larger impact on the distribution of bids in Finland than in Sweden, and they interpret this as that the Swedish cartel was more successful in mimicking competitive behavior, or it had more limited influence in scope. The evidence against monotonicity we find for the Swedish market shows the power of our test in detecting non-competitive behavior which other methods may not be able to detect.



\begin{table}[h]
\begin{small}
\begin{center}
\caption{Test statistics and p-values for California} \label{table:Cal}
\begin{tabular}{ccccccccc}
\toprule
&\multicolumn{4}{c}{without covariates} &\multicolumn{4}{c}{with covariates}\\
\cmidrule(lr){2-5}\cmidrule(lr){6-9}
&\multicolumn{2}{c}{$n_c$=15} &\multicolumn{2}{c}{$n_c$=20}
&\multicolumn{2}{c}{$n_c$=15} &\multicolumn{2}{c}{$n_c$=20} \\
 \cmidrule(lr){2-3}\cmidrule(lr){4-5}\cmidrule(lr){6-7}\cmidrule(lr){8-9}
$N_t$&stat&p-value&stat&p-value&stat&p-value&stat&p-value\\
\midrule
2&0.0318&0.3268&0.0233&0.3878&0.0001&0.6566&0.0001&0.6564\\
(2,3)&0.0414&0.6710&0.0322&0.7170&0.0001&0.9968&0.0001&0.9846\\
(2,3,4)&0.1306&0.7114&0.1210&0.7270&0.0122&0.8146&0.0111&0.7870\\
3&0.0096&0.6790&0.0089&0.6856&0.0000&1.0000&0.0000&1.0000\\
(3,4)&0.0987&0.6844&0.0977&0.6794&0.0121&0.7360&0.0110&0.7140\\
(3,4,5)&0.1165&0.8178&0.1153&0.8114&0.0128&0.8934&0.0115&0.8772\\
4&0.0892&0.5566&0.0888&0.5432&0.0121&0.5784&0.0110&0.5312\\
(4,5)&0.1069&0.7298&0.1064&0.7310&0.0128&0.7786&0.0115&0.7484\\
(4,5,6)&0.1512&0.8396&0.1506&0.8346&0.0172&0.8930&0.0153&0.8834\\
5&0.0178&0.6782&0.0176&0.6676&0.0007&0.9360&0.0005&0.9416\\
(5,6)&0.0620&0.7662&0.0619&0.7728&0.0051&0.9074&0.0043&0.9216\\
(5,6,7)&0.0626&0.8700&0.0624&0.8604&0.0051&0.9892&0.0043&0.9878\\
6&0.0442&0.6314&0.0442&0.6254&0.0044&0.7738&0.0038&0.7734\\
(6,7)&0.0449&0.7794&0.0448&0.7912&0.0044&0.9524&0.0038&0.9518\\
(6,7,8)&0.0489&0.8718&0.0470&0.8662&0.0044&0.9826&0.0038&0.9840\\
7&0.0006&0.8814&0.0006&0.8270&0.0000&1.0000&0.0000&1.0000\\
(7,8)&0.0047&0.7964&0.0028&0.8548&0.0000&1.0000&0.0000&1.0000\\
(7,8,9)&0.0562&0.7548&0.0508&0.7644&0.0000&1.0000&0.0000&1.0000\\
8&0.0041&0.5774&0.0022&0.6120&0.0000&1.0000&0.0000&1.0000\\
(8,9)&0.0555&0.5884&0.0502&0.6024&0.0000&1.0000&0.0000&1.0000\\
(8,9,10)&0.0586&0.7542&0.0530&0.7512&0.0000&1.0000&0.0000&1.0000\\
\bottomrule
\end{tabular}
\end{center}
Note: `Stat' denotes the test statistic, and the p-values of the tests are calculated based on 5000 bootstrap replications. The number of bidders $N_t$ and the number of auctions $L_t$ are: $\{(N_t,L_t):(2,204), (3,308), (4,276), (5,212), (6,163),(7,117),(8,58),(9,41),(10,33)\}$.
\end{small}
\end{table}

\begin{table}[h!]
\vspace{-1.5em}
\begin{small}
\begin{center}
\caption{Test statistics and p-values for Sweden} \label{table:Swe}
\begin{tabular}{ccccccccc}
\toprule
&\multicolumn{4}{c}{without covariates} &\multicolumn{4}{c}{with covariates}\\
\cmidrule(lr){2-5}\cmidrule(lr){6-9}
&\multicolumn{2}{c}{$n_c$=15} &\multicolumn{2}{c}{$n_c$=20}
&\multicolumn{2}{c}{$n_c$=15} &\multicolumn{2}{c}{$n_c$=20} \\
 \cmidrule(lr){2-3}\cmidrule(lr){4-5}\cmidrule(lr){6-7}\cmidrule(lr){8-9}
$N_t$&stat&p-value&stat&p-value&stat&p-value&stat&p-value\\
\midrule
3&0.1381&0.0744&0.0000&1.0000&0.0000&1.0000&0.0000&1.0000\\
(3,4)&0.2589&0.2034&0.1050&0.2434&0.0000&1.0000&0.0000&1.0000\\
(3,4,5)&0.9441&0.0508&0.7738&0.0636&0.0823&0.0698&0.0630&0.0766\\
4&0.1208&0.2340&0.1050&0.2316&0.0000&1.0000&0.0000&1.0000\\
(4,5)&0.8060&0.0616&0.7738&0.0636&0.0823&0.0688&0.0630&0.0874\\
(4,5,6)&0.8907&0.0530&0.8360&0.0484&0.0823&0.0590&0.0630&0.0768\\
5&0.6852&0.0438&0.6688&0.0378&0.0823&0.0836&0.0630&0.0866\\
(5,6)&0.7699&0.0346&0.7311&0.0352&0.0823&0.0838&0.0630&0.0786\\
(5,6,7)&1.4276&0.0722&1.3355&0.0796&0.1719&0.1198&0.0630&0.0842\\
6&0.0847&0.1866&0.0622&0.2058&0.0000&1.0000&0.0000&1.0000\\
(6,7)&0.7424&0.0982&0.6667&0.0786&0.0896&0.1372&0.0000&1.0000\\
\bottomrule
\end{tabular}
\end{center}
Note: `Stat' denotes the test statistic, and the p-values of the tests are calculated based on 5000 bootstrap replications. The number of bidders $N_t$ and the number of auctions $L_t$ are: $\{(N_t,L_t):(3,15), (4,58), (5,67), (6,28),(7,17)\}$.
\end{small}

\vspace{-1em}
\begin{small}
\begin{center}
\caption{Test statistics and p-values for Finland} \label{table:Fin}
\begin{tabular}{ccccccccc}
\toprule
&\multicolumn{4}{c}{without covariates} &\multicolumn{4}{c}{with covariates}\\
\cmidrule(lr){2-5}\cmidrule(lr){6-9}
&\multicolumn{2}{c}{$n_c$=15} &\multicolumn{2}{c}{$n_c$=20}
&\multicolumn{2}{c}{$n_c$=15} &\multicolumn{2}{c}{$n_c$=20} \\
 \cmidrule(lr){2-3}\cmidrule(lr){4-5}\cmidrule(lr){6-7}\cmidrule(lr){8-9}
$N_t$&stat&p-value&stat&p-value&stat&p-value&stat&p-value\\
\midrule
3&0.0000&1.0000&0.0000&1.0000&0.0000&1.0000&0.0000&1.0000\\
(3,4)&0.0000&1.0000&0.0000&1.0000&0.0000&1.0000&0.0000&1.0000\\
(3,4,5)&0.0474&0.7502&0.0293&0.7620&0.0374&0.3484&0.0000&1.0000\\
4&0.0000&1.0000&0.0000&1.0000&0.0000&1.0000&0.0000&1.0000\\
(4,5)&0.0474&0.7260&0.0293&0.7446&0.0374&0.3568&0.0000&1.0000\\
(4,5,6)&0.1595&0.6678&0.1217&0.6842&0.0374&0.3480&0.0000&1.0000\\
5&0.0474&0.6766&0.0293&0.7068&0.0374&0.3318&0.0000&1.0000\\
(5,6)&0.1595&0.6414&0.1217&0.6506&0.0374&0.3208&0.0000&1.0000\\
(5,6,7)&0.1597&0.8132&0.1217&0.8188&0.0374&0.6020&0.0000&1.0000\\
6&0.1121&0.4546&0.0924&0.4774&0.0000&1.0000&0.0000&1.0000\\
(6,7)&0.1124&0.7352&0.0924&0.7092&0.0000&1.0000&0.0000&1.0000\\
(6,7,8)&0.1124&0.8434&0.0924&0.8402&0.0000&1.0000&0.0000&1.0000\\
7&0.0002&0.9390&0.0000&1.0000&0.0000&1.0000&0.0000&1.0000\\
(7,8)&0.0002&0.9898&0.0000&1.0000&0.0000&1.0000&0.0000&1.0000\\
\bottomrule
\end{tabular}
\end{center}
Note: `Stat' denotes the test statistic, and the p-values of the tests are calculated based on 5000 bootstrap replications. The number of bidders $N_t$ and the number of auctions $L_t$ are: $\{(N_t,L_t):(3,15), (4,14), (5,41), (6,29),(7,15),(8,12)\}$.
\end{small}
\end{table}

\clearpage

\section{Conclusion} \label{sec: conclusion}
This paper develops a unified framework for testing whether Bayesian Nash equilibrium strategies are monotone in unobserved types in games of incomplete information. We first establish a rationalization result showing that, within the symmetric independent
private type paradigm, monotonicity of the equilibrium strategy in private types is equivalent to monotonicity of the quasi-inverse of the equilibrium strategy in actions; the quasi-inverse equilibrium strategy is a function of expected allocation
and expected payment rules, and can be identified from observed actions. As a result, our test reduces to testing monotonicity of the quasi-inverse equilibrium strategy, which can be reformulated as testing a countable collection of moment inequalities
that involve only unconditional expectations. The proposed testing procedure is easy to implement, and can accommodate game specific covariates and applies to both observed and unobserved game heterogeneity. The tests have good finite sample performance as demonstrated by the Monte-Carlo experiments. An application to asphalt paving auctions studied in \cite{Aaltio_etal2025} illustrates how the test can be used for cartel detection, and demonstrates the importance of testing for monotonicity and potential applications of our tests. For instance, our tests provide a key tool as a first step in detecting non-competitive behavior. 


\vspace{5em}

{\begin{small}
\baselineskip=12.75pt
\bibliographystyle{chicago}
\bibliography{test_auction.bib}
\end{small}}

\newpage
\setcounter{page}{1}
\appendix
\begin{center}
{\LARGE\bf Appendix}
\end{center}
\section{Details of the Test in Auction Games in Section \ref{subsec: auction}}
Given that our main simulation results and empirical studies focused on the auction case, we provide the details of the test and an implementation procedure for the auction case. Based on the identification results (\ref{eq: auction M}) and (\ref{eq: auction W}) in Lemma \ref{lemma: auction H0 transformation}, we can estimate $\nu(b_1,b_2,q)$ by 
\begin{align}
&\widehat{\nu}(b_1,b_2, q)=\widehat{M}(b_2,q)\widehat{W}(b_1,q)-\widehat{M}(b_1,q)\widehat{W}(b_2,q),~\text{where} \label{eq: auction nu estimator}\\ 
&\widehat{M}(b,q)=\frac{1}{NL}\sum_{i,\ell} \Big(B_{i,\ell} 1\Big( b \leq  B_{i,\ell}\leq b+\frac{a}{q} \Big)\notag\\
&~~~~~~~~~~~~~~~~~~~~~~+\frac{1}{N-1} \Big( 1\Big(B_{i,\ell}\leq b+\frac{a}{q} \Big) \Big(b+\frac{a}{q} - B_{i,\ell}\Big)  - 1(B_{i,\ell}\leq b ) (b - B_{i,\ell}) \Big) \Big),\notag\\
&\widehat{W}(b,q)=\frac{1}{NL}\sum_{i,\ell} 1\Big( b \leq  B_{i,\ell}\leq b+\frac{a}{q} \Big).\notag
\end{align}

Following \cite{HLS2019}, it is straightforward to see that  
\begin{align*}
\phi_{i,\ell,M}(b,q)&=\Big(B_{i,\ell} 1\Big( b \leq  B_{i,\ell}\leq b+\frac{a}{q} \Big)\\
&~~~~+\frac{1}{N-1} \Big( 1\Big(B_{i,\ell}\leq b+\frac{a}{q} \Big) \Big(b+\frac{a}{q} - B_{i,\ell}\Big)  - 1(B_{i,\ell}\leq b ) (b - B_{i,\ell}) \Big) \Big)-M(b,q),\\
\phi_{i,\ell,W}(b,q)&=1\Big( b \leq  B_{i,\ell}\leq b+\frac{a}{q} \Big)-W(b,q),\\
\phi_{i,\ell,\nu}(b_1,b_2,q)&=W(b_1,q)\phi_{i,\ell,M}(b_2,q)+M(b_2,q)\phi_{i,\ell,W}(b_1,q)\\
&~~~~~
-W(b_2,q)\phi_{i,\ell,M}(b_1,q)-M(b_1,q)\phi_{i,\ell,W}(b_2,q),
\end{align*}
are the influence functions for estimators $\widehat{M}(b,q)$, $\widehat{W}(b,q)$ and $\widehat{\nu}(b_1,b_2,q)$, respectively, 
so that 
\begin{align*}
\sqrt{NL}(\widehat{\nu}(\cdot,\cdot,\cdot)-{\nu}(\cdot,\cdot,\cdot))\Rightarrow
\Phi(\cdot,\cdot,\cdot)
\end{align*}
in which  $\Phi(\cdot,\cdot,\cdot)$ is a Gaussian process with covariance kernel being
$\kappa\big((b'_1,b'_2,q'),(b''_1,b''_2,q'')\big)=Cov({\phi}_{i,\ell,\nu}(b'_1,b'_2,q'),{\phi}_{i,\ell,\nu}(b''_1,b''_2,q''))$.

Define $\widehat{\phi}_{i,\ell,M}(b,q)$, $\widehat{\phi}_{i,\ell,W}(b,q)$ and $\widehat{\phi}_{i,\ell,\nu}(b_1,b_2,q)$ by replacing
$M(b,q)$ and $W(b,q)$ with $\widehat{M}(b,q)$ and $\widehat{W}(b,q)$ in ${\phi}_{i,\ell,M}(b,q)$, ${\phi}_{i,\ell,W}(b,q)$ and $\phi_{i,\ell,\nu}(b_1,b_2,q)$, respectively.
Define $\sigma^2_\nu(b_1,b_2,q)=\kappa\big((b_1,b_2,q),(b_1,b_2,q)\big)$ which is the asymptotic variance of
$\sqrt{NL}(\widehat{\nu}(b_1,b_2,q)-{\nu}(b_1,b_2,q))$ and its estimator as $\widehat{\sigma}^2_
\nu(b_1,b_2,q)$
\begin{align*}
\widehat{\sigma}^2_\nu(b_1,b_2,q)=\frac{1}{NL}\sum_{i,\ell} \widehat{\phi}^2_{i,\ell,\nu}(b_1,b_2,q),
\end{align*}
which can be shown to be uniformly consistent for $\sigma^2_\nu(b_1,b_2,q)$. 
Define
\begin{align}
\widehat{\sigma}^2_{\nu,\epsilon}(b_1,b_2,q)=\max\{\widehat{\sigma}^2_
\nu(b_1,b_2,q),\epsilon\cdot
\widehat{\sigma}^2_\nu(\underline{b}, (\underline{b}+\overline{b})/2,2)\}.\label{eq: auction estimated sigma epsilon}
\end{align}
The test statistic is defined as
\begin{align}
\widehat{T}_{NL,au}=\sum_{(b_1,b_2,q)\in\mathcal{L}} \max\Big\{ \sqrt{NL} \frac{\widehat{\nu}(b_1,b_2,q)}
{\widehat{\sigma}_{\nu,\epsilon}(b_1,b_2,q)},0\Big\}^2 Q(b_1,b_2,q). \label{eq: auction test statistics}
\end{align}

For the nonparametric bootstrap, we bootstrap auctions. That is, let $\{(B^*_{1,\ell},\ldots, B^*_{N\ell})
:\ell  \leq L \}$ be a
bootstrap sample in which $B^*_{i,\ell}=B_{i\ell^*}$ and $\{
\ell^*: \ell \leq L\}$ is an i.i.d.\ bootstrap sample
drawn from the empirical distribution of  $\{\ell : \ell\leq L\}$.
Define the bootstrap estimator for $\nu(b_1,b_2, q)$ as 
\begin{align}
&\widehat{\nu}^*(b_1,b_2, q)
=\widehat{M}^*(b_2,q)\widehat{W}^*(b_1,q)-\widehat{M}^*(b_1,q)\widehat{W}^*(b_2,q), \text{where}\label{eq: bootstrap auction nu estimator}\\ 
&\widehat{M}^*(b,q)=\frac{1}{NL}\sum_{i,\ell} \Big(B^*_{i,\ell} 1\Big( b \leq  B^*_{i,\ell}\leq b+\frac{a}{q} \Big)\notag\\
&~~~~~~~~~~~~~~~~~~~~~~+\frac{1}{N-1} \Big( 1\Big(B^*_{i,\ell}\leq b+\frac{a}{q} \Big) \Big(b+\frac{a}{q} - B^*_{i,\ell}\Big)  - 1(B^*_{i,\ell}\leq b ) (b - B^*_{i,\ell}) \Big) \Big),\notag\\
&\widehat{W}^*(b,q)=\frac{1}{NL}\sum_{i,\ell} 1\Big( b \leq  B^*_{i,\ell}\leq b+\frac{a}{q} \Big).\notag
\end{align}
The bootstrapped process is defined as $\Phi^*(\cdot,\cdot,\cdot)=\sqrt{NL}(\widehat{\nu}^*(\cdot,\cdot,\cdot)-\widehat{\nu}(\cdot,\cdot,\cdot)).$\footnote{Note that one can also use the bootstrap estimator for $\sigma^2_\nu(b_1,b_2,q)$ and we will introduce this in Section \ref{sec: observed hetero semiparametric}.}
Define the GMS
function $\psi_{NL}(b_1,b_2,q)$ as
\begin{align}
\psi_{NL}(b_1,b_2,q)=-\beta_{NL} \cdot 1\Big(\frac{\sqrt{NL}\widehat{\nu}(b_1,b_2,q)}{\widehat{\sigma}_{\nu,\epsilon}(b_1,b_2,q)}<-\kappa_{NL}\Big). \label{GMS-function}
\end{align}

For a significance level $\alpha<1/2$, define the bootstrapped critical value
$\hat{c}_{\eta,au}$ as
\begin{align*}
\hat{c}_{\eta,au}=\sup\left\{c\Big|P^*\Big(\sum_{(b_1,b_2,q)\in\mathcal{L}} \max\Big\{
\frac{{\Phi}^*(b_1,b_2,q)}
{\widehat{\sigma}_{\nu,\epsilon}(b_1,b_2,q)}+\psi_{NL}(b_1,b_2,q),0\Big\}^2 Q(b_1,b_2,q)\leq c\Big)\leq 1-\alpha+\eta\right\}+\eta.
\end{align*}

The decision rule is: ``Reject $H_{0,au}$ ($H_{0,au}'$) when $\widehat{T}_{NL,au}>\hat{c}_{\eta,au}$."

A rejection of $H_{0,au}$ ($H_{0,au}'$) would indicate that we have a strong evidence that the key equilibrium assumption, monotonicity of the equilibrium bids, does not hold in the data and the structural model is invalid.
The following theorem shows the asymptotic size control and the power against fixed alternatives of our test. 
\begin{thm} Suppose Assumptions \ref{assu: general DGP}, \ref{assu: general GMS} and \ref{assu: auction DGP} hold. Then,\\
(a) under $H_0$,
${\lim}_{L\rightarrow\infty}
P(\widehat{T}_{NL,au}>\hat{c}_{\eta,au})\leq \alpha$;\\
(b) under $H_1$,
${\lim}_{L\rightarrow\infty}
P(\widehat{T}_{NL,au}>\hat{c}_{\eta,au})= 1.$ \label{thm: auction test size and power}
\end{thm}

Theorem \ref{thm: auction test size and power} shows that our test has asymptotic size control under the null hypothesis and is consistent against fixed alternatives. We briefly summarize the implementation procedure of the test as follows. 

\medskip
\noindent{\bf Implementation Procedure for the Test under Homogeneous Auctions:}
\begin{enumerate}
\item
Estimate $\nu(b_1,b_2, q)$ by $\widehat{\nu}(b_1,b_2, q)=\widehat{M}(b_2,q)\widehat{W}(b_1,q)-\widehat{M}(b_1,q)\widehat{W}(b_2,q)$ according to (\ref{eq: auction nu estimator}).

\item
Compute $\widehat{\sigma}^2_{\nu,\epsilon}(b_1,b_2,q)$ according to  (\ref{eq: auction estimated sigma epsilon}).

\item
Compute the test statistic $\widehat{T}_{NL,au}$ according to (\ref{eq: auction test statistics}).

\item
Compute the GMS
function $\psi_{NL}(b_1,b_2,q)$ according to (\ref{GMS-function}). 

\item 
For $k=1,\ldots, K$ bootstrap sample, compute $\widehat{\nu}^{*,k}(b_1,b_2, q)$ and define
$\Phi^{*,k}_t(\cdot)=\sqrt{NL}(\widehat{\nu}^{*,k}(b_1,b_2, q)-\widehat{\nu}(b_1,b_2, q)).$

\item 
For a significance level $\alpha<1/2$, compute the bootstrapped critical value $\hat{c}_{\eta,au}$ as the $(1-\alpha+\eta)$-th quantile of
\begin{align*}
\Big\{\sum_{(b_1,b_2,q)\in\mathcal{L}} \max\Big\{
\frac{{\Phi}^{*,k}(b_1,b_2,q)}
{\widehat{\sigma}_{\nu,\epsilon}(b_1,b_2,q)}+\psi_{NL}(b_1,b_2,q),0\Big\}^2 Q(b_1,b_2,q)\Big)
:k=1,\ldots,K\Big\}
\end{align*}
plus $\eta$.

\item
Reject the null if $\widehat{T}_{S,au}>\hat{c}_{\eta,au}$.

\end{enumerate}

\noindent{\bf Proof of Theorem \ref{thm: auction test size and power}:}
Given that 
\begin{align*}
&\Big\{\Big(B_{i,\ell} 1\Big( b \leq  B_{i,\ell}\leq b+\frac{a}{q} \Big)+\frac{1}{N-1} \Big( 1\Big(B_{i,\ell}\leq b+\frac{a}{q} \Big) \Big(b+\frac{a}{q} - B_{i,\ell}\Big)  - 1(B_{i,\ell}\leq b ) (b - B_{i,\ell}) \Big) \Big)\\
&~~~~~~~~~~~~~~: b\in [\underline{b},\bar{b}], q=2,\ldots\Big\},\\
&\Big\{ 1\Big( b \leq  B_{i,\ell}\leq b+\frac{a}{q} \Big) : b\in [\underline{b},\bar{b}], q=2,\ldots\Big\},
\end{align*}
are both Vapnik-Chervonenkis (VC) classes of functions, we have
$\sqrt{NL}(\widehat{M}(\cdot,\cdot)-{M}(\cdot,\cdot))$ and $\sqrt{NL}(\widehat{W}(\cdot,\cdot)-{W}(\cdot,\cdot))$ weakly converge to Gaussian processes in the limit with influence functions $\phi_{i,\ell,M}(\cdot,\cdot)$ and $\phi_{i,\ell,W}(\cdot,\cdot)$ defined in Section \ref{subsec: auction}. 
Then by a similar argument of Lemma A.2 of \cite{HLS2019}, we have
$\sqrt{NL}(\widehat{\nu}(\cdot,\cdot,\cdot)-{\nu}(\cdot,\cdot,\cdot))$ weakly converging to a Gaussian process in the limit with influence functions $\phi_{i,\cdot,\nu}(\cdot,\cdot,\cdot)$.  By the same arguments for Lemma A.2 of \cite{HLS2019}, it is true that 
$|\widehat{\sigma}_\nu(\cdot,\cdot,\cdot)-{\sigma}_\nu(\cdot,\cdot,\cdot)|\stackrel{p}{\rightarrow} 0$ uniformly.  Define  
\begin{align*}
    {\sigma}^2_{\nu,\epsilon}(b_1,b_2,q)=\max\{{\sigma}^2_\nu(b_1,b_2,q),\epsilon\cdot
{\sigma}^2_\nu(\underline{b}, (\underline{b}+\overline{b})/2,2)\}.
\end{align*}
It is true that $|\widehat{\sigma}^{-1}_{\nu,\epsilon}(\cdot,\cdot,\cdot)-{\sigma}^{-1}_{\nu,\epsilon}(\cdot,\cdot,\cdot)|\stackrel{p}{\rightarrow} 0$ uniformly. Also, the bootstrapped process  $\Phi^*(\cdot,\cdot,\cdot)=\sqrt{NL}(\widehat{\nu}^*(\cdot,\cdot,\cdot)-\widehat{\nu}(\cdot,\cdot,\cdot))$ will weakly converges to the same limiting Gaussian process as $\sqrt{NL}(\widehat{\nu}(\cdot,\cdot,\cdot)-{\nu}(\cdot,\cdot,\cdot))$ conditional on the sample path with probability approaching 1.  Given all these results, Theorem \ref{thm: auction test size and power} can be proved by similar arguments for Theorem 4.1 and Theorem 5.1 of \cite{HLS2019} and we omit the details.~~~~$\Box$

\section{Testing with Heterogeneous Number of Agents} \label{sec: hetero numbers}
In this appendix, we discuss how to conduct tests when we allow for heterogeneous numbers of agents in each game.  For illustration and for simplicity, we focus on the auction case and it is straightforward to extend to other cases. 
Assume that there are $T$ different types of auctions with different numbers of bidders.  Let $N_t\geq 2$ denote the number of bidders for type $t$ auctions for $t=1,\ldots,T$ and $T$ is finite.  
Let $\xi_t(b)$ be the inverse of the BNE bidding strategy for type $t$ auctions, and $g_t(\cdot)$ and $G_t(\cdot)$ be the PDF and CDF of the bids in type $t$ auctions, respectively.
Note that the inverse of the BNE bidding strategy for type $t$ auctions can be written as
\begin{align}
\xi_t(b)= b+ \frac{1}{N_t-1}\frac{G_t(b)}{g_t(b)}. \label{eq: invese BNE strategy-t}
\end{align}
The null hypothesis in this case is given as:
\begin{align}
&H^t_{0,T}: \xi_t(b)= b+ \frac{1}{N_t-1}\frac{G_t(b)}{g_t(b)}~\text{is weakly increasing in $b$ for all $t=1,\ldots, T$.}  \label{eq: H0 t}
\end{align}

For each $t$, let $L_t$ denote the number of type $t$ auctions. Let $L=\sum_{t=1}^TL_t$ be the total number of games in the data.  

\begin{assumption}\label{assu: t numbers}
Assume that $L_t/L\rightarrow \delta_t>0$ for $t=1,\ldots,T$.
\end{assumption}
Assumption \ref{assu: t numbers} 
implies that the number of auctions of each type diverges to infinity at the same rate. 
Similar to Lemma \ref{lemma: auction H0 transformation},  
$H_{0,T}$ in (\ref{eq: H0 t}) is equivalent to
\begin{align}
H_{0,T}': &~\nu_t(b_1,b_2, q)=M_t(b_2,q)W_t(b_1,q)-M_t(b_1,q)W_t(b_2,q)\leq 0\notag\\
&~~~~~\text{for all $t=1,\ldots, T$ and for all $(b_1,b_2, q)\in\mathcal{L}$} \label{eq: H0' t},
\end{align}
where 
\begin{align*}
&M_t(b,q)= \Big( E\Big[B_{i_t,\ell_t} 1\Big( b \leq  B_{i_t,\ell_t}\leq b+\frac{a}{q} \Big)\Big]\notag\\
&~~~~~~~~~~~~~+
\frac{1}{N_t-1} E\Big[ 1\Big(B_{i_t,\ell_t}\leq b+\frac{a}{q} \Big) \Big(b+\frac{a}{q} - B_{i_t,\ell_t}\Big)
 - 1(B_{i,\ell}\leq b ) (b - B_{i_t,\ell_t}) \Big]\Big), \\
&W_t(b,q)=E\Big[1\Big( b \leq  B_{i_t,\ell_t}\leq b+\frac{a}{q} \Big) \Big].
\end{align*}
We estimate $\nu_t(b_1,b_2, q)$ by $\widehat{\nu}_t(b_1,b_2, q)=\widehat{M}_t(b_2,q)\widehat{W}_t(b_1,q)-\widehat{M}_t(b_1,q)\widehat{W}_t(b_2,q)$ in which
\begin{align*}
&\widehat{M}_t(b,q)=\frac{1}{N_tL_t}\sum_{i_t,\ell_t} \Big(B_{i_t,\ell_t} 1\Big( b \leq  B_{i_t,\ell_t}\leq b+\frac{a}{q} \Big)\\
&~~~~~~~~~~~~~~~~~~~~~~+\frac{1}{N_t-1} \Big( 1\Big(B_{i_t,\ell_t}\leq b+\frac{a}{q} \Big) \Big(b+\frac{a}{q} - B_{i_t,\ell_t}\Big)  - 1(B_{i_t,\ell_t}\leq b ) (b - B_{i_t,\ell_t}) \Big) \Big),\\
&\widehat{W}_t(b,q)=\frac{1}{N_tL_t}\sum_{i_t,\ell_t} 1\Big( b \leq  B_{i_t,\ell_t}\leq b+\frac{a}{q} \Big),
\end{align*}
where $i_t$ and $\ell_t$ are the indexes for those bidders in type $t$ auctions.
Let
\begin{align*}
&\widehat{\phi}_{i_t,\ell_t,M_t}(b,q)=B_{i_t,\ell_t} 1\Big( b \leq  B_{i_t,\ell_t}\leq b+\frac{a}{q} \Big)\\
&~~~~~~~~~~~+\frac{1}{N_t-1} \Big( 1\Big(B_{i_t,\ell_t}\leq b+\frac{a}{q} \Big) \Big(b+\frac{a}{q} - B_{i_t,\ell_t}\Big)  - 1(B_{i_t,\ell_t}\leq b ) (b - B_{i_t,\ell_t}) \Big)-\widehat{M}_t(b,q),\\
&\widehat{\phi}_{i_t,\ell_t,W_t}(b,q)=1\Big( b \leq  B_{i_t,\ell_t}\leq b+\frac{a}{q} \Big)-\widehat{W}_t(b,q),\\
&\widehat{\phi}_{i_t,\ell_t,\nu_t}(b_1,b_2,q)=\widehat{W}_t(b_1,q)\widehat{\phi}_{i_t,\ell_t,M_t}(b_2,q)+\widehat{M}_t(b_2,q)\widehat{\phi}_{i_t,\ell_t,W_t}(b_1,q)\\
&~~~~~~~~~~~~~
-\widehat{W}_t(b_2,q)\widehat{\phi}_{i_t,\ell_t,M_t}(b_1,q)-\widehat{M}_t(b_1,q)\widehat{\phi}_{i_t,\ell_t,W_t}(b_2,q),
\end{align*}
which are the estimated influence functions for estimators $\widehat{M}_t(b,q)$, $\widehat{W}_t(b,q)$ and $\widehat{\nu}_t(b_1,b_2,q)$, respectively.
Define $\sigma^2_{\nu_t}(b_1,b_2,q)=\kappa_t\big((b_1,b_2,q),(b_1,b_2,q)\big)$ which is the asymptotic variance of
$\sqrt{N_tL_t}(\widehat{\nu}_t(b_1,b_2,q)-{\nu}_t(b_1,b_2,q))$ and its estimator as $\widehat{\sigma}^2_
{\nu_t}(b_1,b_2,q)$
\begin{align*}
\widehat{\sigma}^2_{\nu_t}(b_1,b_2,q)=\frac{1}{N_tL_t}\sum_{i_t,\ell_t} \widehat{\phi}^2_{i_t,\ell_t,\nu_l}(b_1,b_2,q).
\end{align*}
For some small $\epsilon>0$, define
\begin{align*}
\widehat{\sigma}^2_{\nu_t,\epsilon}(b_1,b_2,q)=\max\{\widehat{\sigma}^2_
{\nu_t}(b_1,b_2,q),\epsilon\cdot
\widehat{\sigma}^2_{\nu_t}(\underline{b}, (\underline{b}+\overline{b})/2,2)\}.
\end{align*}
Define our test statistic as
\begin{align*}
\widehat{T}_{L,T}=\sum_{t=1}^T\sum_{(b_1,b_2,q)\in\mathcal{L}} \max\Big\{ \sqrt{N_tL_t} \frac{\widehat{\nu}_t(b_1,b_2,q)}
{\widehat{\sigma}_{\nu_t,\epsilon}(b_1,b_2,q)},0\Big\}^2 Q(b_1,b_2,q).
\end{align*}

For the nonparametric bootstrap, we bootstrap auctions within each type. To be specific, for each $t$, let $\{(\{B^*_{1,\ell_t},\ldots, B^*_{N_t\ell_t}\})
:\ell_t  \leq L_t \}$ be a
bootstrap subsample in which $B^*_{i_t,\ell_t}=B_{i_t\ell^*_t}$ and $\{
\ell^*_t: \ell_t \leq L_t\}$ is an i.i.d.\ bootstrap sample
drawn from the empirical distribution of
$\{\ell_t : \ell_t\leq L_t\}$.
Then we collect all $T$ bootstrap subsamples together.  By doing this, the total number of bidders in each bootstrap sample remains as $\sum_{t=1}^TN_tL_t$.
Define the bootstrap estimator for $\nu_t(b_1,b_2, q)$ as $\widehat{\nu}^*_t(b_1,b_2, q)
=\widehat{M}^*_t(b_2,q)\widehat{W}^*_t(b_1,q)-\widehat{M}^*_t(b_1,q)\widehat{W}^*_t(b_2,q)$ in which
\begin{align*}
	&\widehat{M}^*_t(b,q)=\frac{1}{N_tL_t}\sum_{i_t,\ell_t} \Big(B^*_{i_t,\ell_t} 1\Big( b \leq  B^*_{i_t,\ell_t}\leq b+\frac{a}{q} \Big)\\
	&~~~~~~~~~~~~~~~~~~~~~~+\frac{1}{N_t-1} \Big( 1\Big(B^*_{i_t,\ell_t}\leq b+\frac{a}{q} \Big) \Big(b+\frac{a}{q} - B^*_{i_t,\ell_t}\Big)  - 1(B^*_{i_t,\ell_t}\leq b ) (b - B^*_{i_t,\ell_t}) \Big) \Big),\\
	&\widehat{W}^*_t(b,q)=\frac{1}{N_tL_t}\sum_{i_t,\ell_t} 1\Big( b \leq  B^*_{i_t,\ell_t}\leq b+\frac{a}{q} \Big).
\end{align*}
The bootstrapped process is defined as $\Phi^*_t(\cdot)=\sqrt{N_tL_t}(\widehat{\nu}^*_t(b_1,b_2, q)-\widehat{\nu}_t(b_1,b_2, q)).$
Define the GMS
function $\psi_{NL}(t,b_1,b_2,q)$ as
\begin{align*}
	\psi_{N_tL_t}(t,b_1,b_2,q)=-\beta_{N_tL_t} \cdot 1\Big(\frac{\sqrt{N_tL_t}\widehat{\nu}_t(b_1,b_2,q)}{\widehat{\sigma}_{\nu_t,\epsilon}(b_1,b_2,q)}<-\kappa_{N_tL_t}\Big). 
\end{align*}
For a significance level $\alpha<1/2$, define the bootstrapped critical value $\hat{c}_{\eta,T}$ as
	\begin{align*}
		\hat{c}_{\eta,T}=\sup\Big\{c\Big|P^*&\Big(\sum_{t=1}^T\sum_{(b_1,b_2,q)\in\mathcal{L}} \max\Big\{
		\frac{{\Phi}^*_t(b_1,b_2,q)}
		{\widehat{\sigma}_{\nu_t,\epsilon}(b_1,b_2,q)}+\psi_{N_tL_t}(t,b_1,b_2,q),0\Big\}^2 Q(b_1,b_2,q)\leq c\Big)\\
		&~~~~~~~~~\leq 1-\alpha+\eta\Big\}+\eta.
	\end{align*}
Then we reject $H'_{0,t}$ when $\widehat{T}_{L,T}>\hat{c}_{\eta,T}$.  
We summarize the size and power result of the test.

\begin{thm} Suppose that the unobserved the bidders' valuations $\{V_{i,\ell_t}: i=1,\ldots, N_t,~\ell_t=1,\ldots, L_t \}$ satisfy Assumption \ref{assu: auction DGP} for each $t=1,\ldots,T$.  Suppose that Assumptions \ref{assu: general GMS} and \ref{assu: t numbers} hold. 
Then,\\
(a) under $H_0$,
${\lim}_{S\rightarrow\infty}
P(\widehat{T}_{L,T}>\hat{c}_{\eta,T})\leq \alpha$;\\
(b) under $H_1$,
${\lim}_{S\rightarrow\infty}
P(\widehat{T}_{L,T}>\hat{c}_{\eta,T})= 1.$ \label{thm: heterogeneous number of bidders test}
\end{thm}

Note that, for the test with different numbers of bidders, technically we implement a joint test for $T$ types of auctions.

\begin{remarkapp}
In the main text, we also consider nonparametric test with covariates and semiparametric test. For these two cases, we can extend these tests to allow for heterogeneous numbers of bidders in a similar way.  Therefore, we omit the details.  
\end{remarkapp}

\section{Observed and Unobserved Game Heterogeneity: Semiparametric Test}
\subsection{Observed Game Heterogeneity: Semiparametric Test} \label{sec: observed hetero semiparametric}
In Section \ref{sec: observed hetero nonparametric}, we consider a nonparametric test for games of incomplete information covariates.  While it is general, it may suffer from the curse of dimensionality as is usually the case with nonparametric estimation/testing, when the number of (continuous) covariates is large.  In this section, we follow \cite{HHS2003}, \cite{MMSX2021} and \cite{CHS2020} and impose a semiparametric model to homogenize the actions. Specifically, let $\Gamma(x)$ be a strictly positive function of covariates. The rest of the framework follows Section \ref{sec: observed hetero nonparametric}.

\begin{assumption} \label{assu: homogenize}
Assume that (i) $\{(V_{1,\ell}=\Gamma(X_{\ell})u_{1,\ell}, V_{2,\ell}=\Gamma(X_{\ell})u_{2,\ell},\ldots, V_{N,\ell}=\Gamma(X_{\ell})u_{N,\ell}, X_{\ell}): \ell=1,\ldots, L\}$ are i.i.d.\  random vectors;
(ii) $\{u_{i,\ell}: i=1,\ldots, N,~\ell=1,\ldots, L \}$ are i.i.d.\ random variables with PDF $f_u(\cdot)$ and CDF $F_u(\cdot)$ and 
are independent of  
$X_{\ell}$.
(iii)
$f_u(\cdot)$ is strictly positive and bounded away from zero on its support, a compact
interval $[\underline{u}, \overline{u}] \subseteq R$ , and is twice continuously differentiable on $(\underline{u}, \overline{u})$.
\end{assumption}

\begin{assumption}\label{assu: D C homogeneous}
Assume that $D_{i,\ell}(B_{1,\ell},\ldots, B_{N,\ell})$ is a homogeneous function
of order zero and $T_{i,\ell}(B_{1,\ell},\ldots, B_{N,\ell})$ is a homogeneous function
of order one in that for any $r>0$
\begin{align*}
&D_{i,\ell}(rB_{1,\ell},\ldots, rB_{N,\ell})=D_{i,\ell}(B_{1,\ell},\ldots, B_{N,\ell}),\\
&T_{i,\ell}(rB_{1,\ell},\ldots, rB_{N,\ell})= r T_{i,\ell}(B_{1,\ell},\ldots, B_{N,\ell}).
\end{align*}
\end{assumption}
Let $S_v(v|x)$ denote the BNE strategy.
Let $B^u_{i}$ denote the actions when agent $i$'s value is $u_{i}$.   
We claim that $S_v(v|x) =\Gamma(x)s_u(\Gamma^{-1}(x) v )$ for all $x\in\mathcal{X}$.  
That is, if Assumption \ref{assu: D C homogeneous} holds, we want to show that ``$b^u\in s(u)$" $\Leftrightarrow$ ``$b=\Gamma(x)b^u\in s_v(v|x)= s_v(\Gamma(x)u|x)$".

Let $P^u_{i,\ell}(b^u)$ and $T^u_{i,\ell}(b^u)$ be the expected allocation and transfer that agent $i$ will attain if agent $i$ plays $b^u$.  
To show ``$\Rightarrow$" direction, 
note that if $b^u\in s(u)$, we have
\begin{align*}
&~u_{i,\ell}P^u_{i,\ell}(b^u)-T^u_{i,\ell}(b^u) \geq 
u_{i,\ell}P^u_{i,\ell}(b^{u\prime})-T^u_{i,\ell}(b^{u\prime})~~\text{for all $b^{u\prime}$}\\
\Leftrightarrow &~ 
\Gamma(x)u_{i,\ell}P^u_{i,\ell}(b^u)-\Gamma(x)T^u_{i,\ell}(b^u) \geq 
\Gamma(x)u_{i,\ell}P^u_{i,\ell}(b^{u\prime})-\Gamma(x)T^u_{i,\ell}(b^{u\prime})~~\text{for all $b^{u\prime}$}\\
\Leftrightarrow  &~ 
V_{i,\ell}P^u_{i,\ell}(\Gamma(x)b^u)-T^u_{i,\ell}(\Gamma(x)b^u) \geq 
V_{i,\ell}P^u_{i,\ell}(\Gamma(x)b^{u\prime})-T^u_{i,\ell}(\Gamma(x)b^{u\prime})~~\text{for all $b^{u\prime}$}
\\
\Leftrightarrow &~ 
V_{i,\ell}P_{i,\ell}(b|x)-T_{i,\ell}(b|x) \geq 
V_{i,\ell}P_{i,\ell}(b'|x)-T_{i,\ell}(b'|x)
~~\text{for all $b^{\prime}$},
\end{align*}
where the first $\Leftrightarrow$ holds by rescaling everything by $\Gamma(x)$, the second one holds by the fact that $V_{i,\ell}=\Gamma(x)u_{i,\ell}$ and $P^u_{i,\ell}(\cdot)$ and $T^u_{i,\ell}(\cdot)$ are homogeneous functions of degree zero and one, respectively.  The last $\Leftrightarrow$ holds because for agent $\ell$, if other players have $\Gamma(x)b^u\in s_v(v|x)$ being their best response functions, then $P_{i,\ell}(b|x)=P^u_{i,\ell}(\Gamma(x)b^{u\prime})$ and $T_{i,\ell}(b|x)=T^u_{i,\ell}(\Gamma(x)b^{u\prime})$. It follows that $b=\Gamma(x)b^u$ is agent $i$'s best response, too. That is, $b=\Gamma(x)b^u\in s_v(v|x)= s_v(\Gamma(x)u|x)$.   The same arguments apply to the ``$\Leftarrow$" direction of the claim.   

This result implies that 
\begin{align}
&~\xi(b,x) \;\equiv\; \frac{T'(b,x)}{P'(b,x)} ~\text{is strictly increasing in $b$ for all $x\in\mathcal{X}$.} \notag\\
\Leftrightarrow
&~\xi(b^u) \;\equiv\; \frac{T^{u\prime }(b^u)}{P^{u\prime }(b^u)} ~\text{is strictly increasing in $b^u$.} \label{eq: semi equivalence}
\end{align}
(\ref{eq: semi equivalence}) further implies that 
we can re-formulate the null hypothesis  $H_{0,x}$ in (\ref{eq: general observed hetero H0}) as
\begin{align}\label{eq: null semi test}
H_{0,\text{semi}}:~\xi(b^u) \;\equiv\; \frac{T^{u\prime }(b^u)}{P^{u\prime }(b^u)} ~\text{is weakly increasing in $b^u$.} 
\end{align}
Therefore, then $H_{0,\text{semi}}$ in (\ref{eq: null semi test}) is equivalent to
\begin{align*}
&H_{0,\text{semi}}': \nu(b^u_1,b^u_2, q)=M(b^u_2,q)W(b^u_1,q)-M(b^u_1,q)W(b^u_2,q)\leq 0~\text{for all $(b^u_1,b^u_2, q)\in\mathcal{L}$},
\end{align*}
where 
\begin{align*}
&M(b^u,q)= \int_{b^u}^{b^u+\frac{a}{q}} h(\tilde{b})T'(\tilde{b}) d\tilde{b}, 
~~~
W(b^u,q)=\int_{b^u}^{b^u+\frac{a}{q}} h(\tilde{b})P'(\tilde{b}) d\tilde{b},~\text{and}\\
&\mathcal{L}=\Big\{(b^u_1,b^u_2, q): \Big(\frac{b^u_1-\underline{b}^u}{a}, \frac{b^u_2-\underline{b}^u}{a}\Big)\cdot q\in(0,1,\ldots,q-1)^2, b^u_1>b^u_2~\text{for q=2,3,\ldots}  \Big\}.
\end{align*}

We note that Assumption \ref{assu: D C homogeneous} is crucial for the homogenization method to work. The first-price sealed-bid auction model and the linear-cost Tullock benchmark satisfy the stated payoff-homogeneity restrictions; applying homogenization also requires the full feasible action sets to scale accordingly. The nonlinear effort cost in Section \ref{subsec: contest model} is not homogeneous of degree one, so this extension does not generally apply to the revised contest specification. 
If $\Gamma(x)$ is known, then we obtain   
the rescaled actions, $B_{i,\ell}^u=\Gamma^{-1}(X_\ell)B_{i,\ell}$ and apply the test with the rescaled actions in Section \ref{sec: proposed test} to test $H_{0,\text{semi}}$ in (\ref{eq: null semi test}).  If $\Gamma(x)$ follows a parametric specification $\Gamma(x,\theta)$ so that $\theta$ is identified and a consistent estimator $\widehat{\theta}$ is available, then we can apply a similar test in Section \ref{sec: proposed test} after accounting for the estimation effect of $\widehat{\theta}$ in the test.  Next, we present a test in auction games and show its validity under the regularity conditions. 

\subsection{A Semiparametric Test for Auction Games with Observerd Heterogeneity}
In the auction model, the main idea of this semiparametric approach via homogenization of bids is to impose a parametric specification on $\Gamma(x;\theta)$ up to a scalar. Then once we parametrically estimate $\widehat{\theta}$, we can use the estimated factor $\Gamma(x;\widehat{\theta})$ as the rescaling factor.  Then we can apply the test developed in Section \ref{subsec: auction} with $\widehat{B}^u_{i,\ell}=B_{i,\ell}/\Gamma(X_{\ell};\widehat{\theta})$ after we account for the estimation effect of $\widehat{\theta}$.  

Following \cite{HHS2003} and \cite{MMSX2021}, we assume that $\Gamma(X_{\ell};\theta)=\exp(X_{\ell}\theta)$ and that 
\begin{align}
	\log(B_{i,\ell})=X_{\ell}\theta+\log(B^u_{i,\ell}),~\text{with }E[\log(B^u_{i,\ell})]=0,\label{eq: OLS model}
\end{align}
where, the first element of $X_{\ell}$ is a constant term. Then we can estimate $\theta$ consistently by an ordinary least squares (OLS) estimator. 
We get $\{\widehat{B}^u_{i,\ell}=\exp(-X_{\ell}\widehat{\theta}) B_{i,\ell}:~i=1,\ldots, N,~\ell=1,\ldots, L \}$.  Then we can treat $\widehat{B}^u_{i,\ell}$ as the observed bids and construct the test statistic as in Section \ref{subsec: auction}.  However, to construct the critical value, we have to account for the estimation effect of $\widehat{\theta}$ and the estimation effect makes the influence function representations for the estimated moments too complicated.  
Hence, to avoid estimating the asymptotic standard errors of the moments from the estimated influence functions as in previous cases, we suggest to use the bootstrap standard error estimators instead.   
We briefly summarize the implementation procedure of the test as follows.

\medskip
\noindent {\bf Implementation Procedure for Semiparametric Test under Observed Auction Heterogeneity:}
\begin{enumerate}
\item
Estimate $\theta$ by
$\widehat{\theta}$ which is given as
\begin{align*}
\widehat{\theta}=\Big(\frac{1}{NL}\sum_{i,\ell} X_{\ell}' X_{\ell}\Big)^{-1} \Big(\frac{1}{NL}\sum_{i,\ell} X_{\ell}' \log(B_{i,\ell})\Big),
\end{align*}
and get $\widehat{B}^u_{i,\ell}=\exp(-X_{\ell}\widehat{\theta}) B_{i,\ell}.$

\item
Estimate $\nu(b_1,b_2, q,\theta)$ by $\widehat{\nu}(b_1,b_2, q,\widehat{\theta})=\widehat{M}(b_2,q,\widehat{\theta})\widehat{W}(b_1,q,\widehat{\theta})-\widehat{M}(b_1,q,\widehat{\theta})\widehat{W}(b_2,q,\widehat{\theta})$ in which
\begin{align*}
	&\widehat{M}(b,q,\widehat{\theta})=\frac{1}{NL}\sum_{i,\ell} \Big(\widehat{B}^u_{i,\ell} 1\Big( b \leq  \widehat{B}^u_{i,\ell}\leq b+\frac{a}{q} \Big)\\
	&~~~~~~~~~~~~~~~~~~~~~~+\frac{1}{N-1} \Big( 1\Big(\widehat{B}^u_{i,\ell}\leq b+\frac{a}{q} \Big) \Big(b+\frac{a}{q} - \widehat{B}^u_{i,\ell}\Big)  - 1(\widehat{B}^u_{i,\ell}\leq b ) (b - \widehat{B}^u_{i,\ell}) \Big) \Big),\\
	&\widehat{W}(b,q,\widehat{\theta})=\frac{1}{NL}\sum_{i,\ell} 1\Big( b \leq  \widehat{B}^u_{i,\ell}\leq b+\frac{a}{q} \Big).
\end{align*}

\item
For $k=1,\ldots, K$ bootstrap sample, compute $\widehat{\theta}^{*,k}$, the bootstrap estimator for $\theta$, and obtain $\widehat{u}^{*,k}_{i,\ell}=\exp(-X^{*,k}_{i,\ell}\widehat{\theta}^{*,k}) B^{*,k}_{i,\ell}$  and  $\widehat{\nu}^{*,k}(b_1,b_2, q,\widehat{\theta}^{*,k})$.

\item
Compute the bootstrap standard error estimate for $\sigma_{\nu}(b_1,b_2,q,\widehat{\theta})$ as
\begin{align*}
	&\widehat{\sigma}^*_{\nu}(b_1,b_2,q,\widehat{\theta})=\Big(\frac{1}{K} \sum_{k=1}^K NL(\widehat{\nu}^{*,k}(b_1,b_2, q,\widehat{\theta}^{*,k}) -\bar{\nu}^*(b_1,b_2, q,\widehat{\theta}))^2\Big)^{1/2},~\text{where}\\
	&\bar{\nu}^*(b_1,b_2, q ,\widehat{\theta})=\frac{1}{K} \sum_{k=1}^K \widehat{\nu}^{*,k}(b_1,b_2, q,\widehat{\theta}^{*,k}),
\end{align*}
and define
$	\widehat{\sigma}^*_{\nu,\epsilon}(b_1,b_2,q,\widehat{\theta})=\max\{\widehat{\sigma}^*_{\nu}(b_1,b_2,q,\widehat{\theta}),\sqrt{\epsilon}\cdot\widehat{\sigma}^*_{\nu}(\underline{b}, (\underline{b}+\overline{b})/2,2,\widehat{\theta})\}.$
	
\item
Compute the test statistic as
\begin{align*}
	\widehat{T}_{S,se}=\sum_{(b_1,b_2,q)\in\mathcal{L}} \max\Big\{ \sqrt{NL} \frac{\widehat{\nu}(b_1,b_2,q,\widehat{\theta})}
	{\widehat{\sigma}^*_{\nu,\epsilon}(b_1,b_2,q,\widehat{\theta})},0\Big\}^2 Q(b_1,b_2,q).
\end{align*}

\item
Compute the GMS
function $\psi_{NL}(t,b_1,b_2,q)$ as
\begin{align*}
	\psi_{NL}(t,b_1,b_2,q)=-\beta_{NL} \cdot 1\Big(\frac{\sqrt{NL}\widehat{\nu}(b_1,b_2,q,\widehat{\theta})}{\widehat{\sigma}^*_{\nu,\epsilon}(b_1,b_2,q,\widehat{\theta})}<-\kappa_{NL}\Big).
\end{align*}

\item 
Let $\Phi^{*,k}(b_1,b_2, q)=\sqrt{NL}(\widehat{\nu}^{*,k}(b_1,b_2, q,\widehat{\theta}^{*,k})-\widehat{\nu}(b_1,b_2, q,\widehat{\theta})).$

\item 
For a significance level $\alpha<1/2$, compute $\hat{c}$ as the 
$(1-\alpha+\eta)$-th quantile of
\begin{align*}
\Big\{ \sum_{(b_1,b_2,q)\in\mathcal{L}} \max\Big\{
\frac{{\Phi}^{*,k}(b_1,b_2,q)}
{\widehat{\sigma}^*_{\nu,\epsilon}(b_1,b_2,q,\widehat{\theta})}+\psi_{NL}(t,b_1,b_2,q),0\Big\}^2 Q(b_1,b_2,q)\Big):k=1,\ldots,K\Big\}
\end{align*}
and define the bootstrapped critical value $\hat{c}_{\eta,se}=\hat{c}+\eta$.

\item
Reject the null if $\widehat{T}_{NL,se}>\hat{c}_{\eta,se}$.

\end{enumerate}

We summarize the asymptotic properties of the semiparametric test in Appendix C.3. Note that the bootstrap estimator for the asymptotic standard error in Step 5 can be applied to other cases as well.  However, it is more time-consuming to compute bootstrap standard error estimator than to compute the standard error estimator from the estimated influence functions.  
Also, it is straightforward to allow for heterogeneous number of bidders as in Section \ref{sec: hetero numbers} and we omit the details.

\subsection{Asymptotics of Semiparametric Test in Section \ref{sec: observed hetero semiparametric}}
This section provides additional assumptions and asymptotic results of the semiparametric test developed in Section \ref{sec: observed hetero semiparametric}.  

We impose the assumption on the first stage estimator, $\widehat{\theta}$. 
\begin{assumption}\label{assu: theta-est} Assume that (i) $\widehat\theta$ is an M-estimator of $\theta$ so that
\begin{align*}
\widehat{\theta}\equiv \arg\max_{\theta\in \Theta} \frac{1}{NL} \sum_{i,\ell} \zeta(B_{i,\ell},X_{\ell},\theta),
\end{align*}
for some $\zeta$ function and $\{\zeta(B_{i,\ell},X_{\ell},\theta):~\theta\in\Theta \}$ is a Vapnik-Chervonenkis (VC) class of functions.  

(ii) There is a point $\theta_0$ in the interior of the compact parameter space $\Theta$ so that
\begin{align*}
\sqrt{n}(\widehat{\theta}-\theta_0)=\frac{1}{\sqrt{NL}}\sum_{i,\ell}\psi_\theta(B_{i,\ell},X_{\ell},\theta_0)+o_p(1),
\end{align*}
where $\psi_\theta$ is a given function of $B_{i,\ell}$ and $X_{\ell}$ with $E[\psi_\theta(Y,X, \theta_0)]=0$ and $E\|\psi_\theta(Y,X, \theta_0)\|^{2+\epsilon}<\infty$ for some $\epsilon>0$. 

(iii)
The first stage bootstrap estimator $\widehat{\theta}^*$ satisfies that  
\begin{align*}
\sqrt{n}(\widehat{\theta}^*-\widehat{\theta})=\frac{1}{\sqrt{NL}}\sum_{i,\ell}(P_i-1)\cdot \psi_\theta(B_{i,\ell},X_{\ell},\theta_0)+o_p(1),
\end{align*}
$\{P_i:~i=1,\ldots,N\}$ are i.i.d.\ Poisson random variables with mean 1 and are independent of the sample path. 
\end{assumption}

Assumption~\ref{assu: theta-est} states that $\widehat\theta$ is an $M$-estimator with an asymptotically linear representation, implying that $\widehat\theta$ is asymptotically normally distributed.

\begin{assumption}\label{assu: Gamma function}
Assume that $0<\delta \leq \Gamma(x,\theta)\leq M < \infty$ for all $x\in\mathcal{X}$ and $\theta\in\Theta$. 
\end{assumption}
Assumption \ref{assu: Gamma function} implies that $1/\Gamma(x,\theta)$ is uniformly bounded above and uniformly bounded away from zero.  

\begin{assumption}\label{assu: VC assumption}
Assume that for all $x\in\mathcal{X}$, $\Gamma(x,\theta)$ is continuously differentiable in $\theta$ at with bounded derivatives. 
\end{assumption}
Assumption \ref{assu: VC assumption} implies that $\{1(b \leq B_{i,\ell}/\Gamma(X_{\ell},\theta)\leq b+r) : \theta\in\Theta, b\in [\underline{b}, \overline{b}], r\in[0, a]\}$  is a VC class of functions which is crucial for our empirical process results.
Let  
\begin{align*}
&M(b,q,\theta)= E\Big[\frac{B_{i,\ell}}{\Gamma(X_{\ell},\theta)} 1\Big( b \leq  \frac{B_{i,\ell}}{\Gamma(X_{\ell},\theta)}\leq b+\frac{a}{q} \Big)\Big]\\
&~~~~~~~~~~~~~~~~+
\frac{1}{N-1} E\Big[ 1\Big(\frac{B_{i,\ell}}{\Gamma(X_{\ell},\theta)}\leq b+\frac{a}{q} \Big) \Big(b+\frac{a}{q} - \frac{B_{i,\ell}}{\Gamma(X_{\ell},\theta)}\Big)
 - 1\Big(\frac{B_{i,\ell}}{\Gamma(X_{\ell},\theta)}\leq b \Big) \Big(b - \frac{B_{i,\ell}}{\Gamma(X_{\ell},\theta)}\Big) \Big], \notag\\
&W(b,q,\theta)=E\Big[1\Big( b \leq  \frac{B_{i,\ell}}{\Gamma(X_{\ell},\theta)}\leq b+\frac{a}{q} \Big) \Big],
\end{align*}
where $M(b,q,\theta)$ and $W(b,q,\theta)$ are similar to those in (\ref{eq: auction M}) and (\ref{eq: auction W}) except that in $M(b,q,\theta)$ and $W(b,q,\theta)$, we have rescaled bids instead of the original ones.

\begin{assumption}\label{assu: M W derivatives}
Assume that for all $b\in [\underline{b}, \overline{b}]$, and for all $q$, (i) $\nabla_{\theta} M(b,q,\theta)$ and $\nabla_{\theta}M(b,q,\theta)$ are continuous in $\theta$  for all $\theta$ with $\|\theta-\theta_0\|\leq \epsilon$ for some small positive $\epsilon$; (ii)   
$\|\nabla_{\theta} M(b,q,\theta)\|\leq M$ and $\|\nabla_{\theta} M(b,q,\theta)\|\leq M$ for some $M<\infty$. 
\end{assumption}

With a bit abuse of notation, we define
\begin{align*}
&M_{i,\ell}(b,q,\theta)= \frac{B_{i,\ell}}{\Gamma(X_{\ell},\theta)} 1\Big( b \leq  \frac{B_{i,\ell}}{\Gamma(X_{\ell},\theta)}\leq b+\frac{a}{q} \Big)\\
&~~~~~~~~~~~~~~~~+
\frac{1}{N-1}  1\Big(\frac{B_{i,\ell}}{\Gamma(X_{\ell},\theta)}\leq b+\frac{a}{q} \Big) \Big(b+\frac{a}{q} - \frac{B_{i,\ell}}{\Gamma(X_{\ell},\theta)}\Big)
 - 1\Big(\frac{B_{i,\ell}}{\Gamma(X_{\ell},\theta)}\leq b \Big) \Big(b - \frac{B_{i,\ell}}{\Gamma(X_{\ell},\theta)}\Big) , \notag\\
&W_{i,\ell}(b,q,\theta)=1\Big( b \leq  \frac{B_{i,\ell}}{\Gamma(X_{\ell},\theta)}\leq b+\frac{a}{q} \Big),
\end{align*}
Let $\widehat{M}(b,q,\theta)=(NL)^{-1}\sum_{i,\ell} M_{i,\ell}(b,q,\theta)$ and $\widehat{W}(b,q,\theta)=(NL)^{-1}\sum_{i,\ell} W_{i,\ell}(b,q,\theta)$. 
Define $\widehat{\nu}(b_1,b_2, q,\widehat{\theta})=\widehat{M}(b_2,q,\widehat{\theta})\widehat{W}(b_1,q,\widehat{\theta})-\widehat{M}(b_1,q,\widehat{\theta})\widehat{W}(b_2,q,\widehat{\theta})$

\begin{lemma}\label{lemma: empirical process}
Suppose Assumptions \ref{assu: homogenize}, \ref{assu: theta-est}, \ref{assu: Gamma function}, \ref{assu: VC assumption} and \ref{assu: M W derivatives} hold. Then it is true that
\begin{align}
&\sqrt{NL}(\widehat{M}(b,q,\widehat{\theta})-{M}(b,q,\theta_0))\notag\\
=&\frac{1}{\sqrt{NL}}\sum_{i,\ell} M_{i,\ell}(b,q,\theta_0)-{M}(b,q,\theta_0)+\nabla_{\theta} M(b,q,\theta_0)\psi_\theta(B_{i,\ell},X_{\ell},\theta_0)+o_p(1)\label{eq: influ M semi}\\
\equiv& 
\frac{1}{\sqrt{NL}}\sum_{i,\ell} \phi_{i,\ell,M}(b,q,\theta_0)+o_p(1)\notag\\
&\sqrt{NL}(\widehat{W}(b,q,\widehat{\theta})-{W}(b,q,\theta_0))\notag\\
=&\frac{1}{\sqrt{NL}}\sum_{i,\ell} W_{i,\ell}(b,q,\theta_0)-{W}(b,q,\theta_0)+\nabla_{\theta} W(b,q,\theta_0)\psi_\theta(B_{i,\ell},X_{\ell},\theta_0)+o_p(1)\label{eq: influ W semi}\\
\equiv& 
\frac{1}{\sqrt{NL}}\sum_{i,\ell} \phi_{i,\ell,W}(b,q,\theta_0)+o_p(1)\notag
\end{align}
uniformly over $b\in[\underline{b}, \overline{b}]$ and $q$. 
\end{lemma}

Lemma \ref{lemma: empirical process} provides the influence function representations of $\widehat{M}(b,q,\widehat{\theta})$ and $\widehat{W}(b,q,\widehat{\theta})$.  It implies that both $\sqrt{NL}(\widehat{M}(b,q,\widehat{\theta})-{M}(b,q,\theta_0))$ and $\sqrt{NL}(\widehat{W}(b,q,\widehat{\theta})-{W}(b,q,\theta_0))$ weakly converge to Gaussian processes.  

\begin{lemma}\label{lemma: nu empirical process}
Suppose Assumptions \ref{assu: homogenize}, \ref{assu: theta-est}, \ref{assu: Gamma function},\ref{assu: VC assumption} and \ref{assu: M W derivatives} hold. Then it is true that
\begin{align}
&\sqrt{NL}(\widehat{\nu}(b_1,b_2, q,\widehat{\theta})-{\nu}(b_1,b_2, q,\theta_0))\equiv\frac{1}{\sqrt{NL}}\sum_{i,\ell} \phi_{i,\ell,\nu}(b_1,b_2,q,\theta_0)+o_p(1),~\text{where}\\
&\phi_{i,\ell,\nu}(b_1,b_2,q,\theta_0)=W(b_1,q,\theta_0)\phi_{i,\ell,M}(b_2,q,\theta_0)+M(b_2,q,\theta_0)\phi_{i,\ell,W}(b_1,q,\theta_0)\notag\\
&~~~~~~~~~~~~~~~~~~~~~~~~~
-W(b_2,q,\theta_0)\phi_{i,\ell,M}(b_1,q,\theta_0)-M(b_1,q,\theta_0)\phi_{i,\ell,W}(b_2,q,\theta_0),\notag
\end{align}
uniformly over $(b_1,b_2,q)\in \mathcal{L}$. 
\end{lemma}

Lemma \ref{lemma: nu empirical process} provides the influence function representation of $\widehat{\nu}(b_1,b_2, q,\widehat{\theta})$.  It follows that $\sqrt{NL}(\widehat{\nu}(b_1,b_2, q,\widehat{\theta})-{\nu}(b_1,b_2, q,\theta_0))$ weakly converges to a Gaussian process.

Let $\widehat{\theta}^*$ and $\widehat{\nu}^*(b_1,b_2, q,\widehat{\theta}^*)$ denote the bootstrap estimators for $\theta_0$ and $\nu(b_1,b_2,q, \theta_0)$. 
Let $\sigma^2_{\nu}(b_1,b_2,q,\theta_0)$ denote the asymptotic variance of $\sqrt{NL}(\widehat{\nu}(b_1,b_2, q,\widehat{\theta})-{\nu}(b_1,b_2, q,\widehat{\theta}))$.
Also define, ${\sigma}_{\nu,\epsilon}(b_1,b_2,q,\theta_0)=\max\{{\sigma}_\nu(b_1,b_2,q),\sqrt{\epsilon}\cdot{\sigma}_\nu(\underline{b}, (\underline{b}+\overline{b})/2,2,\theta_0)\}.$
Let the bootstrap standard error estimator for $\sigma_{\nu}(b_1,b_2,q,\theta_0)$ as
\begin{align*}
	&\widehat{\sigma}^*_{\nu}(b_1,b_2,q,\widehat{\theta})=\Big(\frac{1}{K} \sum_{k=1}^K NL(\widehat{\nu}^{*,k}(b_1,b_2, q,\widehat{\theta}^{*,k}) -\bar{\nu}^*(b_1,b_2, q,\widehat{\theta}))^2\Big)^{1/2},~\text{where}\\
	&\bar{\nu}^*(b_1,b_2, q,\widehat{\theta})=\frac{1}{K} \sum_{k=1}^K \widehat{\nu}^{*,k}(b_1,b_2, q, \widehat{\theta}^{*,k} ),
\end{align*}
and define
$\widehat{\sigma}^*_{\nu,\epsilon}(b_1,b_2,q,\widehat{\theta})=\max\{\widehat{\sigma}^*_{\nu}(b_1,b_2,q),\sqrt{\epsilon}\cdot\widehat{\sigma}^*_{\nu}(\underline{b}, (\underline{b}+\overline{b})/2,2,\widehat{\theta})\}.$

\begin{lemma}\label{lemma: nu bootstrap}
Suppose Assumptions \ref{assu: homogenize}, \ref{assu: theta-est}, \ref{assu: Gamma function}, \ref{assu: VC assumption} and \ref{assu: M W derivatives} hold. Then it is true that uniformly over $(b_1,b_2,q)\in \mathcal{L}$
\begin{align}
&\sqrt{NL}(\widehat{\nu}^*(b_1,b_2, q,\widehat{\theta}^*)-{\nu}(b_1,b_2, q,\widehat{\theta}))\notag\\
\equiv&\frac{1}{\sqrt{NL}}\sum_{i,\ell} (P_i-1)\cdot \phi_{i,\ell,\nu}(b_1,b_2,q,\theta_0)+o_p(1),~\text{where}
\end{align}
$\{P_i:~i=1,\ldots,N\}$ are the same as those in Assumption \ref{assu: theta-est}(iii). 
In addition, uniformly over $(b_1,b_2,q)\in \mathcal{L}$,
\begin{align}
\lim_{K\rightarrow\infty}|\widehat{\sigma}^*_{\nu,\epsilon}(b_1,b_2,q,\widehat{\theta})- {\sigma}_{\nu,\epsilon}(b_1,b_2,q,\theta_0)|\stackrel{p}{\rightarrow} 0
\end{align}
\end{lemma}
Lemma \ref{lemma: nu bootstrap} can be proved by a similar procedure in Theorem 3.6.1 of \cite{VW1996}.  
Lemma \ref{lemma: nu bootstrap} implies the bootstrap validity of our bootstrap estimators in that it can approximate the limiting process of $\sqrt{NL}(\widehat{\nu}(b_1,b_2, q,\widehat{\theta})-{\nu}(b_1,b_2, q,\theta_0))$ well.  The second part of it shows the uniform consistency of the bootstrap standard error estimators.  

\begin{thm} \label{thm: semi test} Suppose Assumptions \ref{assu: homogenize}, \ref{assu: theta-est}, \ref{assu: Gamma function}, \ref{assu: VC assumption} and \ref{assu: M W derivatives} hold. Then the following results regarding the semiparametric test proposed in Section \ref{sec: observed hetero semiparametric} in which the test statistic, $\widehat{T}_{NL,se}$, and  critical value $\hat{c}_{\eta,se}$ are generated as in ``Implementation Procedure for Semiparametric Test under Observed Auction Heterogeneity" hold: \\
(a) under the null hypothesis,
${\lim}_{L\rightarrow\infty}
P(\widehat{T}_{NL,se}>\hat{c}_{\eta,se})\leq \alpha$;\\
(b) under the alternative hypothesis,
${\lim}_{L\rightarrow\infty}
P(\widehat{T}_{NL,se}>\hat{c}_{\eta,se})= 1.$ 
\end{thm}
The proof for Theorem \ref{thm: semi test} is similar to that for Theorem \ref{thm: general test size and power} given the lemmas in this section and we omit the details. 

\subsection{Semiparametric Test under Unobserved Game Heterogeneity} \label{sec: semi unobserved hetero}
We extend the semiparametric test in Section \ref{sec: observed hetero semiparametric} to allow for unobserved heterogeneity.  In addition to the setup in Section \ref{sec: observed hetero semiparametric},  we assume there exists an unobserved vector of $Z$ denoting the unobserved characteristics.  As in Section \ref{sec: observed hetero semiparametric}, we impose a semiparametric model to homogenize the bids.
Let $\Gamma(x)$ be a strictly positive function of covariates. We assume that Assumption \ref{assu: D C homogeneous} holds.  

\begin{assumption} \label{assu: homogenize unobserved}
Assume that (i) $\{(V_{1,\ell}=\Gamma(X_{\ell})\Lambda(Z_{\ell},u_{1,\ell}), V_{2,\ell}=\Gamma(X_{\ell})\Lambda(Z_{\ell},u_{2,\ell}),\ldots, V_{N,\ell}=\Gamma(X_{\ell})\Lambda(Z_{\ell},u_{N,\ell}), X_{\ell},Z_\ell): \ell=1,\ldots, L\}$ are i.i.d.\  random vectors;
(ii) for each $\ell$, conditional on $Z_{\ell}=z$, $\{\Lambda(Z_{\ell},u_{1,\ell}): i=1,\ldots, N \}$ are i.i.d.\ random variables 
with conditional PDF $f_\lambda(\cdot|Z=z)$ and conditional CDF $F_\lambda(\cdot|Z=z)$ and 
are independent of  
$X_{\ell}$;
(iii)
$f_\lambda(\cdot|Z=z)$ is strictly positive and bounded away from zero on its support, a compact
interval $[\underline{\lambda}_{z}, \overline{\lambda}_{z}] \subseteq R$, and is twice continuously differentiable on $(\underline{\lambda}_{z}, \overline{\lambda}_{z})$.
\end{assumption}
The condition that for each $\ell$, conditional on $Z_{\ell}=z$, $\{\Lambda(Z_{\ell},u_{i,\ell}): \ell=1,\ldots, L \}$ are i.i.d.\ random variables is implied by that 
for each $\ell$, conditional on $Z_{\ell}=z$, $\{u_{i,\ell}: \ell=1,\ldots, L \}$ are i.i.d.\ random variables.  We assume that $V_{i,\ell}=\Gamma(X_{\ell})\Lambda(Z_{\ell},u_{i,\ell})$ so that we can rescale the actions by a function of observables. 

Then the same argument in Section \ref{sec: observed hetero semiparametric}
implies that 
\begin{align*}
&~\xi(b,x,z) \;\equiv\; \frac{T'(b,x,z)}{P'(b,x,z)} ~\text{is strictly increasing in $b$ for all $x\in\mathcal{X}$ and for all $z\in\mathcal{Z}$.} \notag\\
\Leftrightarrow
&~\xi(b^\lambda,z) \;\equiv\; \frac{T^{\lambda\prime }(b^\lambda,z)}{P^{\lambda\prime }(b^\lambda,z)} ~\text{is strictly increasing in $b^\lambda$ for all $z\in\mathcal{Z}$.} 
\end{align*}
Combining  Section \ref{sec: observed hetero semiparametric} and Section \ref{sec: unobserved hetero nonparametric}, we can formulate the testable implications as
\begin{align*}
&H_{0,\text{un,semi}}': \nu(b^\lambda_1,b^\lambda_2, q)=M(b^\lambda_2,q)W(b^\lambda_1,q)-M(b^\lambda_1,q)W(b^\lambda_2,q)\leq 0~\text{for all $(b^\lambda_1,b^\lambda_2, q)\in\mathcal{L}$},
\end{align*}
where 
\begin{align*}
&M(b^\lambda,q)= \int_0^1\int_{b^\lambda}^{b^\lambda+\frac{a}{q}} h(\tilde{b},z)T'(\tilde{b},z) d\tilde{b} dz, 
~~~
W(b^\lambda,q)=\int_0^1\int_{b^\lambda}^{b^\lambda+\frac{a}{q}} h(\tilde{b},z)P'(\tilde{b},z) d\tilde{b} dz,~\text{and}\\
&\mathcal{L}=\Big\{(b^\lambda_1,b^\lambda_2, q): \Big(\frac{b^\lambda_1-\underline{b}^\lambda}{a}, \frac{b^\lambda_2-\underline{b}^\lambda}{a}\Big)\cdot q\in(0,1,\ldots,q-1)^2, b^\lambda_1>b^\lambda_2~\text{for q=2,3,\ldots}  \Big\}.
\end{align*}
Then we can apply the test in Section \ref{sec: observed hetero semiparametric} and we omit the details.

\section{Proofs for Theorems}
\noindent{\bf Proof of Theorem \ref{thm: LZ2018}:}
\label{thm:gpv_general_patched_suff_rewrite}
Consider a symmetric Bayesian game with scalar private type $v\in[\underline v,\bar v]$ and scalar action $b\in[\underline b,\bar b]$.
Given opponents' action vector $\bm B_{-i}$, agent $i$'s interim payoff from choosing $b$ is
\begin{equation}\label{eq:payoff_general_patched}
\pi(b,v;\bm B_{-i}) \;=\; v\,D(b,\bm B_{-i}) \;-\; C(b,\bm B_{-i}).
\end{equation}

Fix a continuously differentiable CDF $G$ on $[\underline b,\bar b]$ with density $g>0$ on
$(\underline b,\bar b)$. Define the reduced-form (expected) payoff
\begin{equation}\label{eq:reduced_forms_patched}
U(b,v;G)\;:=\; v\,P(b;G) \;-\; T(b;G),
\qquad
P(b;G):=E_{\bm B_{-i}}\!\left[D(b,\bm B_{-i})\right],\quad
T(b;G):=E_{\bm B_{-i}}\!\left[C(b,\bm B_{-i})\right],
\end{equation}
where the expectation is taken under $\bm B_{-i}\sim G^{\otimes (N-1)}$ and the expectations exist.

Assume $P(\cdot;G)$ and $T(\cdot;G)$ extend continuously to $[\underline b,\bar b]$ and satisfy
$P(\cdot;G),T(\cdot;G)\in C^{2}(\underline b,\bar b)$ with $P'(b;G)>0$ for all
$b\in(\underline b,\bar b)$. Define the inverse-type map
\begin{equation}\label{eq:xi_def_patched}
\xi(b;G)\;:=\;\frac{T'(b;G)}{P'(b;G)}, \qquad b\in(\underline b,\bar b),
\end{equation}
and assume $\xi(\cdot;G)$ extends continuously to $[\underline b,\bar b]$, with
$\underline\xi:=\xi(\underline b;G)$ and $\bar\xi:=\xi(\bar b;G)$.

\medskip
\noindent\textbf{Necessity.}
Suppose $G$ is rationalizable by a strictly monotone interior symmetric equilibrium, in the sense that
there exist a type distribution $F$ on $[\underline v,\bar v]$ and a symmetric equilibrium strategy
$s:[\underline v,\bar v]\to[\underline b,\bar b]$ such that:

(i) $s$ is continuous and strictly increasing;

(ii) (C1) actions $(B_1,...,B_N)$  are i.i.d. with common CDF $G$.  

(iii) for every $v\in(\underline v,\bar v)$, the action $b=s(v)$ is an interior maximizer of $U(b,v;G)$ on $ [\underline b,\bar b]$.

Then the inverse-type map is strictly increasing on $(\underline b,\bar b)$ and satisfies
\begin{equation}\label{eq:xi_equals_inverse_nec}
\xi(b;G)=s^{-1}(b)\qquad \text{for all } b\in(\underline b,\bar b).
\end{equation}

\noindent{\bf Proof for Necessity:}
Suppose there exists $(F,s)$ satisfying (i)--(iii) in the theorem statement.
By (ii), in the symmetric equilibrium the opponents' action profile $\bm B_{-i}$ is i.i.d.\ with CDF $G$,
so a type-$v$ agent evaluates expected payoff $U(b,v;G)=vP(b;G)-T(b;G)$.
By (iii), $b=s(v)\in(\underline b,\bar b)$ is an interior maximizer for each $v\in(\underline v,\bar v)$,
so the first-order condition holds:
\[
0=\frac{\partial}{\partial b}U(b,v;G)\Big|_{b=s(v)}
= v\,P'(s(v);G)-T'(s(v);G).
\]
Since $P'(b;G)>0$ on $(\underline b,\bar b)$, dividing yields
\[
v=\frac{T'(s(v);G)}{P'(s(v);G)}=\xi(s(v);G),
\qquad v\in(\underline v,\bar v).
\]
Because $s$ is strictly increasing and maps $(\underline v,\bar v)$ onto $(\underline b,\bar b)$, for each $b\in(\underline b,\bar b)$ there exists a unique $v=\xi(b;G)=s^{-1}(b)$ such that $b=s(v)$. Since $s^{-1}$ is strictly increasing, $\xi(\cdot;G)$ is strictly increasing on $(\underline b,\bar b)$.

We now show that the rationalizing type distribution $F$ is unique. Since $\xi(b;G)=s^{-1}(b)$ on $(\underline b,\bar b)$, the equilibrium type satisfies $V=s^{-1}(B)=\xi(B;G)$. Under (C1), $B\sim G$, and $\xi(\cdot;G)$ is uniquely
determined by $G$. Therefore, the rationalizing $F$  is uniquely pinned down as
\[
F(v)=\Pr(V\le v)=\Pr(\xi(B;G)\le v)=G(\xi^{-1}(v;G)),\qquad v\in[\underline\xi,\bar\xi].
\]

\medskip
\noindent\textbf{Sufficiency.}
Conversely, if $\xi(\cdot;G)$ is continuous and strictly increasing on $(\underline b,\bar b)$,
then $G$ is rationalizable: define $F$ on $[\underline\xi,\bar\xi]$ as the pushforward of $G$
through $\xi(\cdot;G)$,
\begin{equation}\label{eq:F_pushforward_nec_suff}
F(v)\;:=\;\Pr\!\big(\xi(B;G)\le v\big)
\;=\;G\!\big(\xi^{-1}(v;G)\big),
\qquad v\in[\underline\xi,\bar\xi],\ \ B\sim G,
\end{equation}
and define the candidate equilibrium strategy
\begin{equation}\label{eq:s_def_nec_suff}
s(v):=\xi^{-1}(v;G),\qquad v\in[\underline\xi,\bar\xi].
\end{equation}
Then $s$ is a strictly monotone symmetric Bayesian Nash equilibrium for i.i.d.\ types $V\sim F$,
and the induced equilibrium action distribution equals $G$.

\noindent{\bf Proof for Sufficiency:}
Assume now that $\xi(\cdot;G)$ is continuous and strictly increasing on $(\underline b,\bar b)$.
Define $F$ by \eqref{eq:F_pushforward_nec_suff} on $[\underline\xi,\bar\xi]$ and define
$s(v):=\xi^{-1}(v;G)$ as in \eqref{eq:s_def_nec_suff}.

\noindent{\bf Step 1 (The inverse of $\xi$ reproduces exactly $G$):}
Let $V\sim F$ and set $B:=s(V)=\xi^{-1}(V;G)$. Then for any $b\in[\underline b,\bar b]$,
\begin{align*}
\Pr(B\le b)
&=\Pr\!\big(\xi^{-1}(V;G)\le b\big)
=\Pr\!\big(V\le \xi(b;G)\big)
=F\!\big(\xi(b;G)\big) \\
&=G\!\big(\xi^{-1}(\xi(b;G);G)\big)
=G(b),
\end{align*}
where the fourth equality uses \eqref{eq:F_pushforward_nec_suff} evaluated at $v=\xi(b;G)$ and the
fifth uses $\xi^{-1}(\xi(b;G);G)=b$. Since $(V_1,...,V_N)$ are i.i.d. $F$ and $B_i = s(V_i)$ is applied componentwise, $(B_1,...,B_N)$ are i.i.d. with common CDF $G$, so (C1) holds. 

\noindent{\bf Step 2 ($s$ satisfies the interior first-order condition under beliefs $G$):}
Fix $v\in(\underline\xi,\bar\xi)$ and suppose the bidder believes $\bm B_{-i}$ are i.i.d.\ with CDF $G$.
Then the bidder evaluates $U(b,v;G)=vP(b;G)-T(b;G)$, whose derivative in $b$ is
\[
\frac{\partial}{\partial b}U(b,v;G)=v\,P'(b;G)-T'(b;G).
\]
At $b=s(v)=\xi^{-1}(v;G)$ we have $\xi(b;G)=v$, so
$T'(b;G)=\xi(b;G)P'(b;G)=vP'(b;G)$ and therefore
\[
\frac{\partial}{\partial b}U\big(s(v),v;G\big)=0.
\]

\noindent{\bf Step 3 (Global optimality check):}
Fix $v\in(\underline\xi,\bar\xi)$ and let $b^*:=s(v)$. For any $b\in[\underline b,\bar b]$,
\begin{align*}
U(b,v;G)-U(b^*,v;G)
&=\int_{b^*}^{b} \big(vP'(t;G)-T'(t;G)\big)\,dt =\int_{b^*}^{b} \big(v-\xi(t;G)\big)P'(t;G)\,dt,
\end{align*}
where the second line uses $T'(t;G)=\xi(t;G)P'(t;G)$. Since $P'(t;G)>0$ and $\xi(\cdot;G)$ is strictly increasing, the integrand is strictly positive for $b<b^*$ and strictly negative for $b>b^*$, implying $U(b,v;G)\le U(b^*,v;G)$ for all $b$, with strict inequality for $b\neq b^*$ in the interior. Therefore $b^*=s(v)$ is a global maximizer of $U(\cdot,v;G)$ on $[\underline b,\bar b]$.

Thus, when types are i.i.d.\ $F$ and all bidders play $s$, the induced action distribution is $G$
and each type best-responds given beliefs $G$, so $(F,s)$ rationalizes $G$.~~~~$\Box$

\noindent{\bf Proof of Theorem \ref{thm: general test size and power}:}
Let \begin{align}
{\sigma}^2_{\nu,\epsilon}(b_1,b_2,q)=\max\{{\sigma}^2_
\nu(b_1,b_2,q),\epsilon\cdot
{\sigma}^2_\nu(\underline{b}, (\underline{b}+\overline{b})/2,2)\}.\label{eq: general sigma epsilon}
\end{align}
Assumption \ref{assu: general consistent variance} implies that $\sup_{(b_1,b_2,q)\in \mathcal{L}}|\widehat{\sigma}^2_{\nu,\epsilon}(b_1,b_2,q)-{\sigma}^2_{\nu,\epsilon}(b_1,b_2,q)|\stackrel{p}{\rightarrow}0$ and ${\sigma}^2_{\nu,\epsilon}(b_1,b_2,q)$ is uniformly bounded away from zero for all $(b_1,b_2,q)\in\mathcal{L}$. Assumption \ref{assu: general gaussian limit} and Assumption \ref{assu: general consistent variance} together imply that
\begin{align*}
\sqrt{NL}\Big(\frac{\widehat{\nu}(\cdot,\cdot,\cdot)}{\widehat{
\sigma}_{\nu,\epsilon}(\cdot,\cdot,\cdot)
}-\frac{{\nu}(\cdot,\cdot,\cdot)}{\sigma_{\nu,\epsilon}(\cdot,\cdot,\cdot)}\Big)\Rightarrow
\Phi_{\kappa_\epsilon}(\cdot,\cdot,\cdot)
\end{align*}
where $\kappa_\epsilon(\cdot,\cdot,\cdot)
=\frac{\kappa\big((b'_1,b'_2,q'),(b''_1,b''_2,q'')\big)}{\sigma_{\nu,\epsilon}(b'_1,b'_2,q')\sigma_{\nu,\epsilon}(b''_1,b''_2,q'')}$ for $(b'_1,b'_2,q'),(b''_1,b''_2,q'')\in \mathcal{L}$.
Let $\mathcal{L}^o\equiv\{(b_1,b_2,q)\in\mathcal{L}: \nu(b_1,b_2,q)=0\}$ which is the collection of moment conditions hold with equality.  Let $\mathcal{L}^+\equiv\{(b_1,b_2,q)\in\mathcal{L}: \nu(b_1,b_2,q)>0\}$ which is the collection of strictly positive moment conditions.  Under the null, $\mathcal{L}^+$ is an empty set and under the alternative, $\mathcal{L}^+$ contains at least one element. By a similar proof of Lemma 2.1 of \cite{DonaldHsu2016}, 
under the null hypothesis, one can show that 
\begin{align}
\widehat{T}_{NL}=&\sum_{(b_1,b_2,q)\in\mathcal{L}} \max\Big\{ \sqrt{NL} \frac{\widehat{\nu}(b_1,b_2,q)}
{\widehat{\sigma}_{\nu,\epsilon}(b_1,b_2,q)},0\Big\}^2 Q(b_1,b_2,q)\notag \\
\stackrel{D}{\rightarrow} &
\sum_{(b_1,b_2,q)\in\mathcal{L}^o} \max\Big\{ \Phi^2_{\kappa_\epsilon}(b_1,b_2,q),0\Big\}^2 Q(b_1,b_2,q), \label{eq: null distribution}
\end{align}
where $\stackrel{D}{\rightarrow}$ denotes convergence in distribution. 
(\ref{eq: null distribution}) shows that the limiting null distribution depends only on those moment conditions holding with equality.  If under null hypothesis, $\mathcal{L}^o$ is an empty set in that every moment condition holds with strictly inequality, then $\widehat{T}_{NL}$ converges in probability to zero at the rate $S$.  Also, if under null hypothesis, $\mathcal{L}^o$ is not an empty set, but $\sigma_{\nu}(b_1,b_2,q)=0$ for all $(b_1,b_2,q)\in \mathcal{L}^o$, then it is true that $\widehat{T}_{NL}$ converges in probability to zero at the rate $L$ as well.  
These two cases are degenerate in that the limiting null distribution converge in probability to a zero point.  
On the other hand, one can also show that $\widehat{T}_{NL}\rightarrow \infty$ at the rate of $L$. Next, similar to the proof of Lemma A.2 of \cite{DonaldHsu2016}, one can show that the simulated test statistic defined as
\begin{align*}
\sum_{(b_1,b_2,q)\in\mathcal{L}} \max\Big\{
\frac{{\Phi}^*(b_1,b_2,q)}
{\widehat{\sigma}_{\nu,\epsilon}(b_1,b_2,q)}+\psi_{NL}(b_1,b_2,q),0\Big\}^2 Q(b_1,b_2,q)
\end{align*}
converge in distribution to 
\begin{align*}
&\sum_{(b_1,b_2,q)\in\mathcal{L}} \max\Big\{
\frac{{\Phi}^*(b_1,b_2,q)}
{\widehat{\sigma}_{\nu,\epsilon}(b_1,b_2,q)}+\psi_{NL}(b_1,b_2,q),0\Big\}^2 Q(b_1,b_2,q)\\
\stackrel{D}{\rightarrow} &
\sum_{(b_1,b_2,q)\in\mathcal{L}^o \cup \mathcal{L}^+} \max\Big\{ \Phi^2_{\kappa_\epsilon}(b_1,b_2,q),0\Big\}^2 Q(b_1,b_2,q)
\end{align*}
conditional on sample path in probability approaching one. 
Therefore, it follows that 
$\hat{c}_{\eta}\stackrel{p}{\rightarrow} c(1-\alpha+\eta)+\eta,$
where $c(1-\alpha+\eta)$ denotes the ($1-\alpha+\eta$)-th quantile of 
\begin{align*}
\sum_{(b_1,b_2,q)\in\mathcal{L}^o \cup \mathcal{L}^+} \max\Big\{ \Phi^2_{\kappa_\epsilon}(b_1,b_2,q),0\Big\}^2 Q(b_1,b_2,q)
\end{align*} 
and it is bounded in probability and strictly greater than $\eta$. 

Then it is true that under the null, the asymptotic size of our test is $\lim_{L\rightarrow\infty}P(\widehat{T}_{NL}\leq \hat{c}_{\eta})\leq \alpha$.  Also, under the alternative, our test is consistent in that $\lim_{L\rightarrow\infty}P(\widehat{T}_{NL}\leq \hat{c}_{\eta})=1$.~~~~$\Box$ 


\section{Proofs for Lemmas}

\noindent{\bf Proof of Lemma \ref{lemma: general H0 transformation}:}
We apply Lemma 3.1 of  \cite{HsuShen2021} to show this lemma.  We set the $\lambda(b)$ and $h(b)$ functions of Lemma 3.1 of  \cite{HsuShen2021} as
$T'(b)/P'(b)$ and $h(b)P'(b)$ in our example.  Then it follows that 
\begin{align*}
M(b,q)\equiv&\int_{b}^{b+\frac{a}{q}} \frac{T'(b)}{P'(b)}\cdot h(b)P'(b) db=
\int_{b}^{b+\frac{a}{q}} T'(b)h(b) db,~\text{and}\\
W(b,q)\equiv&\int_{b}^{b+\frac{a}{q}}   h(b)P'(b) db .
\end{align*}
This shows Lemma \ref{lemma: general H0 transformation}.~~~~$\Box$

\noindent{\bf Proof of Lemma \ref{lemma: auction H0 transformation}:}
It is sufficient to show that $M(b,q)$ and $W(b,q)$ are identified as (\ref{eq: auction W}) and (\ref{eq: auction M}). 
Following the discussion after (\ref{eq: auction H0}), we have 
\begin{align*}
M(b,q)&=\int_{b}^{b+\frac{a}{q}} \big(\tilde{b} g(\tilde{b})+(N-1)^{-1}G(\tilde{b})\big) d\tilde{b}\\
&=E\Big[B_{i,\ell} 1\Big( b \leq  B_{i,\ell}\leq b+\frac{a}{q} \Big)\Big]\\
&~~~~~+
\frac{1}{N-1} E\Big[ 1\Big(B_{i,\ell}\leq b+\frac{a}{q} \Big) \Big(b+\frac{a}{q} - B_{i,\ell}\Big)
 - 1(B_{i,\ell}\leq b ) (b - B_{i,\ell}) \Big],\\
W(b,q)&=\int_{b}^{b+\frac{a}{q}} g(\tilde{b}) d\tilde{b}=E\Big[1\Big( b \leq  B_{i,\ell}\leq b+\frac{a}{q} \Big) \Big].
\end{align*}
Note that the second equality holds because $\int_{\underline{b}}^b G(\tilde{b})d\tilde{b} =E[ 1(B_{i,\ell} \leq \tilde{b}) (\tilde{b}- B_{i,\ell}) ]$.~~~~$\Box$

\medskip

\noindent{\bf  Proof of Lemma \ref{lemma: contest H0 transformation}}
From the discussion after (\ref{eq: contest H0}), it is sufficient to show that 
\begin{align*}
M(b,q)&=
\int_{b}^{b+\frac{a}{q}} \tilde{b} c'(\tilde{b}) g(\tilde{b}) d\tilde{b}=E\Big[B_{i,\ell} c'(B_{i,\ell}) \cdot 1\Big( b \leq  B_{i,\ell}\leq b+\frac{a}{q} \Big)\Big],\\
W(b,q)&=\int_{b}^{b+\frac{a}{q}} E_{\bm B_{-i}}[D(\tilde{b},\bm B_{-i})(1-D(\tilde{b},\bm B_{-i}))] g(\tilde{b}) d\tilde{b}\\
&=E\Big[1\Big( b \leq  B_{i,\ell}\leq b+\frac{a}{q} \Big) \frac{B_{i,\ell}}{\sum_{j=1}^N B_{j,\ell} }\Big(1-\frac{B_{i,\ell}}{\sum_{j=1}^N B_{j,\ell}}\Big)\Big],
\end{align*}
where the last equality holds by law of iterated expectations.~~~~$\Box$

\medskip
The proof for Lemma \ref{lemma: public H0 transformation} is similar to that for Lemma \ref{lemma: contest H0 transformation}, so we omit the details.

\medskip
\noindent{\bf Proof of Lemma \ref{lemma: Cournot H0 transformation}:}
For $J(b,q)=[b,b+a/q]$, $T'(t;G)=c'(t)$ and $h(t)=g(t)$ imply
\[
M(b,q)=\int_{J(b,q)}g(t)c'(t)\,dt
=E[c'(B_i)1\{B_i\in J(b,q)\}].
\]
Differentiating the expected revenue gives
$P'(t;G)=E_G[\mathcal I(t+S_{-i})+t\mathcal I'(t+S_{-i})]$.
Independence and iterated expectations then yield
\begin{align*}
W(b,q)
&=\int_{J(b,q)}g(t)
E_G[\mathcal I(t+S_{-i})+t\mathcal I'(t+S_{-i})]\,dt\\
&=E\left[1\{B_i\in J(b,q)\}
\left\{\mathcal I\left(\sum_{j=1}^N B_j\right)
+B_i\mathcal I'\left(\sum_{j=1}^N B_j\right)\right\}\right].
\end{align*}
The maintained positive expected marginal revenue and compact-support smoothness conditions make $hP'$ positive in the interior and bounded. Lemma \ref{lemma: general H0 transformation} gives the asserted equivalence. The argument does not require marginal revenue to be positive for every realized rival profile; positivity of its expectation is the maintained condition.~~~~$\Box$

\medskip

\noindent{\bf Proof of Lemma \ref{lemma: general observed hetero H0 transformation}:}
We apply Lemma A.1 of \cite{HsuShen2021} to show Lemma \ref{lemma: general observed hetero H0 transformation}.  Similar to proof for Lemma \ref{lemma: general H0 transformation}, 
we set $\lambda(b,x)$ and $h(b,x)$ functions in Lemma A.1 of \cite{HsuShen2021} as  $T'(b,x)/P'(b,x)$ and $h(b,x)P'(b,x)$ in our example. Then it follows that
\begin{align*}
M(b,x, q)\equiv&\int_{x}^{x+\frac{1}{q}}\int_{b}^{b+\frac{a}{q}} \frac{T'(\tilde{b},\tilde{x})}{P'(\tilde{b},\tilde{x})}\cdot h(\tilde{b},\tilde{x})P'(\tilde{b},\tilde{x}) d\tilde{b}d\tilde{x}=
\int_{x}^{x+\frac{1}{q}}\int_{b}^{b+\frac{a}{q}} T'(\tilde{b},\tilde{x})h(\tilde{b},\tilde{x}) d\tilde{b} d\tilde{x},~\text{and}\\
W(b,x,q)\equiv&\int_{x}^{x+\frac{1}{q}}\int_{b}^{b+\frac{a}{q}} h(\tilde{b},\tilde{x})P'(\tilde{b},\tilde{x}) d\tilde{b}d\tilde{x}.
\end{align*}
This shows Lemma \ref{lemma: general observed hetero H0 transformation}.~~~~$\Box$



\medskip
\section{Monotonicity Violations}\label{app: example constructions}
\subsection{Fixed-payoff Contest Construction and Proof}
Consider the lottery contest of Section \ref{subsec: contest model} with two contestants and the entire feasible effort set $\mathcal B=[1,2]$. The minimum effort and upper effort limit are part of this particular game. For own effort $b$ and the rival's effort $s$, fix
\[
D(b,s)=\frac{b}{b+s},\qquad
c(b)=\frac53\left(e^{-3/5}-e^{-3b/5}\right).
\]
Thus $c(1)=0$, $c'(b)=e^{-3b/5}>0$, and $c''(b)<0$. The payoff is $vD(b,s)-c(b)$. Compare i.i.d.\ action profiles with common densities
\[
g_0(b)=\frac{27}{10}(2-b)^2+\frac{1}{10},\qquad
g_1(b)=1,\qquad 1\leq b\leq2.
\]
Both densities integrate to one, are smooth, and are bounded away from zero. We hold the payoff functions and full feasible set fixed. We show that $G_0$ is rationalized by a strictly increasing equilibrium, while $G_1$ violates inverse monotonicity.

Write $S=B_j$ for the other contestant's effort and, for $k\geq2$, define
\[
A_k(b;G)=E_G\!\left[\frac{S}{(b+S)^k}\right],
\qquad
r_G(b)=\frac{2A_3(b;G)}{A_2(b;G)}.
\]
For the fixed cost $c'(b)=e^{-3b/5}$,
\[
\frac{d}{db}\log\xi(b;G)=-\frac35+r_G(b).
\]
All derivatives exist and are bounded on $[1,2]$, and $P'=A_2>0$. Differentiating $A_k$ gives $A_k'=-kA_{k+1}$, hence
\[
r_G'(b)=-\frac{6A_4(b;G)}{A_2(b;G)}
+\frac{4A_3(b;G)^2}{A_2(b;G)^2}<0.
\]
The inequality follows from $A_3^2\leq A_2A_4$, by Cauchy--Schwarz, and $A_4>0$.

Under $g_0(s)=(27/10)(2-s)^2+1/10$, $E_0[S-1]=11/40$. At $b=2$, let $w(s)=s/(2+s)^2$ and let $E_w$ denote expectation under the probability distribution proportional to $w(s)g_0(s)$. Since $1/9\leq w(s)\leq1/8$ for $s\in[1,2]$,
\[
E_w[S]
=1+\frac{E_0[w(S)(S-1)]}{E_0[w(S)]}
\leq1+\frac{9}{8}\frac{11}{40}=\frac{419}{320}.
\]
Jensen's inequality and the monotonicity of $r_{G_0}$ therefore imply
\[
r_{G_0}(b)\geq r_{G_0}(2)
=2E_w\!\left[\frac{1}{2+S}\right]
\geq\frac{2}{2+E_w[S]}
\geq\frac{640}{1059}>\frac35,
\qquad b\in[1,2].
\]
Thus $\xi'(b;G_0)>0$ throughout the feasible set.

Under the uniform density $g_1$, the definition of $r$ as a weighted average gives $r_{G_1}(1)\geq2/3>3/5$. Direct integration at $b=2$ gives
\[
A_2(2;G_1)=\log(4/3)-\frac16,
\qquad A_3(2;G_1)=\frac{5}{144}.
\]
The elementary inequality $\log x>2(x-1)/(x+1)$ for $x>1$ implies $\log(4/3)>2/7$. Consequently,
\[
r_{G_1}(2)
=\frac{5/72}{\log(4/3)-1/6}
<\frac{7}{12}<\frac35.
\]
Since $r_{G_1}$ is continuous and strictly decreasing, it crosses $3/5$ exactly once in $(1,2)$. The quasi-inverse therefore increases and then decreases. This direction of turning is essential: the uniform action distribution violates increasing monotonicity near the upper effort limit.

Both action laws satisfy Assumption \ref{assu: PT}. For $G_0$, put $V=\xi(B;G_0)$ and use the strategy $s=\xi(\cdot;G_0)^{-1}$ on the resulting compact type interval. The function $\xi(\cdot;G_0)$ is smooth, and its derivative is bounded above and bounded away from zero on $[1,2]$. The constructed type density is
\[
f(v)=\frac{g_0(s(v))}{\xi'(s(v);G_0)}.
\]
Since $g_0\geq1/10$, this density is bounded away from zero on the compact type support and is twice continuously differentiable on its interior. With independent draws of $B_i$, the resulting types are also i.i.d.; hence Assumption \ref{assu: auction DGP} holds. For each interior type $v$,
\[
U_b(b,v;G_0)=\{v-\xi(b;G_0)\}P'(b;G_0)
\]
is positive below $s(v)$ and negative above it. Thus $s(v)$ is a global best response over the entire feasible set $[1,2]$, rather than merely a local optimum. The endpoint types choose the respective endpoints by the same sign argument. Since $c(1)=0$ and all types are positive, each type can earn a positive payoff by choosing effort $1$, so a zero outside option does not prevent participation. Under $G_1$, the decreasing segment rules out a strictly increasing equilibrium rationalization. The only difference between the two constructions is the common action distribution.\hfill$\Box$

\medskip

\subsection{Fixed-payoff Public-good Construction and Proof}
Consider a two-player game with the entire feasible action set $\mathcal B=[0,1]$. Compare i.i.d.\ action profiles with common densities
\[
g_0(b)=\frac12+b,\qquad g_1(b)=\frac32-b,\qquad 0\leq b\leq1.
\]

For the public-good model of Section \ref{subsec: public}, fix linear contribution cost and the benefit function
\[
\Omega(q)=\frac{11}{4}q+\frac12q^2-\frac13q^3,\qquad q\in[0,2].
\]
Player $i$'s payoff is $v\Omega(b+s)-b$, where $s$ is the other player's contribution. This specification satisfies $\Omega(0)=0$ and $\Omega'(q)>0$ on the full feasible range; its curvature changes sign.

Within this game, the payoff functions and full feasible action set remain fixed across $G_0$ and $G_1$. The following verification establishes a strictly increasing equilibrium under $G_0$ and a violation of inverse monotonicity under $G_1$.

On $[0,1]$, both $g_0(b)=1/2+b$ and $g_1(b)=3/2-b$ integrate to one, are continuously differentiable, and are at least $1/2$. Direct integration gives
\[
\int_0^1bg_0(b)\,db=\frac7{12},\qquad
\int_0^1bg_1(b)\,db=\frac5{12}.
\]
In the example of a public good, $\Omega'(q)=11/4+q-q^2$ is a concave quadratic whose minimum on $[0,2]$ is $3/4$. Hence $P'(b;G)\geq3/4$ for either candidate law, and
\[
P''(b;G)=E_G[1-2(b+B_j)]=1-2b-2E_G[B_j].
\]
This equals $-1/6-2b$ under $G_0$ and $1/6-2b$ under $G_1$. Since $T'=1$, the first inverse is strictly increasing and the second decreases for $0<b<1/12$.

All expected payoff functions in this example are smooth on the whole feasible interval, and its quasi-inverse has continuous finite positive endpoint values. To verify nonemptiness of the equilibrium null directly, consider this example under $G_0$, set $V=\xi(B;G_0)$ for $B\sim G_0$, and let $F$ be the resulting type distribution. Strict increase of $\xi$ makes $s=\xi^{-1}$ a continuous strictly increasing strategy generating $G_0$. For every interior type $v$, the derivative of the normalized expected payoff satisfies
\[
U_b(b,v;G_0)=\{v-\xi(b;G_0)\}P'(b;G_0).
\]
It is positive for $b<s(v)$ and negative for $b>s(v)$. Thus $s(v)$ is the unique global maximizer over the entire feasible interval $[0,1]$. The corresponding one-sided signs give optimality at the endpoint types as well. Under $G_1$, a decreasing segment of the quasi-inverse violates the necessary condition in Theorem \ref{thm: LZ2018}. 

For the $G_0$ construction, $\xi'$ is continuous and bounded above and away from zero on $[0,1]$: its numerator is $1/6+2b$. Since $g_0\geq1/2$ and the denominators are smooth and positive, the induced density
\[
f(v)=\frac{g_0(\xi^{-1}(v;G_0))}{\xi'(\xi^{-1}(v;G_0);G_0)}
\]
is smooth in the interior and bounded above and away from zero on its compact type support. Thus the null construction also satisfies the private-type regularity in Assumption \ref{assu: auction DGP}.~~~~$\Box$

\subsection{Fixed-payoff Cournot Construction and Proof}
Consider a two-player game with the entire feasible action set $\mathcal B=[0,1]$. Compare i.i.d.\ action profiles with common densities
\[
g_0(b)=\frac12+b,\qquad g_1(b)=\frac32-b,\qquad 0\leq b\leq1.
\]

For the Cournot model of Section \ref{subsec: cournot}, fix
\[
\mathcal I(Q)=4-Q,\qquad Q\in[0,2],\qquad
c(b)=\frac74b-\frac12b^2,\qquad b\in[0,1].
\]
A firm with private efficiency $v>0$ earns $\mathcal I(b+s)b-c(b)/v$. Its capacity is one, so $[0,1]$ is the full feasible output set. Demand is positive and decreasing, revenue is concave, and the production cost is increasing with declining marginal cost. Multiplication by $v>0$ gives the normalized payoff $v\mathcal I(b+s)b-c(b)$.

Within this game, the payoff functions and full feasible action set remain fixed across $G_0$ and $G_1$. The following verification establishes a strictly increasing equilibrium under $G_0$ and a violation of inverse monotonicity under $G_1$.

On $[0,1]$, both $g_0(b)=1/2+b$ and $g_1(b)=3/2-b$ integrate to one, are continuously differentiable, and are at least $1/2$. Direct integration gives
\[
\int_0^1bg_0(b)\,db=\frac7{12},\qquad
\int_0^1bg_1(b)\,db=\frac5{12}.
\]
In the Cournot example, $c'(b)=7/4-b\geq3/4$ and
$P'(b;G)=4-2b-m\geq1$, where $m=E_G[B_j]$. The quotient rule gives
\[
\xi'(b;G)
=\frac{-(4-2b-m)+2(7/4-b)}{(4-2b-m)^2}
=\frac{m-1/2}{(4-2b-m)^2}.
\]
It is strictly positive under $G_0$ and strictly negative under $G_1$.

All expected payoff functions in this example are smooth on the whole feasible interval, and its quasi-inverse has continuous finite positive endpoint values. To verify nonemptiness of the equilibrium null directly, consider this example under $G_0$, set $V=\xi(B;G_0)$ for $B\sim G_0$, and let $F$ be the resulting type distribution. Strict increase of $\xi$ makes $s=\xi^{-1}$ a continuous strictly increasing strategy generating $G_0$. For every interior type $v$, the derivative of the normalized expected payoff satisfies
\[
U_b(b,v;G_0)=\{v-\xi(b;G_0)\}P'(b;G_0).
\]
It is positive for $b<s(v)$ and negative for $b>s(v)$. Thus $s(v)$ is the unique global maximizer over the entire feasible interval $[0,1]$. The corresponding one-sided signs give optimality at the endpoint types as well. In the Cournot model, multiplication of profit by $v>0$ preserves maximizers. Under $G_1$, a decreasing segment of the quasi-inverse violates the necessary condition in Theorem \ref{thm: LZ2018}. 

For the $G_0$ construction, $\xi'$ is continuous and bounded above and away from zero on $[0,1]$: its numerator is $1/12$. Since $g_0\geq1/2$ and the denominators are smooth and positive, the induced density
\[
f(v)=\frac{g_0(\xi^{-1}(v;G_0))}{\xi'(\xi^{-1}(v;G_0);G_0)}
\]
is smooth in the interior and bounded above and away from zero on its compact type support. Thus the null construction also satisfies the private-type regularity in Assumption \ref{assu: auction DGP}.~~~~$\Box$

\section{Additional Simulation Results}

We consider the case with heterogeneous number of bidders and for implementation of such test; see Appendix \ref{sec: hetero numbers} for details. To study the size and power properties, we use the same quantile function of the true bid distribution defined in (\ref{eq:quantile}). We set $k=0.5$ for the size analysis and set $10$, $20$, and $40$ for the power analysis. For all values of $k$, we set $N_t=[2,3,4]$ and $L_t=a[60,40,20]$ with $a=1,2,3$. Figure \ref{fig:bidder_value_DGP3} presents the quasi-inverse equilibrium strategy $\xi(b)$ for different values of $k$ and $N$.

\begin{figure}[t]
\centering
\hspace*{-1cm}
\includegraphics[width=15cm,height=10cm]{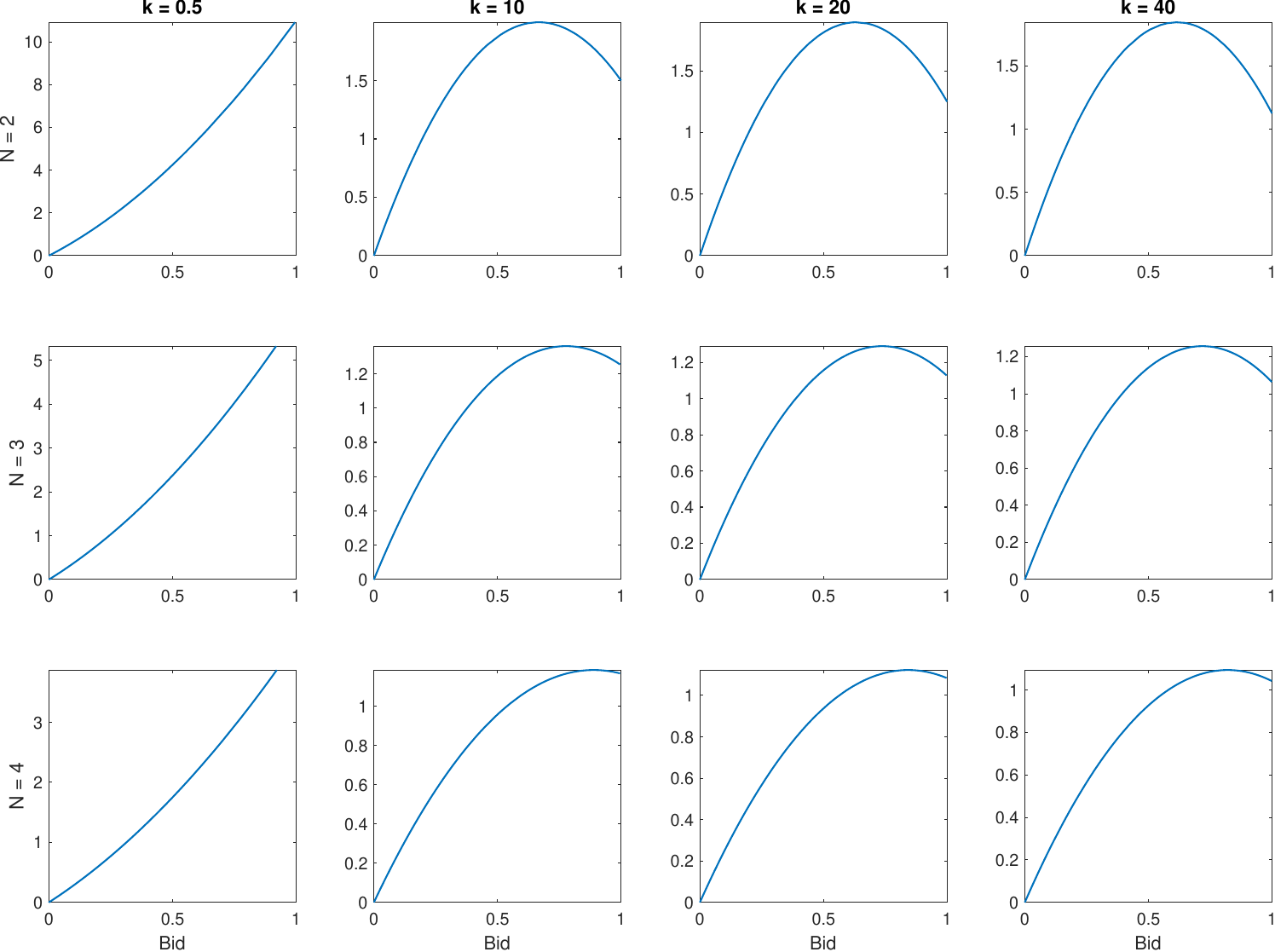}
\caption{Quasi-inverse equilibrium strategy for $N=2$, $3$, or $4$}\label{fig:bidder_value_DGP3}
\end{figure}

Table \ref{table:DGP3} shows the rejection probabilities of our tests for the case with heterogeneous number of bidders. Like the results in Table \ref{table:DGP1}, the proposed test controls size well for $k=0.5$, and the power increases with the sample size and with $k$ for $k\geq 10$. Table \ref{table:DGP3} also shows that the choices of $q_1$ do not affect the test performance much.

\begin{table}[h]
\begin{small}
\begin{center}
\caption{Rejection probabilities for the case without covaraites and with heterogeneous number of bidders}\label{table:DGP3}
\begin{tabular}{cccccc}
\toprule
k&a&$n_c$=15&$n_c$=20&$n_c$=25&$n_c$=30\\
\midrule
0.5&1&0.000&0.000&0.000&0.000\\
0.5&2&0.000&0.000&0.000&0.000\\
0.5&3&0.000&0.000&0.000&0.000\\
\midrule
10&1&0.195&0.204&0.175&0.166\\
10&2&0.209&0.227&0.225&0.229\\
10&3&0.298&0.296&0.275&0.298\\
\midrule
20&1&0.449&0.413&0.442&0.369\\
20&2&0.669&0.640&0.629&0.618\\
20&3&0.800&0.800&0.796&0.805\\
\midrule
40&1&0.723&0.674&0.659&0.616\\
40&2&0.922&0.924&0.920&0.893\\
40&3&0.989&0.988&0.987&0.982\\
\bottomrule
\end{tabular}
\end{center}
\end{small}
\end{table}

\end{document}